%% file: transition_arXiv.tex
\documentclass[aps,superscriptaddress,onecolumn,10pt,prx]{revtex4-2}

\usepackage{amsmath,mathtools,amsthm,amssymb}
\usepackage{tabularx}
\usepackage{tabularray}

\usepackage{graphicx}

\usepackage{color}
\usepackage[dvipsnames]{xcolor}
\usepackage{bbold}
\usepackage{enumitem}
\usepackage{physics}
\usepackage{comment}
\usepackage{array}
\usepackage{graphics}
\usepackage{wrapfig}
\graphicspath{{figures/}}
\usepackage{biolinum}
\usepackage{bbm}
\usepackage{dsfont}
\usepackage{mathrsfs}
\usepackage{mathdots}
\usepackage{enumitem}
\usepackage{pifont}
\definecolor{myrefcolor}{rgb}{0.067,0.5,0.5}
\definecolor{myurlcolor}{rgb}{0.1,0,0.9}

\usepackage[
    breaklinks,
    pdftex,
    colorlinks=true,
    linkcolor=myrefcolor,
    citecolor=myrefcolor,
    urlcolor=myrefcolor
]{hyperref}

\usepackage[capitalize]{cleveref}
\DeclareMathOperator{\poly}{poly}
\newcommand{\pur}[0]{\mathrm{Pur}}

\newtheorem*{theorem*}{Theorem}
\newcounter{thm}
\newtheorem{theorem}[thm]{Theorem}
\newtheorem{lemma}[thm]{Lemma}

\newtheorem{definition}{Definition}
\newtheorem{example}{Example}
\newtheorem{corollary}{Corollary}

\theoremstyle{remark}
\newtheorem{remark}{Remark}

\usepackage{algorithm}
\usepackage{algpseudocode}

\newcommand{\id}{\mathds{1}}
\newcommand{\sep}{\operatorname{SEP}}
\DeclareMathOperator{\supp}{supp}
\DeclareMathOperator{\rowspan}{rowsp}
\DeclareMathOperator{\otoc}{OTOC}
\newcommand{\sepp}[2]{\operatorname{SEP}(\mathscr{H}_{#1}; \mathscr{H}_{#2})}

\newcommand{\parhead}[1]{\noindent \textbf{\textsf{#1}}}

\newcommand{\cnot}{\mathrm{CNOT}}

\newcommand{\be}{\begin{equation}\begin{aligned}\hspace{0pt}}
\newcommand{\ee}{\end{aligned}\end{equation}}
\newcommand{\ba}{\begin{eqnarray}}
\newcommand{\ea}{\end{eqnarray}}

\newcommand{\haar}[0]{\operatorname{Haar}}

\definecolor{airforceblue}{rgb}{0.36, 0.54, 0.66}
\newcommand{\bb}{\begin{equation}\begin{aligned}\hspace{0pt}}
\newcommand{\bbb}{\begin{equation*}\begin{aligned}}
\newcommand{\eb}{\end{aligned}\end{equation}}
\newcommand{\eeb}{\end{aligned}\end{equation*}}
\allowdisplaybreaks
\begin{document}

\title{Magic-induced computational separation in entanglement theory}

\author{Andi Gu}
\affiliation{Department of Physics, Harvard University, 17 Oxford Street, Cambridge, MA 02138, USA}
\author{Salvatore F.E. Oliviero}
\affiliation{NEST, Scuola Normale Superiore and Istituto Nanoscienze, Consiglio Nazionale delle Ricerche, Piazza dei Cavalieri 7, IT-56126 Pisa, Italy}
\author{Lorenzo Leone}
\affiliation{Dahlem Center for Complex Quantum Systems, Freie Universit\"at Berlin, 14195 Berlin, Germany}

\begin{abstract}
\noindent 
Entanglement serves as a foundational pillar in quantum information theory, delineating the boundary between what is classical and what is quantum. The common assumption is that higher entanglement corresponds to a greater degree of `quantumness'. However, this folk belief is challenged by the fact that classically simulable operations, such as Clifford circuits, can create highly entangled states. The simulability of these states raises a question: what are the differences between `low-magic' entanglement, and `high-magic' entanglement? We answer this question in this work with a rigorous investigation into the role of magic in entanglement theory. We take an operational approach to understanding this relationship by studying tasks such as entanglement estimation, distillation and dilution. This approach reveals that magic has notable implications for entanglement. Specifically, we find an operational separation that divides Hilbert space into two distinct regimes: the entanglement-dominated (ED) phase and magic-dominated (MD) phase. Roughly speaking, ED states have entanglement that significantly surpasses their magic, while MD states have magic that dominates their entanglement. The competition between the two resources in these two phases induces a computational phase separation between them: there are {sample- and time-efficient} quantum algorithms for almost any entanglement task on ED states, while these tasks are {provably computationally intractable} in the MD phase. Our results find applications in diverse areas such as quantum error correction, many-body physics, and the study of quantum chaos, providing a unifying framework for understanding the behavior of quantum systems. We also offer theoretical explanations for previous numerical observations, highlighting the broad implications of the ED-MD distinction across various subfields of physics.
\end{abstract}
\maketitle

\tableofcontents
\clearpage
\section{Introduction}
Entanglement has long stood as a central concept in quantum information, serving as a key differentiating feature between classical and quantum theories. This perspective was both enriched and challenged by the introduction of the stabilizer formalism, which identified a subset of quantum states, known as stabilizer states (and their associated circuits, termed Clifford circuits), which could be simulated efficiently using classical resources~\cite{gottesman_heisenberg_1998}. These classically simulable states could be highly entangled, and hence did not conform to the conventional understanding that entanglement could be the sole metric determining the degree to which a state is truly quantum. Building on this understanding, a new metric, referred to as `magic', or nonstabilizerness, was introduced to quantify the amount of non-Clifford resources necessary to prepare a given quantum state, providing a more nuanced measure of a state's quantumness. Indeed, magic has itself been shown to be intimately connected with a number of fundamentally quantum phenomena \cite{campbell_catalysis_2011,bravyi_trading_2016,beverland_lower_2020, leone_stabilizer_2022,leone_nonstabilizerness_2023,goto_probing_2022,garcia_resource_2023}.

However, there is little work on the interaction of these two measures of quantumness. A previous work showed that entanglement could be calculated exactly for stabilizer states \cite{fattal_entanglement_2004}, and more recent works have shown that the degree of magic in a state is linked to the `entanglement response' of the state \cite{tirrito_quantifying_2023}. In this work, we further explore the intersection of these two notions and find them deeply connected. This exploration leads to implications that extend well beyond the realm of quantum information theory. Our framework provides valuable insights into a wide array of physical phenomena, from the robustness of topological entanglement entropy in quantum error correcting codes to the behavior of entanglement in many-body systems. Moreover, our efficient entanglement estimation protocols find applications in the study of scrambling dynamics and the characterization of quantum chaos. By bridging the gap between these seemingly disparate fields, our work highlights the fundamental role of the entanglement-magic interplay in understanding the complex behavior of quantum systems.

There are many ways to understand entanglement. The approach we take in this work is an \emph{operational} one: we argue that studying the role of magic in entanglement estimation and manipulation tasks leads to a clear understanding of their relationship as a whole. However, a simple operational approach does not yet reveal the full picture. A sequence of recent works~\cite{aaronson_quantum_2023,arnon-friedman_computational_2023} have found that operational characterizations of entanglement can change significantly under the constraint of computational limitations. Taking inspiration from this, we augment our operational approach by studying entanglement tasks under a \emph{computational lens}. This approach allows us to uncover a surprising divide between quantum states, defined by the relationship between magic and entanglement. Specifically, we find that Hilbert can be divided into two distinct phases: the entanglement-dominated (ED) phase and magic-dominated (MD) phase. These phases correspond roughly to states whose entanglement significantly surpasses injected magic, and states where magic dominates entanglement. The boundary between the two phases is demarcated by a computational separation induced by the competition between the two resources. That is, for a number of entanglement detection and manipulation problems, there are {sample- and time-efficient} quantum algorithms that solve these problems for ED states. Conversely, we show that these entanglement manipulation tasks are {provably computationally intractable} for the MD phase.

Our first step in understanding the transition ED-MD is in the setting of entanglement estimability. We find that, in the ED phase, we can efficiently estimate entanglement entropy with an asymptotically (in $n$) vanishing error, {even for volume-law states}. Conversely, we show that in the MD phase, entanglement estimation is inefficient beyond logarithmic entanglement (below this, the swap test allows for efficient estimation). Finding clear evidence of an ED-MD computational separation for entanglement estimation, we then study entanglement distillation. We show that within the class of ED states, we can always efficiently find a polynomial-depth circuit that distills almost all of the entanglement into Bell pairs. We also prove the converse, which says that in the MD phase, it is impossible to find such efficient and optimal distillation protocols. Similarly, in the context of entanglement dilution, we demonstrate that we can always identify an efficient dilution protocol that utilizes an optimal number of Bell pairs to prepare ED states. Conversely, for MD states, we rule out the possibility of an efficient dilution protocol that consumes anything close to an optimal number of Bell pairs. Collectively, these findings provide substantial evidence supporting a computational separation between ED and MD states. This leads us to conclude that there is a stark difference in entanglement structure between these two phases of states (as illustrated by the ED-MD phase diagram in \cref{fig:schematic}). Within the ED phase, entanglement is always structured in such a way that it can be manipulated efficiently and (almost) reversibly. In contrast, for general states in the MD phase, there is almost no structure in the entanglement, making entanglement manipulation inefficient and irreversible.

Given the simple entanglement estimation and manipulation protocols for ED states, one might wonder whether entanglement-dominated states are merely an esoteric theoretical construction. Practically speaking, where can we expect to find ED states? Although the majority of states within the Hilbert space are magic-dominated, we show that ED states appear naturally in a wide range of physical settings, from quantum error correcting codes to many-body systems. Moreover, it is well-known that most states in Hilbert space also cannot be prepared in polynomial time, hence are irrelevant in practice~\cite{Poulin_2011,knill1995approximation,nielsen_quantum_2000}. In contrast, there is a large class of well-known unitary operations which almost always produce ED states, namely circuits dominated by Clifford gates. In particular, Clifford-dominated circuits with $t=o(n)$ non-Clifford gates (importantly, this far exceeds the classical simulability limit $O(\log n)$) produces ED states with overwhelming probability. This means that, in contrast with $t$-doped state simulation algorithms which can only operate efficiently under the assumption that $t$ scales logarithmically in $n$ \cite{aaronson_improved_2004,bravyi_simulation_2019}, our efficiency results regarding ED states hold far beyond this classical simulation limit. However, we emphasize that Clifford-dominated circuits only generate a tiny fraction of all ED states; in fact, most ED states are constructed from far more than just $O(n)$ $T$-gates.

{As applications of our theory, we demonstrate the robustness of entanglement witnesses and topological entanglement entropy for ED states, highlighting their potential for practical quantum error correction and the study of many-body physics. We also develop an efficient testing algorithm which can {classify states} within the ED and MD phases. To conclude, we highlight the relevance of our findings in {many-body physics} by demonstrating the {robustness of topological entanglement entropy} in 3D topological models such as the X-cube model and Haah's code. We also find connections between stabilizer code Hamiltonians and ED states, which extends the applicability of our results to quantum error correcting codes. Lastly, we provide rigorous grounds for previous numerical observations in the literature of Clifford-dominated circuits, such as entanglement-cooling and phase transition in hybrid quantum circuits.

\begin{figure}
    \centering
    \includegraphics[width=0.45\textwidth]{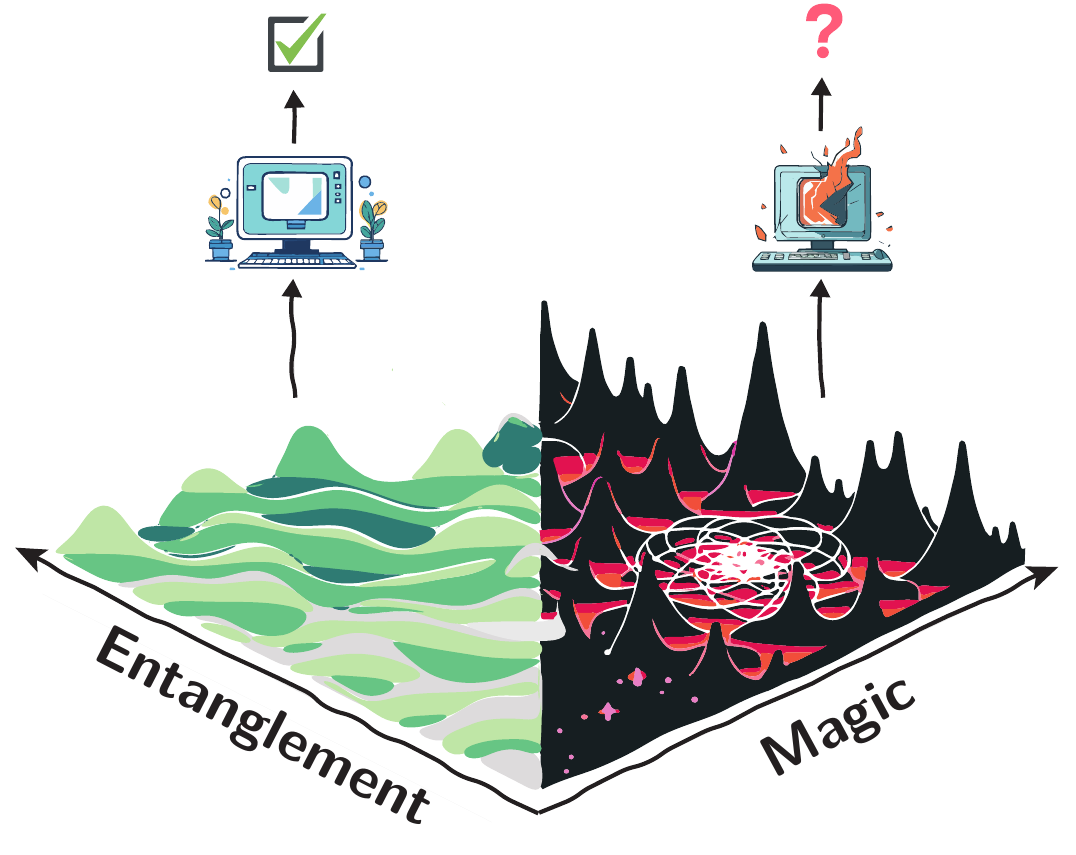}
    \caption{The sharp distinction between entanglement-dominated and magic-dominated states. The landscape visualizes the entanglement structure of states: in the entanglement-dominated phase, entanglement is highly structured and easily manipulable, while in the magic-dominated phase, entanglement can be scrambled in such a complex way that it is completely intractable to measure or manipulate.}
    \label{fig:schematic}
\end{figure}

\subsection{Overview}
This section provides a non-technical summary of our key findings. As highlighted in the introduction, this paper explores the interplay between two fundamental resources that separate quantum computers from classical computers: entanglement and magic. A starting point for our work is the observation that tasks related to entanglement, such as entanglement measurement, witnessing, or manipulation (e.g., EPR-pair extraction, dilution), are sample- and time-efficient for stabilizer states. We investigate how introducing \emph{magic}, through non-Clifford gates in state preparation, affects the entanglement manipulability (hence entanglement structure) of stabilizer states. %

We repeatedly find two classes of states where entanglement-related tasks are either efficiently solvable or provably hard. This stark operational difference splits the Hilbert space into the \textit{entanglement-dominated} phase and \textit{magic-dominated} phase. As previously mentioned, these phases roughly correspond to scenarios where entanglement significantly surpasses injected magic, leading to stabilizer properties dominating, and cases where magic dominates entanglement. For a more precise definition of the two phases, we use resource theoretic monotones --- a natural choice, given the operational approach we take. We use the \textit{stabilizer nullity}~\cite{beverland_lower_2020} as a measure of magic, which has a clear operational meaning in terms of magic-state distillation. On the other hand, we quantify bipartite entanglement in $A|B$ with the \textit{entanglement entropy} $S_{1}(\psi_A)$, which has a clear operational meaning for entanglement distillation and dilution~\cite{horodecki_quantum_2009}. 
\begin{definition}[Entanglement and magic-dominated phases]\label{def:ent-dom}
Let $\ket{\psi}$ be a state with $2^{n-\nu}$ Pauli stabilizers; we say that this state has stabilizer nullity $\nu$~\cite{beverland_lower_2020}. Let $A|B$ an extensive bipartition and let $S(\psi_A)$ be the von Neumann entropy of $\psi_A\coloneqq \Tr_B \ketbra{\psi}{\psi}$. We say $\ket{\psi}$ is entanglement-dominated if $S(\psi_A) = \omega(\nu)$, and it is magic-dominated if $S(\psi_A) = O(\nu)$. We note that this definition depends on the bipartition $A|B$. This choice of bipartition will typically be implicit; whenever it is not clear from the context, we will specify which bipartition we are defining entanglement-dominated or magic-dominated with respect to. Furthermore, we generally assume $n_A,n_B=\omega(1)$, otherwise all states would be trivially magic-dominated.
\end{definition}

We will study the measurability and manipulability of entanglement for these two classes of states. For entanglement manipulation, there are two tasks: entanglement distillation and dilution. The goal of entanglement distillation is to use LOCC (local operations and classical communication) to transform a state $\psi$ into as many Bell pairs $M_+$ as possible. This number is the \emph{distillable entanglement} of $\psi$. For the reverse of entanglement distillation, namely entanglement dilution, the number of Bell pairs $M_{-}$ required to prepare $\psi$ via LOCC is known as its \textit{entanglement cost}. Having outlined the relevant entanglement tasks, we use the table below to summarize the key findings of this work. Throughout the manuscript, we often make use of Bachmann-Landau notation for asymptotics~\footnote{To remind the reader, $f=o(g(n)) \implies \lim_{n \to \infty} {f(n)}/{g(n)}=0$, $f=\omega(g(n)) \implies \lim {f(n)}/{g(n)}=\infty$, $f=O(g(n))\implies \exists C_2>0$ such that $ \lim {f(n)}/{g(n)}\le C_2$, $f=\Omega(g(n))\implies \exists C_1>0$ such that $ \lim {f(n)}/{g(n)}\ge C_1$  and finally $f=\Theta(g(n))$ means $f=\Omega(g(n))$ and $f=O(g(n))$.}.

\begin{table}[H]
    \centering
      \vspace{1em}
    \begin{tblr}{
  colspec = {X[h]||c|c},}
  \centering \textbf{Protocol} & \textbf{Entanglement-dominated} & \textbf{Magic-dominated} \\
         \hline \hline
         An efficient state-agnostic protocol which produces an estimate $\tilde{S}(\psi_A)$ of the true entanglement $S(\psi_A)$. & \large $\frac{\abs{S(\psi_A)-\tilde{S}(\psi_A)}}{S(\psi_A)} = o(1)$ & \large $\frac{\abs{S(\psi_A)-\tilde{S}(\psi_A)}}{S(\psi_A)}=\Omega(1)$ \\ \hline
         An efficient state-agnostic protocol which distills $M_{+}$ Bell pairs from the state $\psi$. & \large $M_+/S(\psi_A) = 1-o(1)$ & \large $M_+/S(\psi_A) = o(1)$ \\ \hline
         An efficient state-agnostic LOCC protocol which uses $M_-$ Bell pairs to prepare a state $\psi$ across $A|B$, where either $A$ or $B$ has local access to polynomial copies of $\psi$. & \large $M_-/S(\psi_A) = 1+o(1)$ & \large $M_-/S(\psi_A) = \omega(1)$
  \end{tblr}
    \label{tab:transition}
\end{table}

\begin{figure}[H]
    \centering
    {
    \def\svgwidth{0.6\textwidth}
    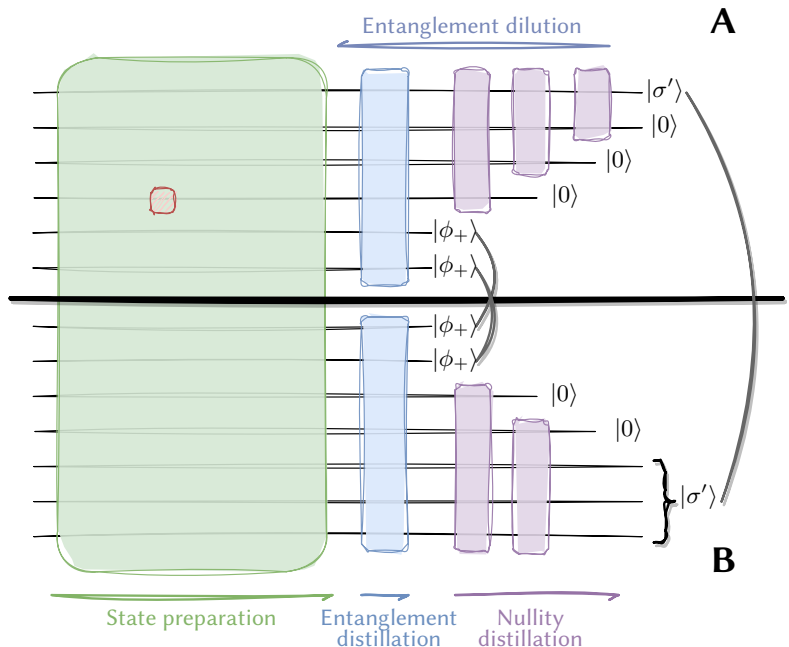
    }
\caption{Schematic of our entanglement distillation protocol (left to right) and dilution protocol (right to left).}\label{fig:distill}
\end{figure}

\begin{theorem*}[Efficient entanglement measurement and manipulation for ED states~{[\cref{thm:stab-approx,eq:detectability,thm:distill,lem:phasetransdistill,thm:dilution,lem:phasetransdilu}]}] Let $\ket{\psi}$ be an unknown entanglement-dominated state across the bipartition $A|B$ with $n_A,n_B=\omega(\log n)$. There exists sample- and time-efficient quantum algorithms for each of the following tasks.
\setlist{nolistsep}
\begin{enumerate}[noitemsep]
    \item \emph{Entanglement estimation:} Estimates $S(\psi_A)$ up to an $o(1)$ relative error.
    \item \emph{Entanglement distillation:} Distills $M_+$ Bell pairs from $\psi$ using LOCC operations, with $M_+/S(\psi_A) = 1-o(1)$. Unlike more conventional protocols~\cite{bennett1996concentrating}, this protocol requires only a single copy of the input state and makes no error: it is a \emph{one-shot, zero error protocol}.
    \item \emph{Entanglement dilution:} Prepares $\psi$ across the bipartition $A|B$ using LOCC, $M_-$ Bell pairs, where $M_-/S(\psi_A) = 1+o(1)$ and a number of bits of classical communication $N_{\mathrm{CC}}= o(S(\psi_A))$.
\end{enumerate}
\end{theorem*}
\noindent \big[\![\emph{Sketch of the proof.}  
The only prior knowledge required by each of these protocols is the stabilizer group $G$ of the state $\psi$ (which can be learned efficiently~\cite{grewal_efficient_2023}).
For the entanglement measurement task, generalizing the result of~\cite{fattal_entanglement_2004}, a purely algebraic analysis reveals that the $\alpha$-R\'enyi entropies for states with stabilizer nullity $\nu$ are also equal to $\frac{1}{2} \log_2\abs{G_{AB}}$, up to an \emph{additive error $O(\nu)$} (hence a relative error $o(1)$ in the ED phase), where $G_{AB}$ is the group of stabilizers for $\psi$ which act nontrivially on both $A$ and $B$. This tells us something important for the entanglement distillation task: for ED states, \emph{most of the entanglement is encoded by the nonlocal stabilizers $G_{AB}$ of the state}. Therefore, simply by using local Clifford unitaries to transform these stabilizers to $XX$ and $ZZ$ pairs across the bipartition, we distill the stabilizer portion of entanglement (which is, for entanglement-dominated states, asymptotically equal to \emph{all} the entanglement) into Bell pairs. These unitaries are shown in blue in \cref{fig:distill}. Finally, our dilution protocol effectively works backward from this. After our distillation protocol, we further distill the remaining local stabilizers $G_A$ and $G_B$ into single-qubit $Z$s, resulting in a state on the rightmost side of \cref{fig:distill}, which contains a bipartite entangled state $\ket{\sigma'}$ that holds the (small) nonstabilizer portion of the entanglement. To prepare the original state $\ket{\psi}$, it suffices to prepare this state $\ket{\sigma'}$ across the bipartition, then run the protocol (comprised of purple and blue unitaries) in reverse; we prepare $\ket{\sigma'}$ with a simple state teleportation protocol, which is asymptotically free because $\ket{\sigma'}$ lives on $O(\nu)\leq o(S(\psi_A))$ qubits.]\!\big]

\medskip

The previous theorem demonstrates that entanglement-related tasks, such as entanglement estimation, distillation, and dilution, can be performed efficiently for states in the entanglement-dominated (ED) phase. These positive results highlight the tractability of manipulating and characterizing entanglement in ED states, even for those with high entanglement entropy (i.e., volume-law states). The efficiency of these tasks in the ED phase can be attributed to the structure of entanglement in these states, which allows for effective compression and manipulation using stabilizer operations. In stark contrast, the following theorem establishes the hardness of these same entanglement-related tasks for states in the magic-dominated (MD) phase. It shows that, in general, estimating entanglement entropy or performing entanglement distillation and dilution for MD states cannot be done efficiently, even when allowing for a constant relative error. This computational intractability arises from the complex and unstructured nature of entanglement for general MD states.

\begin{theorem*}[Hardness of entanglement estimation and manipulation for MD states~{[\cref{thm:no-go-magic-dom}]}] 
Take any bipartition $A|B$ that satisfies $n_A,n_B=\omega(\log n)$ and let $\ket{\psi}$ be a pure state with $S(\psi_A)=\omega(\log n)$. There is no efficient protocol which can estimate $S(\psi_A)$ to within a constant relative error for arbitrary MD states. Furthermore, the distillable entanglement of general MD states is $M_+/S(\psi_A) = o(1)$ while the entanglement cost is $M_-/S(\psi_A) = \omega(1)$.
\end{theorem*}
\noindent \big[\![\emph{Sketch of the proof.} We encode pseudoentangled states~\cite{aaronson_quantum_2023} into MD states. An important conceptual part of this proof is that this encoding is only possible for MD states; if our encoding adds too much entanglement, the pseudoentanglement of the original states is washed out. We then show that a protocol for entanglement estimation with constant relative error could be used to distinguish between these encoded pseudoentangled states and Haar random states, leading to a contradiction. Similarly, a good protocol for either distillation or dilution, in the sense that $M_+,M_- = \Theta(S(\psi_A))$, could be used as a distinguisher, leading to another contradiction. Remarkably, this is the first no-go result for dilution using pseudoentangled quantum states.]\!\big] 

\medskip

We qualify this claim for readers not familiar with hardness proofs. In (almost) every hardness theorem, there are two crucial points to be aware of. The first part of any hardness theorem is the class of problems for which the problem is claimed to be hard. In the theorem above, the only information we are given is that the state is magic-dominated, which means we are forced to consider very general algorithms that work for the broad class of MD states. With more information at hand, in principle, there could be alternative strategies tailored for specific classes of MD states. The theorem merely states that there is no efficient strategy for entanglement-related tasks that works in full generality on \emph{all} MD states (in stark contrast to our results about ED states). This brings us to the second point. Hardness proofs refer to \textit{worst-case hardness}, meaning that among MD states, we guarantee that there is at least one (sub)class of states for which no efficient strategy can perform well. Indeed, the proof of this is constructive: we concretely identify subclasses of states for which this is true (see \cref{sec:ed-vs-md}). In principle, there are probably many more classes spread around the larger set of MD states for which there are no such efficient strategies. This is best understood with a canonical example: it is well-known that finding the ground state of a general classical system is NP-hard. However, does this mean that finding the ground state of \emph{every} classical Hamiltonian is hard? Of course not. However, there are important classical physical systems, such as spin glasses, for which finding the ground state is provably hard. This constitutes a class of systems for which the hardness proof holds true. Given the two axes representing magic and entanglement respectively, we show that the entanglement-related tasks analyzed above are provably efficient and agnostic for \textit{every} ED state. On the other hand, for a wide range of magic, we construct a family of MD states for which entanglement-related tasks are hard. Given that, what is the true meaning of the computational separation found in this work? It means that, in the ED phase, we are guaranteed the existence of efficient and agnostic protocols, while in the MD phase, efficient and agnostic protocols do not exist; hence, tailored protocols for specific classes should be adopted. However, at the same time, there exist classes of states within the MD phase for which no efficient protocols exist at all.

\medskip
\parhead{Multipartite entanglement distillation.} Besides bipartite entanglement distillation, there's also the challenge of multipartite entanglement distillation. In this scenario, if $k$ parties share an entangled state, their goal is to distill some target $k$-partite entangled state (e.g., a GHZ state) using LOCC operations. Interestingly, it has been established that this task is unachievable for the vast majority of states~\cite{PhysRevLett.111.110502,PhysRevX.8.031020}. However, identifying a generalization of the ED phase in $k$-partite setting, in \cref{thm:multipartite} we show that we can \emph{deterministically} distill many copies of a $k$-partite GHZ state from ED states using an efficient LOCC protocol.

\medskip

\parhead{Entanglement witnessing and robustness.} While precisely measuring entanglement can be a challenging and noise-sensitive task, the less ambitious goal of merely witnessing entanglement can be easier and more noise-resilient~\cite{terhal_entanglement_2000,guhne2009entanglement}. We will address this task in~\cref{subsec:witnessingent}. The purpose of a witness is to experimentally validate the presence of genuine entanglement in an imperfectly prepared version of the target state. We define an entanglement witness for ED states that can be measured with $O(1)$ sample complexity. With~\cref{th:multipartiteEentwitness}, we strengthen this result by defining a similar witness for multipartite entanglement. This witness verifies entanglement across $k$ parties --- that is, it rules out the possibility of the state being unentangled across any of the $k$ given partitions. As a corollary of this, we find that entanglement for ED states can be far more robust than the entanglement of generic states, see \cref{cor:robustnessmultipartite}. The reason for this is as follows. The Fannes inequality roughly says that $\abs{S(\rho_A)-S(\psi_A)} \lesssim n_A T,$ where $T$ is the trace distance between $\psi$ and its noisy version $\rho$. If $\psi$ were a generic state with $S(\psi_A) \sim \sqrt{n_A}$, then we would generally need $T < 1/\sqrt{n_A}$ to guarantee that $\rho$ were not separable across $A|B$. On the other hand, if $\psi$ were in the ED phase, we show that we could tolerate up to $T \lesssim 1-2^{-\Omega(\sqrt{n_A})}$, showing that the entanglement within the ED phase is extremely robust.
\medskip

\parhead{Phase classification and testing.} In light of the clear distinction between ED-MD phases, one might ask whether it is possible, given query access to an unknown state $\ket{\psi}$, to determine the phase in which it resides. We formalize this task as a property testing problem and show that the separation between ED-MD phases can indeed be efficiently tested. More precisely, in \cref{th:distinction}, we present a polynomial-time algorithm that can discriminate whether $\ket{\psi}$ is an ED state or it is $\epsilon$-far from any state in the entanglement-dominated phase and, as such, lies in the MD phase.

\medskip

\parhead{Applications to physics.} We conclude by discussing the implications of our results in the context of many-body physics. First, leveraging on the result of Ref.~\cite{gu_hamiltonian_2024}, we stem a potential application of the techniques developed in our work. In~\cref{sec:manybodyphysics}, we show how ED states appear in the context of many-body physics. We demonstrate the robustness of ground state entanglement for such models and study the robustness of topological entanglement entropy to perturbations in models such as the X-cube model or Haah's code. Taking advantage of the fact that the topological entanglement entropy in these models scales extensively, we show that this topological entanglement persists under any perturbation to the Hamiltonian that has a subextensive (i.e., $o(n)$) number of terms, \textit{regardless the strength of the perturbation}. Additionally, in~\cref{sec:lyapunov}, we present a method for quantifying the Lyapunov exponent – the rate at which correlations propagate in unitary dynamics – within random quantum circuits. This algorithm is based on the efficiency of estimating entanglement in the ED phase. Subsequently, leveraging on the theoretical tools developed in this work, in~\cref{sec:phenomenology}, we analytically and rigorously elucidate certain numerical observations in the literature. Specifically, we offer a straightforward proof demonstrating that entanglement in states constructed from Clifford gates, supplemented by, at most, a linear number of non-Clifford gates, can be efficiently and agnostically cooled down. This result was previously observed numerically in Ref.~\cite{true_transitions_2022}. Moreover, we delve into the study of the magic-entanglement transition in hybrid quantum circuits, complementing the analysis from Ref.~\cite{fux_entanglementmagic_2023}. In this context, we rigorously establish the absence of ED states in these circuits, confirming a numerical observation of Ref.~\cite{fux_entanglementmagic_2023}.

\medskip

\parhead{Pseudocodes.} Finally, we provide pseudocode for many of the protocols we have just described. First, we show how to \emph{exactly} compute the $2$-R\'enyi entanglement entropy of a state with stabilizer nullity $\nu$ using a classical algorithm whose runtime is $O(4^\nu n^3)$. We then show how to efficiently monitor the stabilizer entanglement produced by a $t$-doped Clifford circuit using an efficient classical algorithm, and finally show how to efficiently witness entanglement for ED states. These procedures can be of independent interest, so we dedicate a separate section for these pseudocodes in \cref{sec:Pseudocodesfornumerics}.

\subsection{Related works}
This work aims to contribute to the growing body of research within the quantum resource theory that investigates how quantum resources interact with each other. While our work presents a detailed exploration of the relationship between magic and entanglement, similar ideas have been previously explored to some extent, providing a foundation for our investigation. This section discusses and contextualizes our results in light of these earlier contributions.

Recently, there has been a growing interest in understanding how the efficiency of quantum algorithms impacts the feasibility of quantum resource protocols, particularly in the context of entanglement~\cite{aaronson_quantum_2023} and magic~\cite{gu2023little}. To understand these developments, we refer to a seminal work~\cite{bennett1996concentrating}, which showed that the number of Bell pairs that can be distilled per copy of bipartite state $\psi_{AB}$ is asymptotically equal to the von Neumann entanglement $S_1(\psi_A)$ of the state. Furthermore, it was shown that this rate is optimal. Despite the apparent finality of this result, there are several drawbacks associated with the protocol used to accomplish this. First, since distillation techniques make use of the idea of typical sets, it requires \emph{many} copies of the input state -- for even simple cases, up to $10^6$ copies may be required before the desired concentration begins to kick in~\cite{wilde_quantum_2013}. Second, it relies on knowledge of the input state (i.e., it is not input state agnostic). Finally, there are no guarantees on the computational efficiency of the protocol. In fact, the recently introduced concept of quantum pseudoentangled demonstrates rigorously that this protocol cannot be efficient in general~\cite{aaronson_quantum_2023,gu2023little}, a problem elaborated by the framework of \emph{computational entanglement theory}~\cite{arnon-friedman_computational_2023}. Our work presents an entanglement distillation protocol that circumvents all three of these problems for arbitrary entanglement-dominated states: it works just as well on one copy of the state as on many copies, it is agnostic to the input state, and is computationally efficient. The performance tradeoff for these gains is minimal: we are able to distill almost all the entanglement out of entanglement-dominated states. 
Furthermore, we show in a converse result that this is the best one can hope for. No efficient, input state-agnostic and optimal protocol can exist for states that are magic-dominated. Additionally, we generalize these findings to entanglement dilution, proving a magic-induced computational separation in entanglement cost. We offer the first no-go results for entanglement dilution using pseudorandom states. While a computational separation in entanglement cost has been found in Ref.~\cite{arnon-friedman_computational_2023}, this work only ruled out a one-way LOCC dilution protocol. In contrast, our results hold for general LOCC channels, with an arbitrary number of communication rounds.

Beyond the context of quantum resource theories, our results also have strong connections with the quantum information literature on $t$-doped stabilizer states. This class of states is of broad interest because they are widely believed to be the simplest class of states that cannot be classically simulated (for $t=\omega(\log n)$), but can also be prepared on early fault-tolerant quantum hardware. We outline the essential features of these $t$-doped stabilizer states that have been found already in past works. We will also describe how our work relates to these previous results and present directions for future investigation. Unsurprisingly, the first place where $t$-doped stabilizer states primarily emerged was in the context of classical simulation of quantum circuits. The Gottesman-Knill theorem~\cite{gottesman_heisenberg_1998} showed that Clifford circuits can be efficiently represented with $O(n^2)$ bits and classically simulated with runtime $O(n^3)$. The introduction of $t$ non-Clifford gates makes the simulation exponentially harder in $t$. This fact was first explored in Ref.~\cite{aaronson_improved_2004}, and then improved several times, e.g. Refs.~\cite{bravyi_improved_2016,bravyi_simulation_2019}. Despite improvements, the most advanced classical algorithm for simulating $t$-doped Clifford circuits still exhibits exponential scaling in $t$, which has been shown to be optimal~\cite{huang2019explicit}. Therefore, circuits (and the states they generate) which use $t=O(\log n)$ non-stabilizer gates can be efficiently classically simulated. However, once $t=\omega(\log n)$, any classical simulation algorithm requires a super-polynomial amount of resources.

Doped Clifford circuits have been proven to be extremely useful for quantum information processing tasks. In particular, they appear in the context of unitary $k$-designs, which are ensembles of unitaries that reproduce the first $k$ moments of the uniform measure over the unitary group. The Clifford group is a unitary $3$-design~\cite{webb_clifford_2016} and fails to be a unitary $k$-design for any $k\ge 4$~\cite{zhu_clifford_2016}. Therefore, it is natural to ask whether injecting $t$ non-Clifford gates can produce higher $k$-designs. Although Refs.~\cite{leone_quantum_2021,oliviero_transitions_2021} showed that $t=\Theta(n)$ is required for an exact $k$-design, in Ref.~\cite{haferkamp2020QuantumHomeopathyWorks}, it was shown that for $t=O(k^4)$ (independent of $n$), a $t$-doped Clifford circuit successfully produces an \emph{approximate} $k$-design. Since approximate $k$-designs suffice for the vast majority of quantum information processing tasks, this is just one demonstration of the usefulness of $t$-doped Clifford circuits.

Lastly, $t$-doped Clifford circuits and $t$-doped stabilizer states have recently attracted attention in the context of quantum learning theory. The objective of quantum learning is to obtain a classical description of an unknown state or unitary transformation while being restricted to a particular access model, a process known as \emph{tomography}. While the task of learning an arbitrary state $\psi$ or unitary operator $U$ has exponential sample complexity in general, it is widely understood that these learning tasks can become tractable by restricting oneself to learning particular classes of circuits or states. To name just a few examples of this, Ref.~\cite{huang2024learning} recently showed that $O(1)$ depth quantum circuits can be learned efficiently, and Ref.~\cite{cramer_efficient_2010} showed that matrix product states can also be learned with polynomial resource requirements. This pattern holds true for the class of $t$-doped Clifford circuits~\cite{leone_learning_2024,oliviero_unscrambling_2024} and $t$-doped stabilizer states~\cite{leone2023learning,grewal2023efficient,hangleiter_bell_2023,grewal2023efficientII,Chia2024efficientlearningof}: each of these works, using distinct methods and access models, have shown that $t$-doped circuits and states can be learned using $\poly(n, 2^t)$ sample complexity and $\poly(n, 2^t)$ classical postprocessing time. However, this means that learning $t$-doped Clifford circuits and $t$-doped stabilizer states becomes infeasible whenever $t=\omega(\log n)$, exhibiting the same threshold as classical simulation algorithms. Besides suggesting that quantum learning theory has a strong relation to classical simulability of quantum mechanics~\cite{landoncardinal2010efficient,cramer_efficient_2010,grewal_efficient_2023}, this fact raises a question: are there information processing tasks that are both relevant and efficient beyond the seemingly ubiquitous threshold of $t=\omega(\log n)$? 

A concurrent work~\cite{grewal2024pseudoentanglement} made progress along these lines. While their work focuses on bounding entanglement entropy (their Theorem 3.1 resembles our \cref{lem:salpha-bound}), our paper focuses on the operational and computational aspects of the entanglement-magic interplay. Along similar lines, in Refs.~\cite{oliviero_unscrambling_2024,leone_learning_2024}, it was demonstrated that quantum information scrambled by a $t$-doped Clifford circuit could be \emph{unscrambled} and recovered using a carefully constructed Clifford circuit whenever $t<n$. This extends far beyond the classical simulability barrier. Such a result, besides providing insight into information scrambling, is an indication that certain information processing tasks can still exhibit efficiency beyond the classical simulation threshold $t=\omega(\log n)$. In other words, while the task of learning a full description of such $t$-doped stabilizer states might be hard, restricting the objective from learning a full description of the state to merely learning certain characteristics of the state can render the problem feasible even beyond the classical simulability barrier. A particular characteristic of quantum states that occupies a central role in quantum information is its \emph{entanglement}. In this direction, a recent work~\cite{leone_phase_2023} shows how even for highly entangled states, estimating the entanglement entropy of states possessing Clifford symmetries can be significantly improved using techniques developed in~\cite{oliviero_unscrambling_2024,leone_learning_2024}. Here, we present a detailed analysis of learning and manipulating the entanglement of $t$-doped stabilizer states, aiming to extend beyond the conventional limits imposed by classical simulability barriers.

\subsection{Future directions}

This paper thoroughly describes the interplay between two resource theories: the resource theory of magic and the resource theory of entanglement. We discuss how magic influences the efficiency of entanglement tasks, presenting both positive and negative results. This presents new research directions and ideas, a few of which we list below. 
\begin{itemize}
    \item The ED phase includes states that contain up to $t=o(\exp n)$ non-Clifford gates, yet we are able to show that entanglement manipulation, detection, and estimation are efficient. This fact raises the question of whether this class of states is efficiently classically simulable; while we believe that is not the case, as we can always embed intractable problems, none of the known classical-simulation algorithms, e.g.~\cite{bravyi_improved_2016} are specifically designed for entanglement-dominated states. We also need not demand full classical simulability, asking instead the slightly weaker question of whether there are other quantum information-related tasks that can be performed efficiently for ED states.

    \item We explored the synergy between the resource theory of entanglement and the resource theory of magic in the context of entanglement theory. We characterized how magic affects entanglement theory in positive and negative directions. However, the other way around is also possible. For instance, the analogue of entanglement distillation in magic resource theory is magic state distillation, in which one aims to distill a maximal number of copies of the magic state $\ket{B} \propto \ket{0} + e^{-i\pi/4} \ket{1}$ using only Clifford operations. Similarly to entanglement distillation, it is known that at most $O(\mathcal{M}(\rho))$ copies of $\ket{B}$ can be distilled from $\rho$, where $\mathcal{M}(\cdot)$ is an arbitrary magic monotone~\cite{veitch_resource_2014}. It is thus interesting to develop a theory, akin to the one developed here, that explores the flip side of this interplay: how does entanglement affect the efficiency of magic-state distillation? While it is clear that the optimal rate of magic-state distillation can be achieved for product states, how is this rate affected by entanglement? Is it perhaps true that the ED phase is characterized by efficient entanglement distillation, while the MD phase is characterized by efficient magic-state distillation? Answering these questions opens an exciting avenue for future research.

    \item The ED phase is characterized by those states for which the entanglement is dominated by the stabilizer properties of the state and, as such, can be easily manipulated. This concept is akin to ``Gaussification," a process where continuous variable states are analyzed by studying their Gaussian component. Particularly relevant are Refs.~\cite{eisert_distillation_2004,browne_driving_2003,eisert_experimental_2007,Campbell_2012}, where the authors study continuous variable entanglement under this lens. It is, therefore, natural to ask whether a theory similar to the one presented in this work can also be developed in the context of bosons and fermions, analyzing the interplay between the resource theory of fermionic (resp. bosonic) non-Gaussianity.

\item A naturally arising question prompted by our work is the possibility of a computationally-efficient and agnostic distillation protocol for Matrix Product States (MPSs). Our no-go results for the MD phase are applicable only under the condition that the Von Neumann entropy scales strictly faster than logarithmically with the local system size. The inadequacy of these results in capturing MPSs suggests the feasibility of an efficient and agnostic entanglement distillation protocol within this regime. One plausible approach involves the efficient learning of the state, with a sample efficiency suitable for polynomially-large bond dimensions~\cite{Cramer_2010,fanizza2023learning}. Subsequently, an analysis of the computational cost to identify the optimal distillation protocol.
\end{itemize}

\section{Preliminaries}
\noindent We provide below a glossary of notations for convenience.
\begin{table}[H]
    \centering
    \begin{tabular}{c|c}
         \textbf{Symbol} & \textbf{Definition} \\
         \hline \hline
         $\sepp{A}{B}$ & The set of separable states with respect to the bipartition $A|B$ \\ \hline
         $\mathbb{P}_n$ & The Pauli group over $n$ qubits \\ \hline
         $G$ & The Abelian group of stabilizers for a state $\psi$ \\ \hline
         $G_A$ & The Abelian subgroup of stabilizers for $\psi$ that act nontrivially only on subsystem $A$ \\ \hline
         $S$ & A generating set for the full set of stabilizers $S$; always has cardinality $\log_2 \abs{G}$ \\ \hline
         $S_A$ & A generating set for $G_A$ \\ \hline
         $\psi_A$& The partial trace of the state $\psi$ respect to the subsystem $B$ in the bipartion $A|B$\\ \hline
         LOCC & Local operations and classical communication protocols \cite{chitambar_everything_2014} \\ \hline
         $\mathcal{E}(\psi; A|B)$ & The stabilizer entanglement of $\psi$ across the bipartition $A|B$ (\cref{def:stab-ent}) \\ \hline
         $S_\alpha(\rho)$ & The $\alpha$-R\'enyi entropy of a state $\rho$ (\cref{def:renyi}) \\ \hline
         $M_{+}$ & The number of Bell pairs distillable across a bipartition \\ \hline
         $M_{-}$ & The number of Bell pairs consumed to dilute a state across a bipartition \\ \hline
         $t$ & The number of non-Clifford gates used to prepare a $t$-doped stabilizer state \\ \hline
         $l$ & The maximum number of qubits any non-Clifford gate acts on \\ \hline
         $\nu$ & The stabilizer nullity of a state; it is defined by $\abs{S} = n-\nu$, and we can bound $\nu \leq 2lt$ \\ \hline
         $\Pi$ & The stabilizer portion of a $t$-doped stabilizer state (\cref{thm:pure-struct})
    \end{tabular}
\end{table}

Stabilizer states are capable of capturing a key aspect of `quantumness' that other classically simulable states (such as matrix product states) are unable to: they can exhibit any degree of entanglement, ranging from complete separability to full volume law entanglement. Moreover, an early work showed that the entanglement of stabilizer states can be computed exactly in $O(n^3)$ time simply by studying the structure of the generators~\cite{fattal_entanglement_2004}. This work is a representative one, in that it recognized how two important characteristics of stabilizer states (i.e., compact classical representation and ability to exhibit high entanglement) could be combined in powerful ways to better understand the nature of entanglement. Indeed, this combination has made the stabilizer formalism an extremely useful tool for a wide array of applications ranging from quantum error correcting codes to condensed matter physics. However, entanglement is not all there is to quantum computation: early work showed that Clifford circuits (the unitaries that generate stabilizer states) are probably not even universal for purely \emph{classical} computation \cite{aaronson_improved_2004}. This motivated the introduction of quantum magic as a means of quantifying the degree to which a state could be approximated using stabilizer states. In other words, magic measures the amount of non-Clifford resources required to prepare a particular state. Then, simply by definition, entanglement is insufficient for true quantum computational advantage: one must have a nontrivial degree of quantum magic as well.

A stabilizer state $\ket{\psi}$ is a state such that there is a size $2^n$ Abelian subgroup of the Pauli group $\mathbb{P}_n$ whose elements $P$ all satisfy $P \ket{\psi} = \pm \ket{\psi}$. This subgroup, known as the stabilizer group $G$, is finitely generated by $n$ Pauli operators $S=\qty{\phi_j g_j \mid g_j \in \mathbb{P}_n, \phi_j \in \qty{-1,+1}}$, where the generators $g_j$ have an associated phase $\phi_j$ depending on whether $\ket{\psi}$ is a $+1$ or $-1$ eigenvector of $g_j$. In this work, the phase is almost always irrelevant, so to simplify notation, we will absorb the irrelevant phase $\phi_j$ into $g_j$ (so we always assume $g_j \ket{\psi} = \ket{\psi}$). Then, stabilizer states can always be written as
\begin{equation}
    \ketbra{\psi} = \prod_{i=1}^{n} \qty(\frac{\id + g_j}{2}).\label{eq:stab}
\end{equation}
The compact representation of stabilizer states in terms of just $n$ generators immediately lends itself to their efficient simulability
\cite{gottesman_heisenberg_1998}. We note that the generators $g_j$ are not unique -- that is, there are many possible choices of generators $S$ such that $G = \expval{S}$. This can be understood using the tableau representation of the stabilizer group, first introduced in Ref.~\cite{aaronson_improved_2004}, which takes advantage of the isomorphism between the Pauli group and $\mathbb{F}_2^{2n}$. This isomorphism is as follows: any generator $g_j \in \mathbb{P}_n$ can be represented as a length $2n$ binary string $\vec{s}_j = x_{j,1} x_{j,2} \ldots x_{j,n} z_{j,1} z_{j,2} \ldots z_{j,n}$, such that $g_j \propto \prod_{i=1}^n (X_i)^{x_{j,i}} (Z_i)^{z_{j,i}}$ (where the proportionality is up to a factor $i^l$ for some integer $l$). The product of any two Pauli operators $g_j g_k$ corresponds to a binary string that is equal to $\vec{s}_j \oplus \vec{s}_k$, where $\oplus$ is addition modulo 2. We then see that if we choose a particular set of generators $g_j$ with binary string representations as defined above, we can write a tableau representation of the stabilizer group
\begin{equation}
\mathcal{T}_S = \left[
\begin{array}{cccc|cccc}
x_{1,1} & x_{1,2} & \ldots & x_{1,n} & z_{1,1} & z_{1,2} & \ldots & z_{1,n} \\
x_{2,1} & x_{2,2} & \ldots & x_{2,n} & z_{2,1} & z_{2,2} & \ldots & z_{2,n} \\
\vdots & \vdots & \ddots & \vdots & \vdots & \vdots & \ddots & \vdots \\
x_{\abs{S},1} & x_{\abs{S},2} & \ldots & x_{\abs{S},n} & z_{\abs{S},1} & z_{\abs{S},2} & \ldots & z_{\abs{S},n}
\end{array}
\right] \in \mathbb{F}_2^{\abs{S} \times 2n},\label{eq:tableau}
\end{equation}
where each row corresponds to a generator $g_j$. Although we note that $\abs{S}=n$ by definition for all pure stabilizer states, we leave \cref{eq:tableau} in this general form to allow for the possibility that the stabilizer group is not full rank (i.e., $\abs{S} < n$). Since the entire stabilizer group is given by $\rowspan(\mathcal{T}_S)$, the tableau can be changed using any sequence of elementary row operations (i.e., swapping rows or adding rows to other rows), as these operations do not change the row space of a matrix. We will find it useful to define the local stabilizer group $G_A$ ($G_B$), which is the subgroup of stabilizers that act as the identity on the subsystem $B$ ($A$). A tableau representation of the local stabilizer group $G_A$ can be found efficiently as follows. We define a linear map $\mathcal{P}_B: g_A \otimes g_B \to \id \otimes g_B$ which projects out the $A$ component of a given stabilizer. Acting with $\mathcal{P}_B$ on $\mathcal{T}_S$ simply sets all $x_{j,i}=z_{j,i}=0$ for all $j=1,\ldots,\abs{S}$ and any $i \in A$, and leaves the rest of the matrix elements unchanged. Then, $G_A = \ker(\mathcal{P}_B(\mathcal{T}_S)^T)$, so finding a basis for the rows of $\ker(\mathcal{P}_B(\mathcal{T}_S)^T)$ suffices to find a tableau representation for the generators of $G_A$, which we call $S_A$. The procedure for finding $S_B$ follows identically. 

\subsection{Stabilizer nullity and \texorpdfstring{$\nu$}{t}-compressible states}\label{sec:stabnullity}
Closely related to stabilizer states are the Clifford unitaries, which are the group $\mathcal{C}_n$ of unitaries that normalize the Pauli group. That is, $\mathcal{C}_n \coloneqq \qty{U \mid U \mathbb{P}_n U^\dagger = \mathbb{P}_n}$. Since Clifford unitaries map stabilizer states to stabilizer states (they simply act by mapping each of the generators $g_j \to U g_j U^\dagger \in \mathbb{P}_n$), Clifford circuits, initialized with a stabilizer state (such as $\ket{0}^{\otimes n}$) are classically simulable. The canonical gateset for Clifford circuits is comprised of the Hadamard gate $H$, the phase gate $S$, and the $\cnot$ gate. Yet, as discussed previously, Clifford gates alone are insufficient to realize universal quantum computation: one needs to include just one non-Clifford gate to realize full universality. Injection of non-Clifford resources in the state preparation results in a property known as \textit{magic}. There are many choices to reasonably quantify magic~\cite{campbell_catalysis_2011,veitch_resource_2014,beverland_lower_2020,liu_manybody_2022,saxena_quantifying_2022,leone_nonstabilizerness_2023}. In this work, we utilize the \textit{stabilizer nullity}, introduced in Ref.~\cite{beverland_lower_2020}. In this work, this proves to be the most natural choice.
\begin{definition}[Stabilizer nullity] Consider a pure state $\ket{\psi}$. If $\ket{\psi}$ is stabilized by $n-\nu$ algebraically independent Pauli operators $P$ (i.e. $P\ket{\psi}=\ket{\psi}$), then the state has stabilizer nullity $\nu$. Equivalently, the state $\ket{\psi}$ has a stabilizer group, hereby denoted as $G$, of cardinality $\abs{G}=2^{n-\nu}$. This notion can be easily extended to mixed states $\rho$. If $\rho$ commutes with $n-\nu$ algebraically independent Pauli projectors (i.e. $\frac{(\id+P)}{2}\rho \frac{(\id+P)}{2}=\rho$), then we say the state has stabilizer nullity $\nu$.
\end{definition}

Given our earlier claim that these states might be understood as perturbations of exact stabilizer states, one might expect that these states should exhibit a similar structure to stabilizer states. We show that indeed we can write states with stabilizer nullity $\nu$ using an algebraic structure which closely mirrors that of \cref{eq:stab}. 

\begin{theorem}[Algebraic structure of states with stabilizer nullity $\nu$~\cite{leone2023learning}]\label{thm:pure-struct}
Any pure state $\ketbra{\psi}$ with stabilizer nullity $\nu$ can be written
\begin{equation}
    \ketbra{\psi} = \underbrace{\frac{1}{2^{\nu}}\qty(\id + \sum_{i=1}^k \tr(h_i \psi) h_i)}_{H} \underbrace{\qty(\prod_{j=1}^{n-\nu} \qty(\frac{\id + g_j}{2}))}_{\Pi}; \,\, k+1 \leq 2^{2\nu},\, \nu \coloneqq n-\abs{S} \leq 2lt, \label{eq:structure1}
\end{equation}
where $S \coloneqq \qty{g_j \in \mathbb{P}_n \mid j=1,\ldots,\abs{S}}$ is a generating set for the stabilizer group of $\ket{\psi}$. The `bad generators' $h_i \in \mathbb{P}_n$ of $\ket{\psi}$, which do not necessarily obey any special structure, beyond the fact that $\tr(h_i \psi) \neq 0$. $\Pi$ is a rank-$2^{\nu}$ projector, and $H$ is a Hermitian operator whose operator norm is bounded: $\norm{H}_\infty \leq 2^{\nu}$. Finally, $\Pi$ and $H$ commute.
\end{theorem}
\begin{proof}
 By assumption, the size of the stabilizer group $\abs{G}=2^{n-\nu}$. Take any Pauli $h_i$ for which $\tr(h_i \ketbra{\psi}) \neq 0$. Observe that every Pauli in the coset $h_i G$ also has expectation $\tr(h_i \ketbra{\psi})$, as $\tr(h_i P_j \ketbra{\psi})=\tr(h_i \ketbra{\psi})$ for any stabilizer $P_j$, as $P_j \ket{\psi} = \ket{\psi}$. We can then make use of two facts. First, from the above reasoning,  the set of Paulis which have nonzero support on $\ket{\psi}$, namely $\supp(\psi) = \qty{P \in \mathbb{P}_n \mid \tr(P \psi) \neq 0}$, can be partitioned into disjoint cosets $h_i G$ for $i=0,\ldots,k$ and we use the convention that $h_0=\id$. Second, as the $\alpha=0$ stabilizer R\'enyi entropy $M_0$ is upper bounded by $\nu$ \cite{leone_stabilizer_2022}, and $M_0$ is defined as $\log \abs{\supp(\psi)} - \log n$, we immediately have $\abs{\supp(\psi)} \leq 2^{n + \nu}$. These two facts allow us to conclude that there are at most $\frac{\abs{\supp(\psi)}}{\abs{G}}\leq \frac{2^{n+\nu}}{2^{n-\nu}}=2^{2\nu}$ coset representatives $h_i$. This is to say, $k+1 \leq 2^{2\nu}$. In summary, we can write
\begin{equation}
    \ketbra{\psi} = \frac{1}{2^n} \sum_{i=0}^k \tr(h_i \psi) \sum_{P_j \in G} h_i P_j = \frac{1}{2^\nu} \qty(\id + \sum_{i=1}^k \tr(h_i \psi) h_i) \qty(\prod_{j=1}^{n-\nu} \qty(\frac{\id+g_j}{2})).
\end{equation}
\end{proof}
\begin{remark}
The stabilizer component of the state $\Pi$ is a valid state of its own, and importantly, it is a (mixed) stabilizer state. Although it may seem only to be an abstract mathematical tool, this state has a significance of its own. Indeed, in Ref. \cite{bu2023entropy}, it was shown that the state $\Pi$ has a very special relationship with $\ketbra{\psi}$. For one, it is the closest stabilizer state to $\ketbra{\psi}$, as measured by the relative entropy $D(\Pi \parallel \psi) \coloneqq -\Tr \Pi \log \psi - S_1(\Pi)$. Furthermore, $\Pi$ can be realized by `averaging' many copies of $\ketbra{\psi}$ using an operation known as the quantum convolution.
\end{remark}

The structure established by the above theorem shows that states with nullity $\nu$ can be \emph{algebraically} separated into a stabilizer component $\Pi$ and a non-stabilizer component $H$, where each component has a number of terms exponential in $n-\nu$ and $\nu$, respectively. This might lead one to hope that the components can be separated into a tensor product structure; that is, whether $\ketbra{\psi}$ can be factored into two separate \emph{states}, one corresponding to $H$ and one for $\Pi$. The following algorithm establishes this is possible. We term this protocol `nullity distillation', because it essentially distills the non-stabilizer portion of a state $H$ onto a number of qubits equal to the stabilizer nullity (see \cref{fig:alg-struct}). This protocol constructs a Clifford unitary that takes as input an $n$-qubit state $\rho$ with $\abs{S}$ stabilizer generators, and outputs a state $\ketbra{0}^{\otimes \abs{S}} \otimes \rho''$, where $\rho''$ is a state on $n-\abs{S}$ qubits. The idea is simple: we first identify a Clifford unitary that maps the first generator to $Z_1$, after which the state must be $\ketbra{0} \otimes \rho'$ -- this is because if a state is stabilized by $Z_1$, it must have $\Tr(\rho_1 Z_1)=1 \implies \rho_1 = \frac{\id + Z_1}{2} \implies \rho = \ketbra{0} \otimes \rho'$, where $\rho_1$ is the reduced density matrix of $\rho$ on qubit 1. We then find a Clifford (acting on the remaining $n-1$ qubits) that maps the second generator to $Z_2$, and so on, until we have a state $\ketbra{0}^{\otimes \abs{S}} \otimes \rho''$. In \cref{sec:err-correct}, we will see how the nullity distillation circuit can be interpreted as the \emph{decoding circuit} for a stabilizer error correcting code.

\begin{algorithm}[H]
\caption{Nullity distillation}\label{alg:nullity}
\begin{algorithmic}[1]
\Require{A state $\rho$ with stabilizer generators $S=\qty{g_1, g_2, \ldots, g_{\abs{S}}}$}
\Ensure{A state $U\rho U^\dagger = \ketbra{0}^{\otimes \abs{S}} \otimes \rho'$}
\Function{NullityDistillation}{$\rho$}
\State $U \gets \id$
\For{$i \gets 1, \ldots, \abs{S}$}
\State $V \gets $ any Clifford acting on qubits $\qty{i,i+1,\ldots,n}$ such that $V g_i V^\dagger = Z_i$
\State $U \gets U \circ V$
\EndFor
\State \Return{$U \rho U^\dagger$}
\EndFunction
\end{algorithmic}
\end{algorithm}

These considerations motivate us to term states with stabilizer nullity $\nu$ as \textit{$\nu$-compressible states}, consistently reminding the reader that nullity can be distilled away using a simple Clifford operation.

\begin{figure}
    \centering
    {
        \graphicspath{{figures/}}
        \newcommand{\fs}{\scriptsize}
        \newcommand{\fss}{\footnotesize}
    \def\svgwidth{0.65\textwidth}
    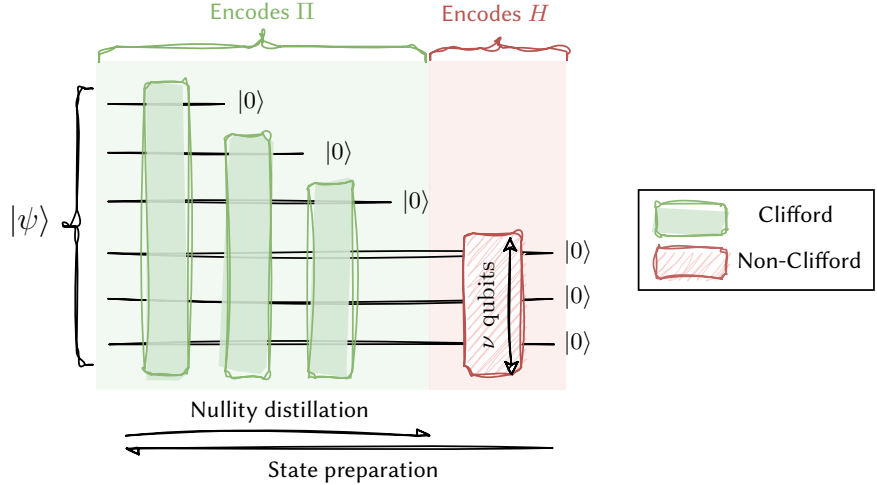
    }
    \caption{We visualize the algebraic structure of $\nu$-compressible states. Any $\nu$-compressible state can be prepared using a general unitary acting on $\nu$ qubits, which depends only on $H$, and some global Clifford unitary which depends only on $\Pi$.}
    \label{fig:alg-struct}
\end{figure}

\begin{definition}[$\nu$-compressible states]\label{cor:nucompressibility} Let $\ket{\psi}$ be a pure quantum state. We say $\ket{\psi}$ is \emph{$\nu$-compressible} if it can be prepared by the application of a unitary $U$ acting on $\nu$ qubits followed by a Clifford circuit $C$. Any state $\ket{\psi}$ is $\nu$-compressible if and only if it has stabilizer nullity $\nu$. The equivalence between stabilizer nullity $\nu$ states and $\nu$-compressible states is established by \Cref{alg:nullity}.
\end{definition}

A well-known class of $\nu$-compressible states is the set of $t$-doped stabilizer states. These are states created by the action of Clifford circuits containing a number $t$ of $l$-qubit non-Clifford gates. Let us formally define them below.

\begin{figure}
    \centering
    {
        \graphicspath{{figures/}}
        \newcommand{\fs}{\scriptsize}
        \newcommand{\fss}{\footnotesize}
    \def\svgwidth{0.7\textwidth}
    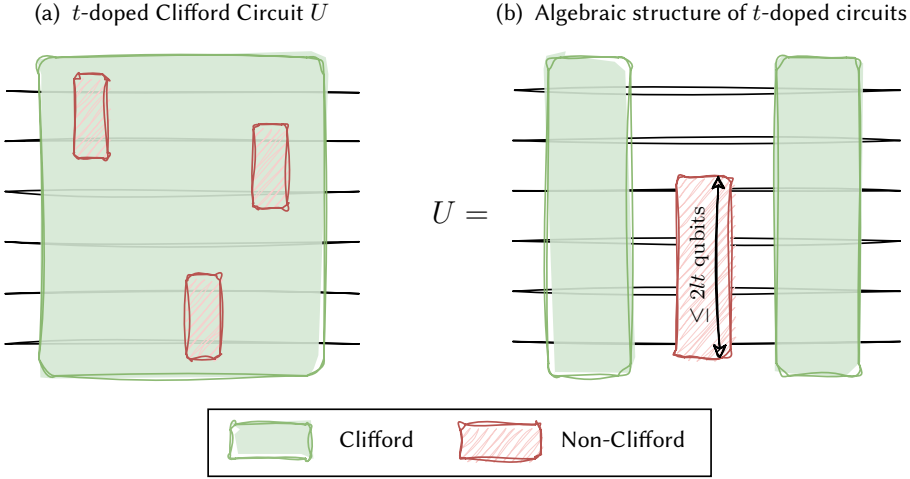
    }
    \caption{In (a), we show a $t$-doped Clifford circuit, where the green background indicates a `substrate' of Clifford gates, which are implicitly assumed to dominate the circuit compared to the doped non-Clifford gates. In this circuit, we have set in particular $t=3$ (the number of non-Clifford gates) and $l=2$ (the number of qubits each non-Clifford acts on). In (b), we show an equivalent representation of the $t$-doped circuit shown in (a): it can always be rewritten as a Clifford $C_1$, with a general (non-Clifford) unitary acting on $\leq 2lt$ qubits, followed by another Clifford $C_2$~\cite{leone_learning_2024,oliviero_unscrambling_2024}.}
    \label{fig:clifft}
\end{figure}

\begin{definition}[$t$-doped stabilizer states]\label{tdopedstabilizerstates}
An $n$-qubit state is a pure $t$-doped stabilizer state if $\ket{\psi} = U \ket{0}^{\otimes n}$, where a $t$-doped Clifford circuit $U$ is a unitary operator containing Clifford gates and at most $t$ non-Clifford $l$-qubit gates with $l=O(1)$ (see \cref{fig:clifft} for a visualization). A mixed $t$-doped stabilizer state is any state $\rho$ for which there exists a purification $\ket{\psi}$ such that $\ket{\psi}$ is a pure $t$-doped stabilizer state, and the purifying system has at most $\lceil \log \rank \rho \rceil$ qubits. For brevity, we will often refer to $t$-doped stabilizer states simply as $t$-doped states.
\end{definition}

\begin{lemma}[$t$-doped stabilizer states are $\nu$-compressible states] Let $\ket{\psi}$ be a $t$-doped stabilizer states. Then $\ket{\psi}$ is $\nu$-compressible with $\nu\le 2lt$.
\begin{proof}
This has been shown in Ref.~\cite{jiang_lower_2021} (see also
Ref.~\cite{leone_learning_2024,oliviero_unscrambling_2024}). The application of any $l$-qubit non-Clifford gate removes at most $2l$ stabilizer generators from a state, so after applying $t$ of these gates, we have at most $\nu \leq 2lt$.
\end{proof}
\end{lemma}
However, the class of $\nu$-compressible states is far from being limited to $t$-doped stabilizer states. For instance, unitaries that encode completely general $H$ (see \cref{fig:alg-struct}) can have up to exponential depth in $\nu$, hence up to $\exp(\nu)$ non-Clifford gates. This shows that not only are $\nu$-compressible states strictly more general than $t$-doped stabilizer states, they even include classes of states which are not efficiently preparable.

\smallskip

\parhead{Simulability of $\nu$-compressible states.} $t$-doped stabilizer states are ubiquitous in the quantum information literature (and as we will see in \cref{sec:physical}, they appear in surprising places in physics as well). One aspect of these states that has been closely studied is their classical simulability; existing algorithms generally achieve a time (and space) complexity $\poly(n, 2^{c \cdot t})$ for some constant $c$ \cite{bravyi_improved_2016,bravyi_quantum_2018}. In fact, it has been shown that no simulation algorithm can achieve subexponential runtime in terms of $T$-gate count~\cite{huang2019explicit}. This result holds for single-qubit non-Clifford gates, but more broadly, it is easy to see that simulation of states that have stabilizer nullity up to $\nu$ must have runtimes that scale exponentially in $\nu$. \cref{alg:nullity} makes this clear: the set of general $\nu$-qubit states are in a one-to-many correspondence with states that have stabilizer nullity $\nu$, so a general simulation algorithm for states with stabilizer nullity $\nu$ would immediately imply a simulation algorithm for general $\nu$-qubit states. Since such a general simulation algorithm is widely believed to require exponential (in $\nu$) classical runtime, it follows that simulating arbitrary states with stabilizer nullity $\nu=\omega(\log n)$ cannot be efficient classically.

\subsection{Entanglement-dominated states and where to find them}

\begin{figure}
    \centering
    {
        \graphicspath{{figures/}}
        \newcommand{\fs}{\scriptsize}
        \newcommand{\fss}{\footnotesize}
    \def\svgwidth{0.9\textwidth}
    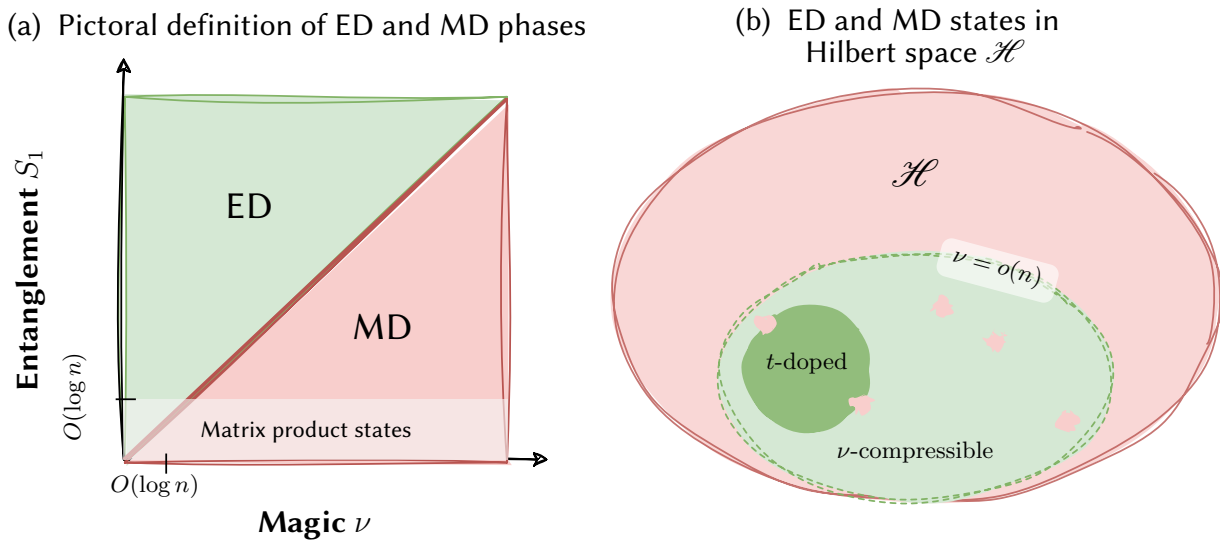
    }
    \caption{In (a), we show a representation of the entanglement- and magic-dominated phase. The entanglement is quantified, given the bipartition, by the entanglement entropy, while magic is quantified by the stabilizer nullity (\cref{def:ent-dom}). The class of matrix product states occupies a special place in our theory, as we are unable to rule out efficient entanglement manipulation protocols for MPS, despite the fact that most of these states are magic-dominated. Indeed, examples of efficient MPS-specific protocols have already been found~\cite{cramer_efficient_2010}. In (b), we show the portion of Hilbert space occupied by $\nu$-compressible states for $\nu=o(n)$. These states occupy a vanishing fraction of the full Hilbert space. Similarly, a counting argument shows that $t$-doped states (which are all $2lt$-compressible states) are also a vanishing fraction of all $2lt$-compressible states. However, for both of these classes of states, the entanglement-dominated phase is typical (see \cref{th:typical-magic,th:typical,lem:typicalitynu}).}
    \label{fig:venn}
\end{figure}
Before exploring the operational distinctions between the entanglement- and magic-dominated phases, a natural question to ask is where entanglement-dominated states can be found in practice. Recall that the two phases were defined in \cref{def:ent-dom} in terms of stabilizer nullity and entanglement entropy. The following lemma states that the vast majority of states are MD. This may seem to throw cold water on the hope of finding any entanglement-dominated states, but for the reasons we discuss below, this is not so.
\begin{lemma}[Typicality of magic-dominated states]\label{th:typical-magic}
For any bipartition $A|B$ with $n_A \leq n_B$, the states in the Haar ensemble are magic dominated with probability 1:
\begin{equation}
    \Pr_{\ket{\psi} \sim \haar}[\ket{\psi} \text{is magic-dominated}\,] = 1.
\end{equation}
\end{lemma}
\begin{proof}
All but a measure zero set of states in the Haar measure have no stabilizers, so they have $\nu=n$ with probability $1$. Since the entanglement of these states is at most $n_A$ (i.e., strictly less than $n$), these states are all magic-dominated.
\end{proof}
In spite of this lemma, the vast majority of states in the Haar measure will also never be observed in practice, as they take exponential time to prepare. Below, we will explore a much more practically relevant setting, where entanglement-dominated states turn out to be ubiquitous: \emph{polynomial-depth random quantum circuits}, or more specifically, unitary $k$-designs. A number of works \cite{haferkamp2020QuantumHomeopathyWorks,leone_quantum_2021,oliviero_transitions_2021} have shown that a very simple form of random circuit forms a unitary $k$-design. These circuits take the form of a sequence of global random Clifford unitaries, interleaved with local Haar random unitaries (which are, importantly, not necessarily Clifford unitaries). The measure over this random ensemble of circuits is $\mu_t$:
\be
\mu_t=\mu_{\text{Cl}_n}*\mu_{\haar_l}*\mu_{\text{Cl}_n} *\mu_{\haar_l}* \underbrace{\cdots}_{t\,\text{times}} *\,\mu_{\text{Cl}_n}*\mu_{\haar_l}*\mu_{\text{Cl}_n}
\ee
where $*$ denotes the convolution between measures, $\mu_{\haar_l}$ is the measure over $l$-qubit Haar random unitaries (acting on qubit the first $l$ qubits) with $l=O(1)$, and $\mu_{\text{Cl}_n}$ is the uniform distribution over $n$-qubit Clifford unitaries. Crucially, this measure generates an approximate $k$-design with just $t=O(\poly k)$. That is, using very few non-Clifford resources, we can quickly generate good $k$-designs~\cite{haferkamp2020QuantumHomeopathyWorks}. This leads us to observe that this ensemble generates states that are almost exclusively entanglement-dominated. We will use $\mu(\mathcal{S}_t)$ to denote the measure over states induced by $\mu_t$: that is, we sample from $\mu(\mathcal{S}_t)$ by sampling a unitary $U \sim \mu_t$ and applying $U \ket{0}^{\otimes n}$.

\begin{lemma}[Typicality of entanglement-dominated phase within $\mu(\mathcal{S}_t)$]\label{th:typical}
For any bipartition $A|B$, most $t$-doped stabilizer states, with $t=o(n_A)$, are entanglement-dominated. That is:
\be
\Pr_{\ket{\psi} \sim \mu(\mathcal{S}_t)}\qty[\ket{\psi}\,\text{is entanglement-dominated}\,]\geq  1-2^{-n_A/2 + 1}
\ee
\begin{proof}
If $t=o(n_A)$, a sufficient condition for $\ket{\psi}$ to be entanglement-dominated is that $S_2(\psi_A) \geq \frac{n_A}{2}$, since $S_1(\psi_A) \geq S_2(\psi_A)$. Recall that $S_{2}(\psi_A)=-\log \Tr(\psi_A^2)$, and applying the Markov inequality, we have
\be
\Pr_{\ket{\psi} \sim \mu(\mathcal{S}_t)}\qty[S_{2}(\psi_A)\ge \frac{n_A}{2}]\ge 1-\frac{\mathbb{E}_{\ket{\psi} \sim \mu(\mathcal{S}_t)}[\pur(\psi_{A})]}{2^{-n_A/2}} \geq 1 - 2^{-n_A/2 + 1},
\ee
where in the second inequality we use the fact that $\mu(\mathcal{S}_t)$ forms at least a $2$-design to conclude that $\mathbb{E}_{\mu(\mathcal{S}_t)}[\pur(\psi_{A})]=\frac{2^{n_A}+2^{n_B}}{2^n+1}\le \frac{2}{2^{n_A}}$. 
\end{proof}
\end{lemma}

However, as explained in \cref{sec:stabnullity}, $t$-doped stabilizer states are only one example of $\nu$-compressible states. Indeed, as illustrated in \cref{fig:venn}, they form only a vanishing fraction of those. Therefore, we study the much more general class of $\nu$-compressible states, and find that they too are almost all ED. To establish this result, we choose the uniform measure on $\nu$-compressible states, defined as the convolution of the uniform measure over the Clifford group $\mu_{\mathrm{Cl}_n}$ and $\mu_{\mathrm{Haar}_{\nu}}$ representing the uniform measure on the unitary group. In other words, $\mu_{\nu}= \mu_{\mathrm{Cl}_n}* \mu_{\mathrm{Haar}_{\nu}}$. Recalling \cref{cor:nucompressibility}, we have the following lemma.

\begin{lemma}[Typicality of ED phase for $\nu$-compressible states]\label{lem:typicalitynu} Most $\nu$-compressible states, drawn from the measure $\mu_\nu$ with $\nu=o(n)$, are entanglement-dominated with respect to exponentially many cuts $A|B$, where each half of the bipartition is size $\Omega(n)$.
\begin{proof}
Take any cut $A|B$ with $n_A,n_B=\Omega(n)$. A sufficient condition for $\ket{\psi}$ to be entanglement-dominated is that $S_2(\psi_A) \geq \frac{n_A}{2}$, since $S_1(\psi_A) \geq S_2(\psi_A)$. Recall that $S_{2}(\psi_A)=-\log \Tr(\psi_A^2)$, and applying the Markov inequality, we have
\be
\Pr_{\ket{\psi} \sim \mu_\nu}\qty[S_{2}(\psi_A)< \frac{n_A}{2}]\le\frac{\mathbb{E}_{\ket{\psi} \sim \mu(\mathcal{S}_t)}[\pur(\psi_{A})]}{2^{-n_A/2}} \le 2^{-n_A/2 + 1},
\ee
where in the second inequality we use the fact that $\mu_{\nu}$ forms at least a $2$-design to conclude that $\mathbb{E}_{\mu(\mathcal{S}_t)}[\pur(\psi_{A})]=\frac{2^{n_A}+2^{n_B}}{2^n+1}\le \frac{2}{2^{n_A}}$ (where we chose $n_A\le n_B$ WLOG). Since $n_A=\Omega(n)$, using the union bound, we can tolerate up to exponentially many cuts and still conclude that the probability that $\ket{\psi} \sim \mu_{\nu}$ is ED over exponentially many cuts is $1-o(1)$.
\end{proof}
\end{lemma}

The findings of \cref{th:typical-magic,th:typical,lem:typicalitynu} are summarized in \cref{fig:venn}(b). The majority of states in the Hilbert space are in the MD phase, but the majority of $\nu$-compressible states are ED states. Within $\nu$-compressible states, $t$-doped stabilizer states with a sublinear number of non-Clifford gates are also ED states.

\subsection{Stabilizer entanglement}
When dealing with entanglement, a fundamental question is to understand whether a state is entangled or not, and more specifically, what the degree of entanglement is. There are many different available measures for entanglement. Here, we focus on the well-known $\alpha$-Rényi entropy.
\begin{definition}[$\alpha$-R\'enyi entanglement entropy]\label{def:renyi}
For a pure bipartite state $\ket{\psi}_{AB}$, the $\alpha$-R\'enyi entanglement entropy across the bipartition $A|B$ is
\begin{equation}
    S_\alpha(\psi_A) = \frac{1}{1-\alpha} \log(\tr(\psi_A^\alpha)) = S_\alpha(\psi_B).
\end{equation}
For the special cases $\alpha\in \qty{0,1,\infty}$, the R\'enyi entropy can be defined as the limit of $S_\alpha$ as $\alpha \rightarrow 0$, $1$ (which corresponds to the familiar von Neumann entropy), or $\infty$. Explicitly, for these three cases
\begin{equation}
    S_\alpha(\psi_A) = \begin{cases}
        \log \rank \psi_A &\qq{for $\alpha=0$} \\
        -\tr(\psi_A \log \psi_A) &\qq{for $\alpha=1$} \\
        -\log \norm{\psi_A}_\infty &\qq{for $\alpha=\infty$.}
    \end{cases}
\end{equation}
An important fact about the R\'enyi entanglement entropies is that they are monotonically decreasing in $\alpha$, which is to say that $S_\alpha(\psi_A) \geq S_\beta(\psi_A)$ for any $\alpha \leq \beta$.
\end{definition}
Given the fact that the R\'enyi entropies are defined in terms of the reduced density matrix $\psi_A$, we can use the structure of a $\nu$-compressible state to \emph{exactly} calculate any R\'enyi entropy of a generic state. By using this invariance, as summarized by \cref{eq:structure1}, we show below that the case of $\alpha=2$ is particularly simple.
\begin{lemma}[Exact 2-R\'enyi entropy]\label{lem:exact-2}
For any pure state $\ket{\psi}$, the $2$-R\'enyi entanglement entropy across a bipartition $A|B$ is
\begin{equation}
    S_2(\psi_A) = (n_A - \abs{S_A}) - \log(\sum_{i=0}^k \delta_i \tr^2(h_i \psi)),\label{eq:s2}
\end{equation}
where $S_A$ are the generators for the local stabilizer group $G_A$ (i.e., the subgroup of the original stabilizer group that acts trivially on the $B$ subsystem), $\delta_i \in \qty{0,1}$ is an indicator variable for whether there is some stabilizer $P_i$ of $\ket{\psi}$ such that $P_i h_i$ acts trivially on the $B$ subsystem, and by convention $h_0=\id$, $\delta_0=1$.
\end{lemma}
\begin{proof}
An equivalent formulation of \eqref{eq:structure1} is
\begin{equation}
    \ketbra{\psi} = \frac{1}{2^{n-\abs{S}}}\sum_{i=0}^k \tr(h_i \psi) h_i \Pi = \frac{1}{2^n} \sum_{i=0}^k \tr(h_i \psi) \sum_{P_j \in G} h_i P_j.
\end{equation}
Note that if there is no stabilizer $P_j \in G$ such that $P_j h_i$ acts trivially on the $B$ subsystem, then $\tr_B(h_i \Pi) = 0$. We will call the surviving indices $\mathcal{I} = \qty{i \mid \exists P_i \in G;\, \tr_B(h_i P_i) \neq 0}$. We can write
\begin{equation}
    \psi_A = \tr_B(\ketbra{\psi}) = \frac{1}{2^{n-\abs{G}}}\sum_{i \in \mathcal{I}} \tr(h_i \psi) \tr_B(h_i P_i \Pi),
\end{equation}
since $P_i \Pi = \Pi$. If we expand $h_i P_i \Pi$ as a linear combination of Paulis $h_i P_i P_j$, we quickly see that since $h_i P_i$ acts as the identity on $B$, $h_i P_i P_j$ survives the partial trace if and only if $P_j$ also acts as the identity on $B$, which is to say that it is generated by $S_A$. That is, $\tr_B(h_i P_i \Pi)=h_i P_i \Pi_A$, where $\Pi_A = \prod_{j=1}^{\abs{S_A}} \qty(\frac{\id+g_j}{2})$. In summary, defining $h_i' \coloneqq h_i P_i$:
\begin{equation}
    \psi_A = \frac{1}{2^{n_A}}\sum_{i \in \mathcal{I}} \tr(h_i \psi) \sum_{P_j \in G_A} h_i' P_j = \underbrace{\frac{1}{2^{n_A-\abs{S_A}}}\qty(\id + \sum_{i\in \mathcal{I}} \tr(h_i \psi) h_i')}_{H_A} \underbrace{\qty(\prod_{j=1}^{\abs{S_A}} \qty(\frac{\id + g_j}{2}))}_{\Pi_A}.\label{eq:structure2}
\end{equation}
The purity of this state is exactly
\begin{equation}
    \tr(\psi_A^2) = \frac{1}{4^{n_A}} \sum_{i \in \mathcal{I}} \tr^2(h_i \psi) 2^{n_A+\abs{S_A}} = \frac{1}{2^{n_A - \abs{S_A}}} \sum_{i=0}^k \delta_i \tr^2(h_i \psi).
\end{equation}
Since $S_2(\psi_A) = -\log \tr(\psi_A^2)$, this shows the desired claim.
\end{proof}

Each of the linear-algebraic manipulations in \cref{alg:2renyi} can be done efficiently using the tableau formalism, including the calculation of $\abs{S_A}$ and $\delta_i$. The primary computational cost is simply the loop over $i=0,\ldots,k$, as $k$ can be as high as $2^{2\nu}$. The overall computational complexity of \cref{alg:2renyi} is $O(4^\nu n^3)$. Therefore, similar to low-rank stabilizer simulation algorithms, this exact calculation of the 2-R\'enyi entropy can be done in polynomial time for $\nu=O(\log n)$.

At this point, we recall our earlier assertion that most of the entanglement structure could be ascertained simply by looking at the stabilizers of a state. Indeed, the form of \cref{eq:s2} is quite suggestive: it appears that there is a `stabilizer' component of the entanglement $n_A-\abs{S_A}$, and a correction due to the non-stabilizer components of the state. The following lemma makes this intuition precise.
\begin{lemma}\label{lem:salpha-bound}
For any pure state $\psi$ and bipartition $A|B$, the following holds for all $\alpha$:
\begin{equation}\label{eq:naga}
    n_A - \abs{S_A} - 2\nu \leq S_\alpha(\psi_A) \leq n_A - \abs{S_A},
\end{equation}
where $\psi_A = \tr_B(\psi)$. For the von Neumann entropy, a tighter bound
\begin{equation}
    n_A - \abs{S_A} - \nu \leq S_1(\psi_A) \leq n_A - \abs{S_A}\label{eq:vn-tight}
\end{equation}
holds. An identical bound holds for any R\'enyi entropy with $\alpha \leq 2$.
\end{lemma}
\begin{proof}
Recalling  the monotonicity of the $\alpha$-R\'enyi entropies, we have $S_\infty(\psi_A) \leq S_\alpha(\psi_A) \leq S_0(\psi_A)$, so it suffices to lower bound $S_\infty(\psi_A)$ and upper bound $S_0(\psi_A)$. Referring to the algebraic structure of $\psi_A$ in \cref{eq:structure2}, the lower bound proceeds as follows: $-S_\infty(\psi_A) = \log \norm{\psi_A}_\infty = \log \norm{H_A \Pi_A}_\infty \leq \log \norm{H_A}_\infty \leq \log(\frac{k+1}{2^{n_A - \abs{S_A}}}) \leq -(n_A - \abs{S_A}) + 2\nu$, so we have $S_\infty(\psi_A) \geq n_A - \abs{S_A} - 2\nu$. For the upper bound, $S_0(\psi_A) = \log \rank \psi_A = \log \rank(\Pi_A H_A) \leq \log \rank \Pi_A = n_A - \abs{S_A}$.

The proof of the second inequality follows by lower bounding $S_2(\psi_A)$ using \cref{lem:exact-2}. Specifically, it suffices to show that $\log(\sum_{i=0}^k \delta_i \tr^2(h_i \psi)) \leq \nu$. Observe that by applying a nullity distillation unitary $C$ to the state $\psi$, we get $\ket{0}^{\otimes (n-\nu)} \otimes \ket{\phi}$, where $\ket{\phi}$ is a $\nu$-qubit state. We have that
\begin{equation}
    \sum_{i=0}^k \tr^2(h_i \psi) = \sum_{i=0}^k \tr^2((C h_i C^\dagger) (C \psi C^\dagger)) = \sum_{i=0}^k \tr^2((\id_{n-\nu} \otimes P_i) (\ketbra{0}^{\otimes(n-\nu)} \otimes \ketbra{\phi})) \leq \sum_{P \in \mathbb{P}_\nu} \tr^2(P \phi) = 2^\nu,
\end{equation}
where the final equality holds because  $\ket{\phi}$ is a pure $\nu$-qubit state. We then have
\begin{equation}
    \log(\sum_{i=0}^k \delta_i \tr^2(h_i \psi)) \leq \log(\sum_{i=0}^k \tr^2(h_i \psi)) = \log(\sum_{P \in \mathbb{P}_\nu} \tr^2(P \phi)) = \nu,
\end{equation}
as desired.
\end{proof}

It may at first appear that $n_A - \abs{S_A}$ is a fundamental quantity for computing the entanglement of a state. Indeed, the above lemma recovers the fact that, as shown in Ref. \cite{fattal_entanglement_2004}, the entanglement of a stabilizer state across a bipartition $A|B$ is given exactly by $n_A - \abs{S_A}$. However, for general states, the quantity $n_A - \abs{S_A}$ is not necessarily symmetric across the bipartition. That is, unlike stabilizer states, we do not necessarily have $n_A - \abs{S_A} = n_B - \abs{S_B}$. This can be seen by a simple example: consider $T_1 \ket{+}^{\otimes 2}$, where $T_1$ is the $T$-gate acting on the first qubit. This state has only one stabilizer: $X_2$. Therefore, if $A$ is the first qubit and $B$ is the second, we will have $n_A-\abs{S_A}=1$ while $n_B-\abs{S_B}=0$. This problem motivates us to introduce a closely related quantity, termed the `stabilizer entanglement', that is manifestly symmetric across the bipartition. We will show that this stabilizer entanglement not only captures the entanglement just as well as the quantity $n_A - \abs{S_A}$, but it also serves as an extremely useful entanglement monotone for restricted classes of LOCC protocols, namely stabilizer LOCC protocols.
\begin{definition}[Stabilizer entanglement]\label{def:stab-ent}
For a given bipartition $A|B$, the stabilizer entanglement $\mathcal{E}(\rho;A|B)$ is 
\begin{equation}
    \mathcal{E}(\rho;A|B) \coloneqq \frac{\abs{S_{AB}}}{2} = \frac{\abs{S} - (\abs{S_A} + \abs{S_B})}{2},
\end{equation}
where $S_{AB}$ are the generators for the subgroup $G_{AB} \cong G/(G_A \cdot G_B)$ of stabilizers that are unaccounted for by the local stabilizer groups $G_A$ and $G_B$. Note that $\rho$ may be a mixed or pure state.
\end{definition}
\begin{theorem}[Bounded approximation error]\label{thm:stab-approx}
For a pure state, the difference between $\mathcal{E}(\psi;A|B)$ and $S_\alpha(\psi_A)$ is $O(\nu)$. Specifically, the following holds for all $\alpha$:
\begin{equation}
    \mathcal{E}(\psi;A|B) - 3\nu/2 \leq S_\alpha(\psi_A) \leq \mathcal{E}(\psi;A|B) + \nu/2.
\end{equation}
For the von Neumann entropy,
\begin{equation}
    \mathcal{E}(\psi;A|B) - \nu/2 \leq S_1(\psi_A) \leq \mathcal{E}(\psi;A|B) + \nu/2.
\end{equation}
\end{theorem}
\begin{proof}
For brevity, let us define $z \coloneqq \frac{(n_A - \abs{S_A}) + (n_B - \abs{S_B})}{2} = \frac{n-\abs{S_A}-\abs{S_B}}{2}$. Using the fact that $S_\alpha(\psi_A) = S_\alpha(\psi_B)$, applying \cref{eq:naga}, we have $z - 2\nu \leq S_\alpha(\psi_A) \leq z$. Now, since $\abs{S} = n-\nu \implies z = \mathcal{E}(\psi;A|B) + \frac{\nu}{2}$, so
\begin{equation}
    \mathcal{E}(\psi;A|B) - 3\nu/2 \leq S_\alpha(\psi_A) \leq \mathcal{E}(\psi;A|B) + \nu/2.
\end{equation}
The bound for von Neumann entropy holds with identical reasoning, simply by using the tighter bound from \cref{eq:vn-tight}.
\end{proof}

\begin{definition}[Stabilizer LOCC protocols]
Stabilizer protocols are all procedures that are comprised of the following four operations (i) evolution by Clifford unitaries, (ii) inclusion of an arbitrary number of ancillae initialized in $\ket{0}$, (iii) discarding any subset of qubits, (iv) measurement in the computational basis, and (v) any combination of the first four operations conditioned on these measurement outcomes. Stabilizer LOCC protocols with respect to a bipartition $A|B$ are stabilizer protocols wherein the Clifford unitaries are allowed only to act locally on either $A$ or $B$, but no entangling unitaries are allowed across the bipartition. Classical communication is still allowed for stabilizer LOCC protocols, so operations on either side of the bipartition can be conditioned on measurement outcomes from the other side of the bipartition. In short, stabilizer LOCC protocols are to general stabilizer protocols as LOCC protocols are to general quantum channels.
\end{definition}
Note that this definition of stabilizer protocols allows for protocols conditioned on classical randomness, as $m$ random bits can be generated by including $m$ ancillae, acting with a Hadamard on each (to get $\ket{+}^{\otimes m}$), and measuring in the computational basis. 

\begin{theorem}[Monotonicity under stabilizer LOCC]\label{thm:monotone}
The stabilizer entanglement $\mathcal{E}$ is non-increasing under stabilizer LOCC operations.
\end{theorem}
\begin{proof}
It suffices to show that $\abs{S_{AB}}$ is non-increasing under (i) local Clifford unitaries, (ii) inclusion of $m$ ancillae in $\ket{0}^{\otimes m}$, (iii) partial trace, and (iv) computational basis measurements. 
\begin{enumerate}[label=(\roman*)]
    \item Consider a Clifford unitary $U = U_A \otimes \id$, which acts nontrivially only on subsystem $A$ (the proof for a Clifford acting on subsystem $B$ follows identically). By definition, for any generator $g_A \otimes g_B = g \in \mathbb{P}_n$, $U g U^\dagger = U_A g_A U_A^\dagger \otimes g_B$. Therefore, any generator which acts trivially only on subsystem $B$ must still act trivially on $B$ after applying $U_A$, and the same goes for generators which act trivially on $A$. This is to say, $\abs{S_A}$ and $\abs{S_B}$ cannot decrease, and since $U$ is a Clifford unitary, $\abs{S}$ remains unchanged. Therefore, $\abs{S_{AB}}$ cannot increase.
    \item Adding $m$ ancillae initialized in $\ket{0}^{\otimes m}$ adds $m$ generators which are single qubit $Z$ operators on each of the ancilla sites. However, since these generators act locally on either side of the bipartition, this does not change the number of generators $\abs{S_{AB}}$ spanning $A$ and $B$.
    \item Tracing out a subset of the qubits can only act to discard stabilizers of the state. 
    \item A measurement on site $Z_i$ in the computational basis can have one of three effects.
    \begin{enumerate}
        \item If $Z_i$ is already a stabilizer of $\ket{\psi}$, there is no change whatsoever.
        \item If $Z_i$ commutes with each of the stabilizers of $\ket{\psi}$ but is not itself a stabilizer of $\ket{\psi}$, it is added to the list of generators for the stabilizer group of the state, with a $\pm$ phase depending on the measurement outcome. This does not change the number of generators $\abs{S_{AB}}$ that spans $A$ and $B$.
        \item If $Z_i$ anticommutes with one of the generators of the stabilizer group for $\ket{\psi}$, the general update procedure is to pick an arbitrary anticommuting stabilizer $g_j$ and update all the other stabilizers by multiplying them with $g_j$, then replace $g_j$ with $Z_i$. To show that $\abs{S_{AB}}$ does not increase, assume we have written our generators in a form $S = S_A \cup S_B \cup S_{AB}$. There are two possibilities. First, if there is a $g_j \in S_A$ that anticommutes with $Z_i$, we will need to multiply some generators in $S_A$ and $S_{AB}$ by $g_j$, but none in $S_B$ (since all of the generators in $S_B$ commute with $Z_i$). Multiplying a generator in $S_A$ by $g_j$ does not change its locality, and neither does replacing $g_j$ by $Z_i$, so $\abs{S_{AB}}$ cannot increase. In the second case, all the $g_j \in S_A$ commute with $Z_i$, so there must be some $g_j \in S_{AB}$ that anticommutes with $Z_i$. The update step simply comprises multiplying some generators in $S_{AB}$ by $g_j$, and replacing $g_j$ by $Z_i$ -- again, since the other generating sets $S_A$ and $S_B$ are untouched, this does not increase $\abs{S_{AB}}$.
    \end{enumerate}
\end{enumerate}
\end{proof}

\section{Entanglement theory through stabilizer protocols}

In this section, we explore entanglement manipulation and detection tasks, examining the significance of stabilizer nullity in these tasks.
 We focus on three core tasks: entanglement distillation, entanglement dilution, and the witnessing of multipartite entanglement, all through the unifying framework of the stabilizer entanglement. Starting with entanglement distillation, we investigate the ability to purify entangled states from noisy states with stabilizer nullity $\nu$, crucial for quantum communication and computation. Conversely, entanglement dilution addresses the synthesis of entangled states from a set of more basic quantum resources, including shared entangled pairs and local non-Clifford resources. This discussion is framed within the constraints imposed by the nullity $\nu$, highlighting the interplay between entanglement, magic, and the operational limits of stabilizer LOCC protocols. Finally, we move away from entanglement manipulation to the mere task of verifying the presence of entanglement, known as entanglement witnessing. We develop a framework for witnessing genuine multipartite entanglement for $\nu$-compressible states, which we use to prove the robustness of entanglement of these states against noise and imperfections. We note that each of the algorithms we present below will assume a priori knowledge of the stabilizer generators $S$ of the state.  This assumption is not necessary: as we discuss in \cref{sec:unknownstates}, $S$ is readily available in almost any setting we care to consider.

\subsection{Entanglement distillation and dilution}
\parhead{Entanglement distillation.} Maximally entangled states are a powerful resource in quantum information theory, with many protocols leveraging the properties of such a state. However, noise makes it harder to synthesize a maximally entangled state without errors. Entanglement distillation is based on the idea that, given multiple copies of a mixed state, it is possible to convert the entanglement in the noisy state into a certain number of pure, useful (i.e., maximally entangled) states~\cite{horodecki_quantum_2009,khatri_principles_2020}. The specific entanglement distillation task we focus on in this work is as follows. Two parties, $A$ and $B$, are given a noisy state $\rho_{AB}$, and their goal is to use some LOCC protocol to distill as many two-qubit maximally entangled states (i.e., Bell pairs, commonly known as `ebits') as possible from $\rho_{AB}$. Given the theme of this work, we will further restrict our attention to stabilizer LOCC protocols, for which the framework of $\nu$-compressible states gives us a very rich set of tools to construct good entanglement distillation protocols. In this section, we will describe a single-shot, deterministic, and error-free entanglement distillation protocol for $\nu$-compressible states. Importantly, our algorithm requires only one copy of the input state, and does not rely on using many copies of a state and projecting into some typical set; this is motivated by practical concerns, as it has been shown that one may need up to $\sim 10^6$ copies \cite{wilde_quantum_2013} to achieve the necessary concentration for the asymptotically optimal rates found in Ref.~\cite{bennett1996concentrating}.

To start, we state a no-go result that circumscribes the optimal performance of any stabilizer LOCC protocols for entanglement distillation.
\begin{lemma}[Distillability bounds]\label{cor:bounds}
No stabilizer LOCC can distill more than $\mathcal{E}(\psi;A|B)$ Bell pairs from a state $\psi$.
\end{lemma}
\begin{proof}
Assume there were some stabilizer protocol $\mathcal{S}$ that could distill $M >\mathcal{E}(\psi;A|B)$ Bell pairs from $\psi$. Then, $\psi' \coloneqq \mathcal{S}(\psi)$ is a state with at least $2M$ stabilizers in $S_{AB}'$, as each Bell pair contributes $2$ generators $XX$ and $ZZ$ spanning the bipartition $A|B$. However, then $\mathcal{E}(\psi';A|B) \geq M > \mathcal{E}(\psi;A|B)$, which contradicts the above theorem.
\end{proof}
Remarkably, in \cref{thm:distill}, we show that we can almost saturate this upper bound. First, we need the following lemma.

\begin{lemma}\label{lem:bell}
If a state $\rho$ is stabilized by $X_i X_j$ and $Z_i Z_j$, it must be of the form $\rho = \ketbra{\phi_+}_{ij} \otimes \sigma$, where $\ket{\phi_+} = \frac{1}{\sqrt{2}} \qty(\ket{00} + \ket{11})$ and $\sigma$ is some arbitrary mixed state.
\end{lemma}
\begin{proof}
Let $\rho_A$ be the reduced density matrix of $\rho$ on sites $i$ and $j$. By assumption $\tr(\rho_A X_i X_j) = \tr(\rho_A Z_i Z_j) = 1$. There is only one density matrix that satisfies this, which is $\rho_A = \ketbra{\phi_+}$. Since $\rho_A$ is pure, it must be unentangled with the rest of the system, which implies $\rho = \rho_A \otimes \sigma = \ketbra{\phi_+} \otimes \sigma$ for some mixed state $\sigma$.
\end{proof}
\begin{theorem}[Efficient bipartite entanglement distillation]\label{thm:distill}
For any bipartite (pure or mixed) state $\rho$, there exists a stabilizer LOCC protocol $\mathcal{S}(\rho) \coloneqq (U_A \otimes U_B) \rho (U_A \otimes U_B)^\dagger$ comprised solely of two Clifford unitaries $U_A$ and $U_B$ such that
\begin{equation}
    \mathcal{S}(\rho) = \ketbra{\phi_+}^{\otimes M}_{A' \cup B'} \otimes \sigma_{E};\,\, M = \left\lfloor \mathcal{E}(\rho;A|B) - \nu/2 -S_0(\rho)/2 \right\rfloor, \label{eq:distill}
\end{equation}
where $A'$ and $B'$ are size-$M$ subsystems of $A$ and $B$, respectively, and $E\coloneqq \overline{A' \cup B'}$ is the rest of the system. Moreover, a depth $O(n)$ circuit representation of $U_A$ and $U_B$ can be found classically in $O(n^3)$ time.
\end{theorem}
\begin{proof}
For any bipartition $A | B$, the generators $G$ for the stabilizer group of any state can always be brought into a `canonical form' $G =G_A \cup G_B \cup G_{AB}$. Ref. \cite{fattal_entanglement_2004} showed that for a complete set of stabilizers, $G_{AB}$ could be written as exactly $M=\frac{\abs{G_{AB}}}{2}$ pairs $(g_k, \bar{g}_k)$ such that the projections of $g_k$ and $\bar{g}_k$ onto $A$ (or $B$) anticommute, but commute with the projections of all other generators. As seen in Lemma 3 of Ref. \cite{fattal_entanglement_2004}, the fact that there are no `unpaired' stabilizers in $S_{AB}$ for a stabilizer state follows from the fact there are exactly $n$ generators for a stabilizer state. Our case is slightly different: for a pure state $\psi$ with nullity $\nu$, there can be $n-\nu$ generators. For a mixed state, the purified state has at least $n+n^*-\nu$ generators, where the purifying system has at most $n^* = \lceil S_0(\rho) \rceil$ qubits (by definition of a mixed state). By row reduction on the tableau representation of the stabilizers, the generating set can always be written in a form where the first $2n^*$ generators act nontrivially on both the purifying system and the original system, whereas the last $n-\nu-n^*$ act trivially on the purifying system. Therefore, there are at least $n-\nu-\lceil S_0(\rho) \rceil$ generators for the stabilizer group of $\rho$. Adapting the analysis of Ref.~\cite{fattal_entanglement_2004}, we find that the number of unpaired generators $z$ obeys
\begin{subequations}\label{eq:unpair-bound}
    \begin{gather}
        \abs{S_A}+\abs{S_B}+2M+2z \leq n \\
        \abs{S_A}+\abs{S_B}+2M+z \geq n-\nu-\lceil S_0(\rho) \rceil.
    \end{gather}
\end{subequations}
This clearly gives $z \leq \nu + \lceil S_0(\rho) \rceil$. Therefore, there are at least $M \geq \mathcal{E}(\rho;A|B) - \nu/2 - \lceil S_0(\rho) \rceil/2$ `paired stabilizers' in $S_{AB}$. As we show below, these paired stabilizers are precisely those which are useful for entanglement distillation.

We label the paired stabilizers in $S_{AB}$ as $(g[i], \bar{g}[i])$ for $i=1,\ldots,M$. We will denote the restriction of $g[i]$ to subsystem $A$ as $g_A[i]$, with analogous definitions for $g_B[i],\bar{g}_A[i],\bar{g}_B[i]$. It suffices to design our Cliffords $U_A$, $U_B$ as follows. We require
\begin{subequations}\label{eq:distill-cond}
\begin{equation}
U_A g_A[i] U_A^\dagger = X_{A_i}\qc U_B g_B[i] U_B^\dagger = X_{B_i},    
\end{equation}
and for the other stabilizer in each pair, we require    
\begin{equation}
U_A \bar{g}_A[i] U_A^\dagger = Z_{A_i}\qc U_B \bar{g}_B[i] U_B^\dagger = Z_{B_i},
\end{equation}
\end{subequations}
where $A_i$ and $B_i$ are the $i$th sites of $A$ and $B$, respectively. This can always be done: the resulting stabilizers have the same commutation and anti-commutation relations as the original stabilizers. That is, $g_A[i]$ anticommutes with $\bar{g}_A[i]$, but commutes with all other stabilizers. For convenience, let us abbreviate $U = U_A \otimes U_B$. Now, to see, for instance, that $X_{A_i} X_{B_i}$ stabilizes $\rho' \coloneqq U \rho U^\dagger$, we observe that $X_{A_i} X_{B_i} = U g[i] U^\dagger$. Therefore, using the fact that $g[i] \Pi = \Pi$:
\begin{equation}
    X_{A_i} X_{B_i} \rho'  = (U g[i] U^\dagger) (U \Pi H U^\dagger) = U (g[i] \Pi H) U^\dagger = U \Pi H U^\dagger = \rho'.
\end{equation}
A similar proof holds for $Z_{A_i} Z_{B_i}$. Finally, applying \cref{lem:bell}, we see that we have distilled a product of $M$ Bell pairs on sites $(A_i \cup B_i)$, which are all unentangled with the rest of the system.

Finally, to see that these two unitaries $U_A$ and $U_B$ can be designed in $O(n^3)$ time, we simply use the fact that the constraints \cref{eq:distill-cond} can be satisfied using linear-algebraic manipulations (Gaussian elimination) using the tableau representation of the stabilizer generators. The efficiency of the two unitaries $U_A$ and $U_B$ follows from the fact that they are Clifford unitaries, which can all be implemented with $O(n^2 /\log n)$ gates and depth $O(n)$ \cite{aaronson_improved_2004}.
\end{proof}

This demonstrates how the \emph{bipartite} entanglement of a state can be distilled into Bell pairs. We can consider an even more ambitious task: distilling \emph{multipartite} entanglement. That is, if we have a state that is split across $k > 2$ parties $A_1,A_2,\ldots,A_k$, can we distill some target $k$-partite entangled state using stabilizer LOCC? To begin answering this question, we first need to establish a meaningful measure of multipartite entanglement. In Ref. \cite{fattal_entanglement_2004}, a measure of multipartite entanglement for stabilizer states was proposed. For a $k$-partition $\qty{A_1, A_2, \ldots, A_k}$, the authors first defined a local subgroup $G_{\text{loc}}$:
\begin{equation}
    G_{\text{loc}} \coloneqq \bigcup_{i=1}^k G_{\bar{A}_i},
\end{equation}
which is the union of all stabilizers that act trivially on at least one of the partitions $A_i$. As this is an Abelian subgroup of $\mathbb{P}_n$, it has a generating set $S_{\text{loc}}$ of size $|S_{\text{loc}}|=\log_2 \abs{G_{\text{loc}}}$. The multipartite entanglement measure proposed by Ref. \cite{fattal_entanglement_2004} was
\begin{equation}
    \mathcal{E}(\psi; A_1|\ldots|A_k) = \abs{S} - \abs{S_{\text{loc}}}.\label{eq:multi}
\end{equation}
It was later shown that this entanglement measure has a convenient operational interpretation: it is exactly equal to the number of GHZ states distillable from the stabilizer state. Remarkably, we find that we can naively extend \cref{eq:multi} to any state and achieve an almost identical operational interpretation. In this setting, we will assume that the input state has no stabilizers which act entirely locally on any given partition $A_i$ -- this is called the assumption of `full local rank' in \cite{bravyi_ghz_2006}. This is an easy assumption to satisfy in practice; if our state has any such local stabilizers, we can always preprocess our state by applying some local Clifford that maps these local stabilizers to a single qubit $Z$. The state then separates into a tensor product $\ket{0} \otimes \ket{\psi'}$, so we trace this qubit out, as it cannot contribute to any form of entanglement distillation anyways.

\begin{theorem}[Efficient multipartite entanglement distillation]\label{thm:multipartite}
Consider a $k$-partition $\qty{A_1, \ldots, A_k}$ and a pure state $\ket{\psi}$ with stabilizers that act purely locally on any particular $A_i$. Then, there are $k$ local Clifford unitaries $\qty{U_{A_i} \mid i=1,\ldots,k}$ such that $\ket{\psi'} \coloneqq \bigotimes_{i=1}^k U_{A_i} \ket{\psi}$ is a state with at least
\begin{equation}
    p\coloneqq \mathcal{E}(\psi; A_1|\ldots|A_k)-\nu
\end{equation}
copies of the $k$-partite GHZ states shared across all $k$ parties. Furthermore, if $k=3$ (the special case of a tripartition), $\ket{\psi'}$ can also be guaranteed to contain at least
\begin{itemize}
    \item $(\abs{S}-\abs{S_A} - p)/2 - \nu$ shared Bell pairs between $B$ and $C$,
    \item $(\abs{S}-\abs{S_B} - p)/2 - \nu$ shared Bell pairs between $A$ and $C$, and
    \item $(\abs{S}-\abs{S_C} - p)/2 - \nu$ shared Bell pairs between $A$ and $B$.
\end{itemize}
\end{theorem}
\begin{proof}
For a given state $\ket{\psi}$, we will construct a pure (exact) stabilizer state $\ket{\sigma}$ that we term the stabilizer completion of $\ket{\psi}$. We will then construct a set of $k$ local Clifford unitaries $\qty{U_{A_i} \mid i=1,\ldots,k}$ that are designed to distill multipartite entanglement from $\ket{\sigma}$. However, we will finally show that nevertheless, these unitaries work very well on $\ket{\psi}$ as well.

We first define the stabilizer completion. We take the stabilizer generators $S$ of the state and find any Pauli $P$ (which is not a stabilizer of $\ket{\psi}$) which simultaneously commutes with each of the generators, then add this Pauli to the stabilizer generators. This can be done efficiently using the tableau formalism: the kernel of $\mathcal{T}_S$ is equal to exactly all those Paulis which commute with all of the generators in $S$; from there, it is easy identify a Pauli which is not in the rowspace of $\mathcal{T}_S$ (i.e., is independent of all the existing generators). This process of expanding the stabilizer generators $S$ is repeated iteratively until we have constructed a full set of stabilizer generators $T$ of size $n$ (we note here that $S \subseteq T$ by construction). The new set of generators $T$ corresponds to some pure stabilizer state $\ket{\sigma}$. 

For pure stabilizer states, $\mathcal{E}(\sigma; A_1|\ldots|A_k) = \abs{T} - \abs{T_{\text{loc}}}$ GHZ states can be distilled from $\ket{\sigma}$ using some local Clifford unitaries $\qty{U_{A_1},\ldots,U_{A_k}}$~\cite[Theorem 3]{bravyi_ghz_2006}. In the process of building $T$ up from $S$, we added at most $n-\abs{S}$ stabilizers which act trivially on at least one of the partitions $A_i$. That is, $\abs{T_{\text{loc}}} \leq \abs{S_{\text{loc}}} + (n-\abs{S})$. Also, since $\ket{\sigma}$ is a pure stabilizer state, $\abs{T}=n$. Therefore, $\mathcal{E}(\sigma;A_1|\ldots|A_k) \geq n-\qty(\abs{S_{\text{loc}}} + (n-\abs{S}))=\abs{S} - \abs{S_{\text{loc}}}=\mathcal{E}(\psi; A_1|\ldots|A_k)$. Now, we ask what happens when we naively apply the local Clifford unitaries on the original state: $\ket{\psi'} \coloneqq \bigotimes_{i=1}^k U_{A_i} \ket{\psi}$? These unitaries acted on the stabilizer completion $\ket{\sigma}$ by mapping $k$-tuples of stabilizers $\qty{g_1, g_2, \ldots, g_k} \subseteq T$ onto the $k$ stabilizers for a $k$-partite GHZ state. However, we note that some of these stabilizers may be `phantom' stabilizers, as may not necessarily have been present in the original set of stabilizers $S$ (i.e., they were added to $T$ in the stabilizer completion process). In this case, the desired GHZ state would not have been created, and some garbage may be output. However, there are at most $\nu$ `phantom' stabilizers. Therefore, at most $\nu$ GHZ states (that would have been created if $\ket{\psi}$ were equal to $\ket{\sigma}$) will not be distilled correctly. However, the remaining GHZ states are unaffected by these `defective' GHZ states, so we then arrive at the desired claim that at least $\mathcal{E}(\psi; A_1|\ldots|A_k)-\nu$ GHZ states have been distilled with these Clifford unitaries. The claim about the simultaneous distillability of Bell pairs in the tripartite case follows identically from \cite[Theorem 5]{bravyi_ghz_2006}, and applying our reasoning about `phantom' stabilizers (i.e., the $\nu$ phantom stabilizers destroy at most $\nu$ Bell pairs).
\end{proof}

\parhead{Entanglement dilution.} Now, consider the converse task of entanglement distillation, known as entanglement dilution. The objective is to prepare an entangled target state $\ket{\psi}$ shared between two parties, Alice and Bob. Since this state is entangled, it cannot be prepared with LOCC alone: Alice and Bob must start with some amount of shared entanglement. Namely, we typically assume that they start with some number of Bell pairs shared between them. Since we will also restrict ourselves to stabilizer LOCC, we will also allow for some amount of magic resources in preparing the state. More precisely, we will allow one of the parties to hold a non-stabilizer state of a bounded size to start with.  Given these constraints, entanglement dilution can be accomplished using a number of resources depending on the nullity $\nu$ of $\ket{\psi}$. The bare minimum amount of resources one might expect is that we require $\mathcal{E}(\psi;A|B)$ shared ebits (and, if we are unreasonably optimistic, \emph{no} local magic states or classical communication). We will show that these optimistic expectations are only off-target by $O(\nu)$. 

\begin{theorem}[Efficient entanglement dilution]\label{thm:dilution}
If Bob has local access to a pure bipartite state $\ket{\psi_{AB}}$, then Alice and Bob can together prepare the state $\ket{\psi_{AB}}$ across the bipartition $A|B$ via a stabilizer LOCC protocol that uses at most
\begin{equation}
    \begin{gathered}
        \mathcal{E}(\psi;A|B) + \nu/2 \qq{shared ebits and} \\
        \nu \qq{bits of classical communication.}
    \end{gathered}
\end{equation}

\end{theorem}
\begin{proof}
Recall the output of the entanglement distillation protocol \eqref{eq:distill}:
\begin{equation}
    \mathcal{S}(\ketbra{\psi}) = \ketbra{\phi_+}^{\otimes M}_{A' \cup B'} \otimes \sigma_{E},
\end{equation}
where $E = (A \setminus A') \cup (B \setminus B')$. We will define $A'' \coloneqq A\setminus A'$ and $B'' \coloneqq B \setminus B'$ for brevity. We observe that if we are able to prepare $\sigma_E$ across $A''$ and $B''$, we will be able to prepare $\ket{\psi}$ simply by inverting $\mathcal{S}$, which comprises applying $U_A^\dagger$ and $U_B^\dagger$ locally on $A$ and $B$, respectively.

First, we develop the structure of the state $\sigma_E$. $\sigma_E$ must have at least $n-2M-\nu$ stabilizers, and from \cref{eq:unpair-bound}, there are at most $\nu$ stabilizers that act nontrivially on $A''$ and $B''$ (these are the `unpaired' stabilizers in $G_{AB}$ that were useless for entanglement dilution). Therefore, there are at least $n-M-2\nu$ stabilizers which act entirely locally on $A''$ or $B''$. Using nullity distillation, we can find two Clifford unitaries $V_{A''}$ and $V_{B''}$ such that $\qty(V_{A''} \otimes V_{B''}) \sigma_E \qty(V_{A''} \otimes V_{B''})^\dagger = \ketbra{0}^{\otimes(n-2M-2\nu)} \otimes \sigma'$, where $\sigma'$ is a state spanning $A$ and $B$ on at most $2\nu$ qubits.
\begin{figure}[H]
    \centering
    {
        \graphicspath{{figures/}}
        \newcommand{\fs}{\scriptsize}
        \newcommand{\fss}{\footnotesize}
    \def\svgwidth{0.5\textwidth}
    \input{figures/distill.pdf_tex}
    }
    \caption{State preparation (green), followed by entanglement distillation (blue), and finally `nullity distillation' (purple). The nullity distillation Clifford unitaries $V_{A''},V_{B''}$ produce $\ket{0}$ states and a state $\ket{\sigma'}$ on at most $2\nu$ qubits.}
    \label{fig:dilute}
\end{figure}
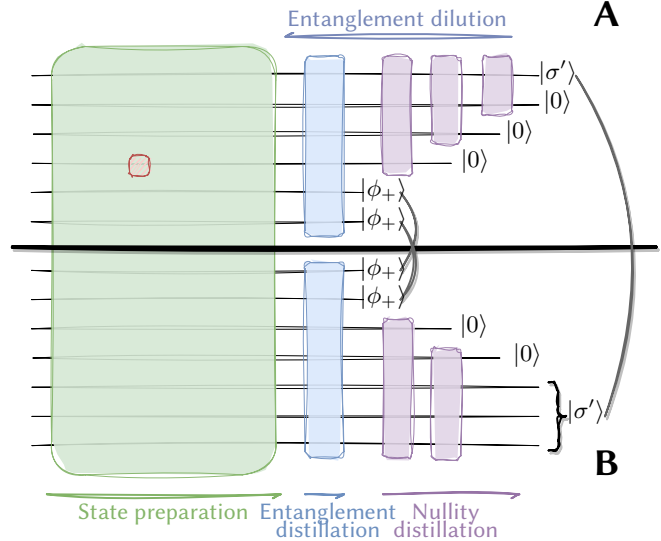
Our entanglement dilution protocol is then as follows. Assume WLOG that $\ket{\sigma'}$ is supported on more qubits in $B$ than in $A$ (so, it acts on at most $\nu$ sites in $A$). Then, we first prepare $\ket{\sigma'}$ completely locally on $B$, simply by using the procedure shown in \cref{fig:dilute}. Specifically, Bob uses his copy of $\ket{\psi}$ locally. Having prepared $\ket{\psi}$, Bob then executes the rest of the protocol entirely locally, which results in some number of local Bell pairs (which Bob can throw out) and some $\ket{0}$ states, which he also throws out. He is left with the state $\ket{\sigma'}$. Having prepared $\ket{\sigma'}$, $B$ then sends $A$'s portion of $\ket{\sigma'}$ using state teleportation, which costs at most $\nu$ ebits and $\nu$ bits of classical communication. Having received $\ket{\sigma'}$, and assuming $A$ and $B$ share $M$ additional ebits, the nullity distillation (purple) and entanglement distillation (blue) unitaries in \cref{fig:dilute} can be inverted. With this, $A$ and $B$ now share $\ket{\psi}$.
\end{proof}

\subsection{Witnessing multipartite entanglement}\label{subsec:witnessingent}
An entanglement witness $\mathcal{W}$ is an observable defined with respect to a target state $\psi$. The purpose of a witness is to experimentally validate the presence of genuine entanglement in an imperfectly prepared version of the target state $\psi$, which we call $\rho$. Mathematically, an observable $\mathcal{W}$ serves as an entanglement witness with respect to a particular bipartition $A|B$ when
\begin{equation}
    \tr(\mathcal{W} \rho) < 0 \implies \rho \not\in \sepp{A}{B}, \label{eq:witness-crit}
\end{equation}
where $\sepp{A}{B}$ is the set of states that is separable with respect to the bipartition $A|B$. More precisely, $\sepp{A}{B}$ is the convex hull of the set of unentangled pure states
\begin{equation}
    \qty{\ket{\psi_A} \otimes \ket{\psi_B} \mid \ket{\psi_A} \in \mathscr{H}_A, \ket{\psi_B} \in \mathscr{H}_B}.
\end{equation}
A typical form for an entanglement witness is
\begin{equation}
    \mathcal{W} = \alpha \id - \ketbra{\psi}\qc \alpha \coloneqq \max_{\rho \in \sepp{A}{B}} \tr(\rho \ketbra{\psi}).
\end{equation}
This observable satisfies \cref{eq:witness-crit} by definition of $\alpha$. The usefulness of this entanglement witness is measured by $\alpha$: for instance, in the extreme case $\alpha=1$, the value of $\tr(\mathcal{W}\rho)$ does not tell us anything about the entanglement of $\rho$, as $\tr(\mathcal{W} \rho)$ will then always be positive.

In this section, we will aim to develop an even more useful tool: a singular entanglement witness that simultaneously detects entanglement across a large number of bipartitions $\mathcal{B} = \qty{(A_i,B_i) \mid i=1,\ldots,\abs{\mathcal{B}}}$. That is, 
\begin{equation}
    \tr(\mathcal{W} \rho) < 0 \implies \rho \not\in \sepp{A_i}{B_i}\quad  \forall i=1,\ldots,\abs{\mathcal{B}}.
\end{equation}
We call this a witness of genuine \emph{multipartite} entanglement. In fact, we will construct something even stronger than this: we will define a class of entanglement witnesses that is not only able to rule out separability, but can rule out `approximate' separability. To be precise, for any fixed level of entanglement $E$, we define the set of $E$-entangled states $\sep^{(E)}(\mathscr{H}_A; \mathscr{H}_B)$ as the convex hull of the set of pure states $\mathcal{X}_E \coloneqq \qty{\ketbra{\psi} \mid S_{1/2}(\psi_A) \leq E}$, noting that $\sep^{(E=0)}(\mathscr{H}_A; \mathscr{H}_B)=\sepp{A}{B}$ recovers the original set of exactly separable states. Moreover, it holds that $\sep^{(E)}(\mathscr{H}_A; \mathscr{H}_B)\supset\sepp{A}{B}$. We will say that a state $\rho$ exhibits genuine multipartite $E$-entanglement with respect to a set of bipartitions $\mathcal{B}$ if $\rho \not\in \sep^{(E)}(\mathscr{H}_{A_i}; \mathscr{H}_{B_i})$ for all $i=1,\ldots,\abs{\mathcal{B}}$. The main result of this section will be a simple and efficient construction for a multipartite $E$-entanglement witness $\mathcal{W}^{(E)}$, satisfying
\begin{equation}
    \tr(\mathcal{W}^{(E)} \rho) < 0 \implies \rho \not\in \sep^{(E)}(\mathscr{H}_{A_i}; \mathscr{H}_{B_i})\quad \forall i=1,\ldots,\abs{\mathcal{B}}.
\end{equation}
Again, we show that the stabilizer component $\Pi$ of the target state $\psi$ contains almost all the relevant information for this multipartite entanglement detection. However, before this, we first need the following lemma.
\begin{lemma}\label{lem:overlap-bound}
For a given bipartition $A|B$ where both halves of the system have size $M$, and for some level of entanglement $E$,
\begin{equation}
    \max_{\rho \in \sep^{(E)}(\mathscr{H}_A; \mathscr{H}_B)} \tr(\rho \ketbra{\phi_+}^{\otimes M}) \leq 2^{-M+E},
\end{equation}
where $\ket{\phi_+}$ is a Bell pair spanning $A$ and $B$.
\end{lemma}
\begin{proof}
First, we observe that we only need to do our maximization over states $\ket{\psi_{AB}} \in \mathcal{X}_E$, as $\tr(\rho \ketbra{\phi_+}^{\otimes M})$ is linear in $\rho$, and $\sep^{(E)}(\mathscr{H}_A; \mathscr{H}_B)$ is simply the convex hull of $\mathcal{X}_E$. We first write the optimization over all states in $\mathcal{X}_E$ by doing a trivial transformation, wherein we maximize first over all reduced density matrices $\psi_A$ that have bounded entropy, and then maximize over all possible purifications $\ket{\psi_{AB}}$ of that state.
\begin{equation}
    \max_{\ket{\psi_{AB}} \in \mathcal{X}_E} \abs{\braket{\psi_{AB}}{\phi_+}^{\otimes M}}^2 = \max_{\substack{\psi_A \\ S_{1/2}(\psi_A) \leq E}}  \quad \max_{\substack{\ket{\psi_{AB}} \\ \tr_B(\ketbra{\psi_{AB}}) = \psi_A}} \abs{\braket{\psi_{AB}}{\phi_+}^{\otimes M}}^2.
\end{equation}
Note that the inner optimization is known from Uhlmann's theorem to be equal exactly to $\mathcal{F}(\tr_B(\ketbra{\phi_+}^{\otimes M}), \psi_A)$, where $\mathcal{F}(\rho,\sigma) \coloneqq \tr(\sqrt{\sqrt{\rho} \sigma \sqrt{\rho}})^2$ is the fidelity. Since $\tr_B(\ketbra{\phi_+}^{\otimes M}) = 2^{-M} \id$,
\begin{equation}
    \max_{\ket{\psi_{AB}} \in \mathcal{X}_E} \abs{\braket{\psi_{AB}}{\phi_+}^{\otimes M}}^2 = \max_{\substack{\psi_A \\ S_{1/2}(\psi_A) \leq E}} 2^{-M} \tr(\sqrt{\psi_A})^2.
\end{equation}
Since $\tr(\sqrt{\psi_A})^2 = \exp(S_{1/2}(\psi_A))$, the maximum possible value is at most $2^E$ by assumption that $S_{1/2}(\psi_A) \leq E$. Therefore, the overall maximum is bounded by $2^{-M+E}$.
\end{proof}

\begin{theorem}[Multipartite $E$-entanglement witness]\label{th:multipartiteEentwitness}
Consider a set of bipartitions $\mathcal{B} = \qty{(A_i,B_i) \mid i=1,\ldots,\abs{\mathcal{B}}}$, a pure state $\psi$, and some level of entanglement $E$. Recall, from \cref{thm:pure-struct}, that $\Pi$ is the stabilizer projector of the target state $\psi$: $\Pi = \prod_{g_j \in S} \qty(\frac{\id+g_j}{2})$, where $S$ is a generating set for the stabilizer group of $\psi$. The following observable is a multipartite $E$-entanglement witness with respect to the bipartitions $\mathcal{B}$:
\begin{equation}
    \mathcal{W}^{(E)} = 2^{-M(\mathcal{B})+ E} \id - \Pi\qc M(\mathcal{B}) \coloneqq \min_i \lfloor\mathcal{E}(\psi;A_i|B_i) - \nu/2\rfloor,
\end{equation}
where $\mathcal{E}(\psi;A_i|B_i)$ is the stabilizer entanglement with respect to the bipartition $A_i|B_i$. When $E=0$, we recover the conventional definition of a multipartite entanglement witness.
\end{theorem}
\begin{proof}
Let us fix a particular bipartition $A | B$. It suffices to show that 
\begin{equation}
    \max_{\rho \in \sep^{(E)}(\mathscr{H}_A, \mathscr{H}_B)} \tr(\Pi \rho) \leq 2^{-M+\lceil E \rceil},
\end{equation}
where we abbreviate $M(\mathcal{B})$ to just $M$ for brevity.

We can freely transform $\Pi$ by a product of local unitaries without changing the maximum value of the above, as this is equivalent to transforming $\rho$ by local unitaries (which leaves $\sep^{(E)}(\mathscr{H}_A, \mathscr{H}_B)$ invariant), and the maximum is over all possible $\rho \in \sep^{(E)}(\mathscr{H}_A, \mathscr{H}_B)$ anyways. Therefore, we use the machinery we developed in entanglement distillation, wherein we showed that $\Pi$ is equivalent, under a local unitary $U_A \otimes U_B$, to a projector onto a stabilizer group with generators $S_A \cup S_B \cup S_{AB}$. Furthermore, the generators in $S_{AB}$ contained at least $M$ pairs $X_{A_i} X_{B_i}, Z_{A_i}, Z_{B_i}$. Denoting $\tilde{S}_{AB}$ to be the remaining unpaired generators in $S_{AB}$, we have:
\begin{equation}
    (U_A \otimes U_B)^\dagger \Pi (U_A \otimes U_B) = \qty[\prod_{g \in S_A \cup S_B \cup \tilde{S}_{AB}} \frac{\id+g}{2}] \qty[\prod_{i=1}^M \frac{\id+X_{A_i} X_{B_i}}{2}] \qty[\prod_{i=1}^M \frac{\id+Z_{A_i} Z_{B_i}}{2}].
\end{equation}
Since for any two projectors $\Pi_1$ and $\Pi_2$, $\tr(\Pi_1 \Pi_2 \rho) \leq \tr(\Pi_2 \rho)$, we simply drop the first projector (containing terms in $S_A,S_B,\tilde{S}_{AB}$) and find
\begin{subequations}
    \begin{align}
    \max_{\rho \in \sep^{(E)}(\mathscr{H}_A, \mathscr{H}_B)} \tr(\Pi \rho) &\leq \max_{\rho \in \sep^{(E)}(\mathscr{H}_A, \mathscr{H}_B)} \tr(\qty[\prod_{i=1}^M \frac{\id+X_{A_i} X_{B_i}}{2}] \qty[\prod_{i=1}^M \frac{\id+Z_{A_i} Z_{B_i}}{2}] \rho).
        \intertext{Denoting the first $M$ sites of $A$ as $A'$ (with a similar definition for $B'$), we see that $\qty[\prod_{i=1}^M \frac{\id+X_{A_i} X_{B_i}}{2}] \qty[\prod_{i=1}^M \frac{\id+Z_{A_i} Z_{B_i}}{2}] = \ketbra{\phi_+}^{\otimes M}_{A' \cup B'} \otimes \id_{E}$, where $\phi_+$ is a Bell pair spanning the bipartition $A|B$ and $E=\overline{(A' \cup B')}$.}
        &\leq \max_{\rho' \in \sep^{(E)}(\mathscr{H}_{A'}, \mathscr{H}_{B'})} \tr(\ketbra{\phi_+}^{\otimes M} \rho')\\
        &\leq 2^{-M+E},
    \end{align}
\end{subequations}
where the last bound follows from \cref{lem:overlap-bound}.
\end{proof}

\begin{corollary}[Robustness of multipartite entanglement]\label{cor:robustnessmultipartite}
For any set of bipartitions $\mathcal{B}$ and any target pure state $\ketbra{\psi}$, if an imperfectly prepared state $\rho$ satisfies
\begin{equation}
    \frac{1}{2} \norm{\rho-\ketbra{\psi}}_1 < 1-2^{-M(\mathcal{B})/2},\label{eq:robustness}
\end{equation}
then $\rho$ exhibits genuine multipartite entanglement with respect to $\mathcal{B}$.
\end{corollary}
\begin{proof}
The condition \cref{eq:robustness} implies $\tr(\rho \ketbra{\psi}) > 2^{-M}$. It then suffices to show that $\tr(\rho \ketbra{\psi}) > 2^{-M} \implies \tr(\rho \Pi) > 2^{-M}$, as we showed that $\tr(\rho \Pi) \leq 2^{-M}$ for all separable $\rho$. To show this, we show the contrapositive. Assume $\tr(\rho \Pi) \leq 2^{-M}$. Defining $\rho' \coloneqq \Pi \rho \Pi$, we see that $\tr(\rho') = \tr(\rho \Pi) \leq 2^{-M} \implies \norm{\rho'}_\infty \leq 2^{-M}$, so $\tr(\rho \ketbra{\psi})=\tr(\rho' \ketbra{\psi}) \leq \norm{\rho'}_\infty \leq 2^{-M}$.
\end{proof}

\begin{corollary}[Robustness of multipartite $E$-entanglement]\label{cor:Erobustnessmultipartite}
For any set of bipartitions $\mathcal{B}$ and any target pure state $\ketbra{\psi}$, if an imperfectly prepared state $\rho$ satisfies
\begin{equation}
    \frac{1}{2} \norm{\rho - \ketbra{\psi}}_1 < 1-2^{(-M(\mathcal{B})+E)/2},
\end{equation}
then $\rho$ exhibits genuine multipartite $E$-entanglement with respect to $\mathcal{B}$.
\begin{proof}
The proof is identical to the one of \cref{cor:robustnessmultipartite} and follows entirely from \cref{lem:overlap-bound} and \cref{th:multipartiteEentwitness}.
\end{proof}
\end{corollary}

\begin{lemma}[Efficient measurement protocol]\label{lem:efficienwitnessing}
Assume we have target pure state $\ketbra{\psi}$. For any state $\rho$ that is bounded in trace distance from $\ket{\psi}$ (i.e., $\frac{1}{2} \norm{\rho - \ketbra{\psi}}_1 \leq \epsilon$), and some maximum failure probability $\delta$,
\begin{equation}
    N = \left\lceil\frac{2 \log(2/\delta)}{(1-(\epsilon + 2^{-M(\mathcal{B})}))^2}\right\rceil
\end{equation}
Pauli measurements suffice to detect genuine multipartite entanglement with respect to the bipartitions $\mathcal{B}$.
\end{lemma}
\begin{proof}
If $\frac{1}{2} \norm{\rho - \ketbra{\psi}}_1 \leq \epsilon$, we have $\abs{\tr(\mathcal{W} \ketbra{\psi}) - \tr(\mathcal{W} \rho)} = \abs{\tr(\Pi \ketbra{\psi}) - \tr(\Pi \rho)} \leq \epsilon$. Since $\tr(\Pi \ketbra{\psi}) = \tr(\Pi^2 H) = \tr(\Pi H) = \tr(\ketbra{\psi}) = 1$, we must then have $\tr(\Pi \rho) \geq 1- \epsilon$. The required resolution $\Delta$ to which we need to measure the entanglement witness, in order to rule out biseparability across any of the bipartitions in $\mathcal{B}$, is $\Delta = 1-(\epsilon+2^{-M(\mathcal{B})})$. 

To estimate $\tr(\Pi \rho)$, we use the fact that $\Pi$ can be expanded as a linear combination $\abs{S}^{-1} \sum_{P_j \in G} P_j$. We can construct an unbiased estimator $\tilde{\Pi}_j$ for $\tr(\Pi \rho)$ by measuring a random Pauli observable $P_j \in G$ for $\rho$. We note that randomly selecting $P_j \in G$ can be done simply by choosing a subset of generators from $S$, and multiplying them together. Equivalently, $P_j = \prod_{i=1}^{\abs{S}} g_i^{x_i}$, where $x_i \sim \text{Bernoulli}(1/2)$. Finally, since $-1 \leq \tilde{\Pi}_j \leq 1$, we can use Hoeffding's inequality to find that the estimator $\tilde{\Pi} \coloneqq N^{-1} \sum_{j=1}^N \tilde{\Pi}_j$ obeys
\begin{equation}
    \Pr[\abs*{\tilde{\Pi} - \tr(\Pi \rho)} \geq \Delta] \leq 2\exp(-\frac{N \Delta^2}{2}),
\end{equation}
from which we can conclude that $N = \left\lceil\frac{2 \log(2/\delta)}{\Delta^2}\right\rceil$ samples suffices to limit the failure probability to $\delta$.
\end{proof}

\section{The entanglement-dominated phase vs. magic-dominated phase}\label{sec:ed-vs-md}

We now apply all our findings to characterize the computational separation between entanglement-dominated states and magic-dominated states. To achieve this, let us examine the separation from the bottom-up. The main results of the previous section (e.g., \cref{thm:stab-approx,thm:distill,thm:dilution,thm:multipartite,th:multipartiteEentwitness}) have a common formal similarity: for almost all entanglement characterization and manipulation tasks, the stabilizer entanglement $\mathcal{E}$ can be used to replace the true entanglement entropy $S_\alpha$, up to an additive correction $\propto \nu$. This simple fact proves to be extremely useful, as everything about the stabilizer entanglement (through the generating set $S_{AB}$) can be computed and manipulated using a polynomial amount of classical or quantum resources. Of course, this comes with the drawback that we incur an error or inefficiency (for entanglement characterization or manipulation, respectively) of size $O(\nu)$. This led us to draw a fairly natural distinction between the two classes, or phases, of states, ED and MD. We find an operational distinction between these two phases. As we will show in this section, if we restrict ourselves to stabilizer protocols, we can achieve essentially optimal performance for ED states on a number of entanglement related tasks. Conversely, there are almost always clear counterexamples that prove stabilizer protocols are completely ineffective for entanglement related tasks on MD states. We should qualify this claim: we are only able to exhibit counterexamples in the magic-dominated phase for any $\nu=\omega(\log n)$. For MD states with $\nu=O(\log n)$ or $S_1=O(\log n)$, we are unable to rule out the possibility of efficient protocols for the simple reason that these states are both efficiently learnable and classically simulable. Indeed, it is reasonable to expect that for these states, almost any entanglement-related task can be done efficiently simply by brute force inspection of the classical description of the state. To begin with, let us provide a characterization of the entanglement-dominated phase in terms of the stabilizer entanglement (see \cref{def:stab-ent}), which proves to be instrumental in the efficient detection of ED states, as demonstrated later in \cref{th:distinction}.
\begin{corollary}[Characterization of the ED phase] Let $A|B$ be a bipartition. A pure state $\ket{\psi}$ is entanglement-dominated if and only if $\mathcal{E}(\psi;A|B)=\omega(\nu)$.
\begin{proof}
Recalling \cref{def:ent-dom}, ED states are defined by the condition $S_{1}(\psi_A)=\omega(\nu)$. Then, from \cref{thm:stab-approx}, it is immediate to see that an analogous condition is $\mathcal{E}(\psi;A|B)=\omega(\nu)$.
\end{proof}
\end{corollary}

To appreciate the significance of the ED phase, we start with an example that illustrates the extreme robustness of entanglement within the entanglement-dominated phase.

\begin{example}[Robustness of entanglement for entanglement-dominated states]
Consider a volume-law ED state $\ket{\psi}$. Let us fix a cut $A|B$ with $n_A,n_B=\Omega(n)$. By definition, this means that for some constant $\kappa>0$ and a large enough $n$, $S_{1}(\psi_A) \geq \kappa n$. For any (possibly mixed) state $\rho$, if
\be
\frac{1}{2}\|\rho-\ketbra{\psi}{\psi}\|_1<1-o\left(\frac{1}{\poly n}\right) \implies \rho\not\in\sep^{(E)}(\mathscr{H}_{A}; \mathscr{H}_{B}),
\ee
for any $E=\kappa n-o(n)$. This is to say, $\rho$ must also be a volume-law entangled state.

\begin{proof}
The proof is a straightforward consequence of \cref{cor:Erobustnessmultipartite}. Since $\mathcal{E}(\psi;A|B)\ge \kappa n-\nu/2$, so $M \ge \kappa n-\nu$. Let $\delta$ be any function satisfying $\nu+\omega(\log n)<\delta<o(n)$ (since $\nu=o(n)$, this is always possible). Then, choosing $E=\kappa n - \delta$, we then have $M-E > \omega(\log n)$. In virtue of \cref{cor:Erobustnessmultipartite}, it follows that:
\be
\frac{1}{2} \norm{\rho-\ketbra{\psi}{\psi}}_1<1-o\qty(\frac{1}{\poly n}) \implies \rho\not\in\sep^{(E)}(\mathscr{H}_{A}; \mathscr{H}_{B}),\,\, E= \kappa n-o(n),
\ee 
which concludes the proof.
\end{proof}
\end{example}
In practice, this example means that for any convex decomposition of $\rho=\sum_{i}p_i\ketbra{\phi_i}{\phi_i}$, there exists at least one state $\ket{\phi_i}$ whose reduced density matrix $\phi_{i,A}$ satisfies $S_{1/2}(\phi_{i,A})\ge\kappa n-o(n)$. In other words, volume law entanglement-dominated states feature robust bipartite entanglement, as one can draw a ball of \textit{almost maximal radius} (i.e., close to $1$ up to a quickly decaying factor) around $\ket{\psi}$, inside which every state has volume-law entanglement.

In what follows, we will show that almost any entanglement-related task can be solved efficiently and is nearly optimal for any state in the ED phase. Given the desirable features of these ED states, we anticipate that one might wonder whether it is possible to detect whether a particular state is ED or not. We resolve this question later in \cref{th:distinction}, where we show that indeed we can verify whether an unknown state is in the ED phase or not. 

\medskip 
\parhead{Entanglement characterization in the ED phase.} A common task in entanglement characterization is checking the scaling of a state's entanglement. For instance, volume law states are characterized by an entanglement that scales with $\Theta(n_A)$. This task of \emph{entanglement characterization} can then be formalized as follows. Fix a bipartition $A|B$, and consider an entanglement class $f(n_A)$. Given a state $\ket{\psi}$, the aim is simply to check whether $S_\alpha(\psi_A) = \Theta(f(n_A))$. In general, this is not possible to do efficiently. Consider for instance setting $\alpha=2$ and $f(n_A) \sim n_A$ (which simply means we are checking whether a state is volume-law entangled or not). Without any prior knowledge about the state $\psi$, this necessitates checking whether $\Tr(\psi_A^2) = \exp(-\Omega(n_A))$. Although $\Tr(\psi_A^2)$ can be estimated using just two copies of $\psi$ and a simple swap test, exponentially many (in $n_A$) repeated measurements are necessary to resolve $\Tr(\psi_A^2)$ to the required level of accuracy to determine whether it is exponentially small. For anything but a small bipartition $n_A$, this quickly becomes infeasible. In contrast, we show below that for entanglement-dominated states, the task of entanglement characterization comes easy. 

\begin{lemma}[Efficient entanglement characterization]\label{eq:detectability}
Let us fix a bipartition $A|B$, and take an entanglement class $f(n_A)$. For any ED state $\psi$, we can check whether $S_\alpha(\psi_A) = \Theta(f(n_A))$ using polynomially many copies of $\psi$ and polynomial classical computing resources. Moreover, if indeed $S_\alpha(\psi_A) = c f(n_A) + o(f(n_A))$ for some $c$, we can estimate the coefficient of proportionality $c$ with an asymptotically vanishing error as $n_A \to \infty$.
\end{lemma}
\begin{proof}
This follows simply by combining two facts. First, the stabilizer entanglement $\mathcal{E}(\psi; A|B)$ is measurable with polynomially many copies of $\psi$ and polynomial classical computing resources (see \cref{thm:stab-learn}). Second, by \cref{thm:stab-approx}, the difference between stabilizer entanglement and entanglement entropy is at most $O(\nu)$. Since $\mathcal{E}(\psi;A|B) = \omega(\nu)$ by assumption, see \cref{def:ent-dom}, we have that $S_\alpha(\psi_A) = \Theta(f(n_A)) \iff \mathcal{E}(\psi;A|B) = \Theta(f(n_A))$. That is, $\mathcal{E}(\psi;A|B) = \Theta(f(n_A))$ is both necessary and sufficient to conclude that $S_\alpha(\psi_A) = \Theta(f(n_A))$ as well.

To show the second claim, we again apply \cref{thm:stab-approx} to conclude that 
\begin{equation}
    \abs{\frac{S_\alpha(\psi_A)}{f(n_A)} - \frac{\mathcal{E}(\psi;A|B)}{f(n_A)}} = \frac{o(f(n_A))}{f(n_A)} = o(1),
\end{equation}
so that $\frac{\mathcal{E}(\psi;A|B)}{f(n_A)}$ serves as an estimate for $c$ with asymptotically vanishing error.
\end{proof}
An immediate corollary is that we can easily distinguish between volume law entanglement and sub-volume law entanglement. We simply set $f(n_A) = n_A$ and calculate $c$ using the machinery of the above proof. When $\nu=o(n_A)$, we will have that $c=\Theta(1)$ for volume-law states and $c=o(1)$ for sub-volume law states. This sharp divide allows us to quickly determine whether a state is volume-law or not.

A considerably simpler entanglement characterization task is merely witnessing entanglement. In the ED phase, the entanglement witnessing becomes effective and efficient. This efficiency stems from \cref{th:multipartiteEentwitness}, which is extremely versatile and can be applied to the following tasks: bipartite entanglement witnessing, bipartite $E$-entanglement witnessing, multipartite entanglement witnessing, and multipartite $E$-entanglement witnessing. In all these tasks, as shown in \cref{lem:efficienwitnessing}, the entanglement can be witnessed simply by measuring the projector $\Pi$ associated with the stabilizer group $G$ of the state $\ket{\psi}$. Indeed, for some state $\rho$ that is an imperfectly prepared version of the state $\ket{\psi}$, it is sufficient to measure the overlap between $\Pi$ and the state $\rho$. If this overlap is larger than the number $2^{-M(\mathcal{B})+E}$, then we can conclude the state $\rho$ must have entanglement at least $S_{1/2} \geq E$.  Here, $M(\mathcal{B}):=\min_{i}[\mathcal{E}(\psi;A_i|B_i)-\nu/2]$ encodes the entanglement of $\ket{\psi}$ on the (possibly multipartite) partition $\mathcal{B}=\{(A_i,B_i)\,|\, i=1,\ldots, |\mathcal{B}|\}$ and $E$ is the level of entanglement of the set of states within $\mbox{SEP}^{(E)}$ defined in \cref{subsec:witnessingent}.

In the ED phase, the entanglement witness is extremely useful, and it is worth giving an explicit example below. Although we focus on the case of bipartite entanglement, the generalization to the multipartite case is straightforward.

\begin{example}
Consider a bipartition $A|B$, and let $\ket{\psi}$ be an ED state. By \cref{def:ent-dom}, one thus has $\mathcal{E}(\psi;A|B)=\omega(\nu)$ and therefore $M(\mathcal{B})=\mathcal{E}(\psi;A|B)(1-o(1))\ge\mathcal{E}(\psi;A|B)/2$. Since estimating the overlap $\tr(\Pi\rho)$ that can be done efficiently according to \cref{lem:efficienwitnessing}, one can efficiently witness any level of entanglement $E$ up to a threshold determined by the stabilizer entanglement $\mathcal{E}(\psi;A|B)$. In formula:
\begin{equation}
\tr(\Pi\rho) > 2^{E-\mathcal{E}(\psi;A|B)/2} \implies \rho\not\in \sep^{(E)}(\mathscr{H}_{A}; \mathscr{H}_{B}).\label{eq.condition}
\end{equation}
In other words, if the stabilizer entanglement $\mathcal{E}(\psi;A|B)$ is known, depending on the measured value of $\tr(\Pi\rho)$, one can rule out the possibility that $\rho$ has entanglement (as measured by $S_{1/2}$) less than $E$. The multipartite case follows easily by applying identical reasoning with $\min_i\mathcal{E}(\psi;A_i|B_i)$, that is, the minimum stabilizer entanglement across each cut $A_i|B_i$.
\end{example}

Despite our positive results for entanglement witnessing in the ED phase, unlike other entanglement-related tasks, we cannot prove any general hardness results for the MD phase. With that being said, note that in the MD phase, the entanglement witness presented in \cref{subsec:witnessingent} becomes completely ineffective: $\mathcal{M}(\mathcal{B})\le 0$ by definition of MD states. For any level of entanglement $E$, the relation $\tr(\Pi\rho)\le 1$ is always trivially satisfied and does not carry any useful information. 

\medskip 
\parhead{Entanglement manipulation in the entanglement-dominated phase.} Similar to entanglement characterization, entanglement manipulation is extremely effective for ED states.
\begin{lemma}[Optimal distillation]\label{lem:phasetransdistill}
For any state $\ket{\psi}$ in the ED phase, there exists a one-shot and zero error stabilizer LOCC protocol that distills a number $M_+$ of Bell pairs that satisfies
\begin{equation}
    \frac{M_+}{S_1(\psi_A)} = 1 - o(1).
\end{equation}
\end{lemma}
\begin{proof}
We have already shown in \cref{thm:distill} that there always exists a stabilizer protocol that distills $M_+ = \lfloor \mathcal{E}(\psi;A|B) - \nu/2 \rfloor$ Bell pairs from any state. Using \cref{thm:stab-approx}, we see that $M_+ \geq \lfloor S_1(\psi_A) - 2\nu \rfloor \geq S_1(\psi_A) - 2\nu-1$. Since $\frac{\nu}{S_1(\psi_A)} = o(1)$ by definition for ED states, we have
\begin{equation}
    \frac{M_+}{S_1(\psi_A)} = 1 - o(1)
\end{equation}
in the ED phase.
\end{proof}
One can easily show a simple converse claim for the above lemma. Namely, there exist MD states $\ket{\psi}$ for which no stabilizer protocols can distill \emph{any} Bell pairs. This begins to suggest a computational phase transition between ED and MD states, but still relies on the restriction to \emph{stabilizer} LOCC protocols. Later, in \cref{thm:no-go-magic-dom}, we show a much stronger converse result that holds for \emph{all} LOCC protocols. For now, the proof idea for the simpler claim is as follows: if the circuit used to prepare $\ket{\psi}$ uses non-Clifford resources which destroy all the stabilizers in $S_{AB}$, no entanglement can be distilled via stabilizer LOCC. Concretely, consider a circuit where a Haar random unitary acts on $\min(n_A,n_B)<s<n$ qubits on both sides of the bipartition. This is an ensemble consisting almost entirely of MD states, as $\nu=\Theta(s)$ for all but a measure-zero set of states in this ensemble, while the entanglement is also $\Theta(s)$ with probability $1-o(1)$. Also, $\abs{S_{AB}}=0$ for all but a measure-zero set of states $\psi$ in this ensemble, while $S_1(\psi_A) \geq 1$ asymptotically almost surely. This is all to say, that there must exist some MD state $\ket{\psi}$ in the ensemble with $\abs{S_{AB}}=0$ but $S_1(\psi_A) \geq 1$. Since no stabilizer protocol can distill more than $\frac{\abs{S_{AB}}}{2}$ Bell pairs from $\ket{\psi}$ (\cref{cor:bounds}), there are no distillable Bell pairs from $\ket{\psi}$.

\begin{lemma}[Optimal dilution]\label{lem:phasetransdilu}
For any state $\ket{\psi}$ in the ED phase, there exists a stabilizer LOCC for dilution that achieves
\begin{equation}
    \frac{M_-}{S_1(\psi_A)} = 1 + o(1),
\end{equation}
where $M_-$ is the required number of Bell pairs.  Furthermore, the required number of bits of classical communication $N_{\mathrm{CC}}$ and local magic resource states $N_{\mathrm{magic}}$ both satisfy
\begin{equation}
    \frac{N_{\mathrm{CC}}}{S_1(\psi_A)},\, \frac{N_{\mathrm{magic}}}{S_1(\psi_A)} = o(1).
\end{equation}
\end{lemma}
\begin{proof}
    The proof follows straightforwardly from \cref{thm:dilution} and the definition of ED states (\cref{def:ent-dom}).
\end{proof}

\parhead{Hardness results in the magic-dominated phase.} While in previous sections, we showed the feasibility of basically every entanglement task in the ED phase, this section shows the converse: namely, entanglement tasks become hard in the MD phase. Specifically, we show that there is no sample-efficient algorithm for entanglement estimation in the MD phase, and furthermore that there is no distillation protocol that can distill even a constant fraction of the entanglement into Bell pairs. Remarkably, while for the ED phase, one-shot and zero-error stabilizer LOCC protocols have near-optimal performance, below we rule out any sample-efficient LOCC protocol for the MD phase, even those that are allowed multiple rounds of classical communication or multiple copies of the input state. These findings suggest an information-theoretic separation between the two phases. We note that states with $\nu=O(\log n)$ can be learned efficiently, and therefore the computation of the entanglement entropy for these states is efficient. This means that our hardness results, while they rule out entanglement manipulation protocols that work well on \emph{all} MD states, do not necessarily rule out efficient protocols which are specifically tailored to work on the class of MD states with $\nu=O(\log n)$.
\begin{lemma}[Hardness of entanglement tasks for the magic-dominated phase]\label{thm:no-go-magic-dom}
Let $A|B$ be any bipartition satisfying $n_A,n_B=\omega(\log n)$. Let $S_1(\psi_A)=\omega(\log n)$.
\begin{enumerate}
    \item There is no sample-efficient protocol which can estimate $S_1(\psi_A)$ to within $o(1)$ relative error for all MD states $\psi$.
    \item There is no efficient LOCC protocol which can distill even a constant fraction of an arbitrary MD state into Bell pairs. In other words, the number of Bell pairs $M$ it distills from a general MD state must be $M/S_1(\psi_A) = o(1)$. This holds even if the protocol is given access to polynomially many copies of the input state.
    \item There is no efficient LOCC protocol which requires $M_-/S_1(\psi_A)=O(1)$ Bell pairs to prepare a general MD state $\ket{\psi}$ across $A|B$, even if either Alice or Bob (i.e., the parties on either side of the bipartition) has \emph{local} access to polynomially many copies of $\ket{\psi}$. In other words, any efficient LOCC dilution protocol requires $M_-/S_1(\psi_A) = \omega(1)$. This holds even if the protocol is given access to polynomially many copies of the input state.
\end{enumerate}
Moreover, these no-go results hold even if the protocols are guaranteed a priori that the input MD states all obey $\nu=\Theta(f(n))$, for any $f(n) =\omega(\log n)$.
\end{lemma}
\begin{proof}
Let us set $\nu=\Theta(\log^{c}(n))$ for $c>1$. Define the following ensemble of states 
\be
\mathcal{E}_{\haar(\nu)}=\qty{\ket{\phi_{\haar}} \otimes \ket{0}^{\otimes(n-\nu)} \mid \ket{\phi_{\haar}} \in \haar_\nu}
\ee
where $\haar_\nu$ is the $\nu$-qubit Haar ensemble. Similarly, we define
\be
\mathcal{E}_{f,S,\nu}=\qty{\ket{\phi_{f,S}} \otimes \ket{0}^{\otimes(n-\nu)} \mid \ket{\phi_{f,S}} \in \mathcal{E}_{f,S}(\nu)}
\ee
where $\mathcal{E}_{f,S}(\nu)$ is the ensemble of subset phase states defined on $\nu$ qubits, with $f$ being a random boolean function and $S$ being a random subset of $\qty{0,1}^\nu$~\cite{aaronson_quantum_2023}. Let us set $\abs{S}=\Theta(\exp\log^{c'}(n))$ with $1<c'<c$. 

Now, assume WLOG that our bipartition $A|B$ has $A$ being the first $\nu/2$ qubits and $B$ being the rest. Note that $\ket{\psi}\in \mathcal{E}_{\haar(\nu)}$ has $S_{1}(\psi_A)=\Theta(\nu)=\Theta(\log^c(n))$ up to negligible failure probability, while $\ket{\psi}\in \mathcal{E}_{f,S,\nu}$ obeys $S_{1}(\psi_A)=\Theta(\log \abs{S})=\Theta(\log^{c'}(n))$ up to negligible failure probability. We also note that the states in both ensembles $\mathcal{E}_{\haar(\nu)}$ and $\mathcal{E}_{f,S,\nu}$ are MD -- that is, these states obey $S_{1}(\psi_A)=O(\nu)$ up to a negligible failure probability~\cite{gu2023little}. Finally, as shown in Ref. \cite{aaronson_quantum_2023}:
\be
\norm{\mathbb{E}_{\ket{\psi} \sim \mathcal{E}_{f,S,\nu}} \qty[\ketbra{\psi}^{\otimes m}] - \mathbb{E}_{\ket{\psi} \sim \mathcal{E}_{\haar(\nu)}} \qty[\ketbra{\psi}^{\otimes m}]}_1 \le o\qty(\frac{1}{\poly n})\qc \forall m = O(\poly n),
\ee
where we used that $\abs{S}=\Theta(\exp \poly\log n)$. This is all to say that the two ensembles are statistically indistinguishable, meaning that any quantum algorithm that uses a polynomial number of copies of $\ket{\psi}$ cannot distinguish whether a given state belongs to $\mathcal{E}_{\haar(\nu)}$ or $\mathcal{E}_{f,S,\nu}$. 
\begin{enumerate}
    \item We first show that there cannot be any sample efficient algorithm estimating $S_{1}(\psi_A)$ up to a relative error $o(1)$. Assume, for the sake of contradiction, that there does exist such an algorithm. However, this would enable one to distinguish between states drawn from the two ensembles $\mathcal{E}_{f,S,\nu}$ and $\mathcal{E}_{\haar(\nu)}$ as follows. We would simply apply the algorithm to estimate the entanglement of the given state, and if the estimated entanglement were $\Theta(\log^{c'} n)$, we could conclude the state belonged to $\mathcal{E}_{f,S,\nu}$ (since the relative error is $o(1)$) and otherwise infer it belonged to $\mathcal{E}_{\haar(\nu)}$.
    \item We now show that there is no sample-efficient (hence no efficient) algorithm that distills more than $M$ Bell pairs obeying $M/S_1(\psi_A)=\Omega(1)$. Assume for the sake of contradiction that such an algorithm exists. One could use such an algorithm as a distinguisher, simply by counting the number of Bell pairs, as for any $\ket{\phi}\in\mathcal{E}_{f,S,\nu}$, this distiller must produce $\Theta(\log^{c'} n)$ Bell pairs, while for $\ket{\psi}\in\mathcal{E}_{\haar(\nu)}$, it must produce $\Theta(\log^c n)$ Bell pairs. Given that $c^{\prime}<c$, the distinguisher can actually discern between the two ensembles, thus we have arrived at a contradiction. Therefore, the maximal number $M$ of distillable Bell pairs obeys $M/S_1=o(1)$. In the extreme case when $\nu=n$, this inequality grows to $M/S_1=O(\log^c n/n)$ for any $c>1$.
    \item Finally, to show that any entanglement dilution protocol requires $M/S_1=\omega(1)$ Bell pairs, assume for the sake of contradiction that there exists a LOCC dilution protocol that requires $M_-/S_1=O(1)$ Bell pairs to prepare a general MD state $\psi$. This implies that, if Alice is given access to polynomial copies of a pseudorandom state $\psi \sim \mathcal{E}_{f,S}$ (with $\nu=n$ for simplicity -- the general case of $\nu=\Theta(f(n))$ follows similarly), Alice and Bob can together prepare $\psi$ across the bipartition $A|B$ using $M=O(S_1(\psi))=O(\poly \log n)$ Bell pairs. Similarly, if Alice is given access to polynomial copies of a Haar random state $\psi \sim \mathcal{E}_{\haar}$, Alice and Bob can prepare $\psi$ across the bipartition $A|B$ using $M=\Theta(n)$ Bell pairs (the lower bound is because $\psi$ has $\Omega(n)$ entanglement with overwhelming probability). However, this leads to a contradiction; we can construct a distinguisher between $\mathcal{E}_{f,S}$ and $\mathcal{E}_{\haar}$ as follows. We run the dilution protocol, supplied with $\Theta(n^2)$ Bell pairs. If the dilution protocol is given a pseudorandom state $\psi$ (or polynomially copies thereof), only $O(\poly \log n)$ of the Bell pairs should be consumed, while if it is given a Haar random state, $\Theta(n)$ Bell pairs should be consumed. The number of consumed Bell pairs at the end of the protocol can be found using a simple swap test on each of the supplied Bell pairs. This would then allow us to distinguish whether a pseudorandom state or Haar random state was input, leading to a contradiction.
\end{enumerate}
\end{proof}

As a consequence of the above proof, we can give strong bounds for general entanglement dilution protocols, analogous to the bounds for entanglement distillation given in Ref.~\cite{aaronson_quantum_2023}.
\begin{corollary}[No-go for efficient entanglement dilution]\label{cor:hardnessidulution}
Given local query access to polynomially many copies of an unknown state $\ket{\psi}$, the minimum number of Bell pairs needed to dilute $\ket{\psi}$ across $A|B$ is $M_{-}=\Omega(\exp(S_1(\psi_A)^{1/c}))$ for any $c > 1$.
\begin{proof}
In the extreme case of the proof of Item 3 in \cref{thm:no-go-magic-dom}, when we set $\nu=n$, we must have $M_{-}=\Theta(n)$, while $S_1(\psi_A)=\Theta(\log^c n)$ with $c>1$ for pseudoentangled states, so $M_- = \Omega(\exp(S_1(\psi_A)^{1/c}))$ for $c>1$.
\end{proof}
\end{corollary}

\parhead{A computational phase transition between entanglement- and magic-dominated phases.} We are now in a position to synthesize all of our results to summarize the operational distinction between ED and MD states. We have shown in \cref{eq:detectability,lem:phasetransdistill} that entanglement characterization and distillation are efficient for ED states, while we show the converse in \cref{thm:no-go-magic-dom}. These two results are summarized in \cref{tab:transition}.

Another way of understanding this computational phase transition is through the lens of entanglement reversibility. Entanglement distillation allows us to transform (multiple copies of) a state $\rho$ to a number $M_+$ of Bell pairs, and entanglement dilution allows us to use a number of Bell pairs to prepare the state $\rho$ via LOCC. In particular, the minimum amount of Bell pair $M_-$ necessary to distill $\rho$ is named as \textit{entanglement cost}. The ratio $M_{+}/M_{-}$ quantifies the \textit{reversibility} of entanglement for the state $\rho$. For pure states, without any restriction to efficient (sample-wise and computationally) LOCC protocol, it has been shown that entanglement is always reversible, meaning that the ratio is always one~\cite{PhysRevA.53.2046}. However, when considering efficient and state-agnostic LOCC protocols, it has been demonstrated in \cite{aaronson_quantum_2023} that the maximal distillable entanglement is in general $M_{+}=O(\log^{2} S_{1}(\psi_A))$. Furthermore, in \cref{cor:hardnessidulution}, we demonstrated that $M_{-}=\Omega(\exp(S_{1}(\psi_A)^{1/2}))$ Bell pairs are necessary for every efficient state-agnostic dilution protocol. Combining these two results, it immediately follows that, for efficient and state-agnostic LOCC protocol, the reversibility ratio $M_+/M_-$ for pure states is always asymptotically vanishing in $n$ $M_{+}/M_{-}=o(1)$ (for every bipartition $A|B$ with $\omega(\log n)=n_A,n_B$). This shows that, even for pure states, entanglement manipulation is not computationally reversible when no further information about the state is available.
However, given partial information about the state, one might wonder if entanglement reversibility is somewhat restored. For instance, in the case of ED states, how does the reversibility ratio behave? In stabilizer states, the reversibility ratio is always one \cite{fattal_entanglement_2004}: since the entanglement spectrum of stabilizer states is flat, essentially the same one-shot stabilizer LOCC protocol can be used for both distillation and dilution. Below, we put together all our findings and prove that the separation between entanglement- and magic-dominated phases marks a sharp phase transition in entanglement reversibility under efficient and state-agnostic LOCC protocols.

\begin{corollary}[Phase transition in entanglement reversibility]\label{threversibilityphase}
Given an unknown state $\ket{\psi}$, let $M_{+}$ be the number of Bell pairs distillable by an efficient LOCC protocol, and let $M_{-}$ be the number of Bell pairs necessary to prepare $\ket{\psi}$ via LOCC. For any ED $\ket{\psi}$,
\begin{equation}
    \frac{M_{+}}{M_{-}}\ge1-o(1),
\end{equation}
while if $\ket{\psi}$ is MD, no efficient LOCC protocol can do better than
\begin{equation}
    \frac{M_{+}}{M_{-}}\le o(1).
\end{equation}
\begin{proof}
From \cref{thm:distill}, we know that there exists a state-agnostic one-shot protocol that distills a number $M_{+}$ of Bell pair given by $M_{+}\ge S_{1}(\psi;A|B)-\nu/2$. Conversely, from \cref{thm:dilution} we know that $M_{-}\le\mathcal{E}(\psi;A|B)+\nu/2$ Bell pairs are sufficient to prepare $\ket{\psi}$ via LOCC. If $\ket{\psi}$ is ED, then
\be
\frac{M_{+}}{M_{-}}\ge \frac{S_{1}(\psi;A|B)-\nu/2}{\mathcal{E}(\psi;A|B)+\nu/2}\ge \frac{S_{1}(\psi;A|B)-\nu/2}{S_{1}(\psi;A|B)+2\nu}=1-o(1)
\ee
Conversely, let $\ket{\psi}$ be a MD state. We know that $M_{-}\ge S_{1}(\psi;A|B)$ for any dilution protocol~\cite{PhysRevA.53.2046}, and from \cref{thm:no-go-magic-dom} we know that for any state-agnostic efficient protocol $M_{+}/S_{1}(\psi;A|B)=o(1)$. Therefore:
\be
\frac{M_{+}}{M_{-}}\le \frac{M_{+}}{S_{1}(\psi;A|B)}=o(1)
\ee
which concludes the proof. 
\end{proof}
\end{corollary}

Before concluding this section, we mention a recent result that proves the impossibility of reversible entanglement manipulation for general mixed states~\cite{lami_no_2023}, even under a class of entanglement-free operations that strictly generalize LOCC, and if we do not insist upon state-agnostic or efficient protocols.

\section{Inference on entanglement- vs. magic-dominated states}\label{sec:unknownstates}
\parhead{Stabilizer group inference.} A theme throughout this work is that the stabilizer generators $S$ of the state in \cref{eq:structure1} contain most of the relevant information about its entanglement structure, particularly when the state is within the ED phase. However, in the statement of all our results, we implicitly assumed a priori knowledge of the stabilizer group generators $S$. This assumption is not necessary: $S$ is readily available in almost any setting we care to consider. If we know the unitary that generates $\ket{\psi}$ (i.e., we have a efficient circuit description of the state), we do not even need to run the circuit on real quantum hardware; the generators $G$ can be calculated in polynomial time using the stabilizer formalism, simply by processing the circuit on classical hardware (see \cref{alg:monitor}). Another setting might be where we have black-box access to an unknown state $\ket{\psi}$. In this case, we can efficiently learn the generators $G$ for its stabilizer group using just $O(n)$ copies of $\ket{\psi}$. using the tools of Ref.~\cite{grewal2023efficient}. For the reader's convenience, we state the result below.

\begin{lemma}[Efficient stabilizer group learning \cite{grewal2023efficient}]\label{thm:stab-learn}
For any $\epsilon,\delta \in (0,1)$ there exists an algorithm that outputs an Abelian group $\hat{G} \subset \mathbb{P}_n$, that uses only $\frac{16n+8\log(1/\delta)}{\epsilon^2}$ 4-copy measurements of an unknown state $\ket{\psi}$. This algorithm takes $O(\frac{n^3+n^2 \log(1/\delta)}{\epsilon^2})$ time, and up to a failure probability of at most $\delta$, this Abelian group $\hat{G}$ satisfies two properties:
\begin{enumerate}
    \item The true stabilizer group $G$ of $\ket{\psi}$ is a subgroup of $\hat{G}$, which implies $\log_2 \abs*{\hat{G}} \geq \log_2 \abs{G} \geq n-\nu$; and
    \item $\ket{\psi}$ is $\epsilon$-close in trace distance to some state with a stabilizer group $\hat{G}$.
\end{enumerate}
\end{lemma}
\begin{proof}
This follows from Algorithm 1 and Theorem 5.1 of Ref.~\cite{grewal2023efficient}. 
\end{proof}
\begin{remark}
Note that this stabilizer group learning algorithm does not guarantee that we can learn the true stabilizer group of the state $\ket{\psi}$. However, the output group $\hat{G}$ is the stabilizer group of some other state $\ket{\sigma}$ that is $\epsilon$-close in trace distance to $\ket{\psi}$. If $\epsilon=\frac{1}{n}$, by the Fannes inequality, the entanglement of $\ket{\sigma}$ across any cut $A|B$ differs from the entanglement of $\ket{\psi}$ across $A|B$ by at most $2$. If $\epsilon=o(n^{-1})$, this gap is improved to $o(1)$ (i.e., it is vanishing in the limit of large $n$). This is all to say that when we do entanglement calculation, when $\epsilon$ is small enough, we can pretend $\hat{G}$ is the true stabilizer group of the state. Furthermore, for entanglement manipulation tasks (e.g., Bell pair distillation), we can design protocols that work for $\ket{\sigma}$ (i.e., we will assume that $\hat{G}$ is the true stabilizer group). Although $\ket{\psi}$ may be the true input state, when we run our protocol on $\ket{\psi}$, we are guaranteed that the output will be $\epsilon$-close in trace distance to the desired output, since the trace distance is monotonically decreasing under all CPTP maps. In summary, for all intents and purposes, when $\epsilon$ is very small, we can pretend $\hat{G}$ is the true stabilizer group of the state.
\end{remark}

\begin{corollary}[Robustness of stabilizer group learning]\label{cor:robust-learn}
If an unknown (potentially mixed) state $\psi$ is $\epsilon_1$-close in trace distance to some other unknown state $\ket{\psi'}$ with stabilizer group $G'$, the stabilizer group learning algorithm will still succeed with failure probability at most $\frac{64\epsilon_1}{\epsilon^2} \qty(2n +\log(\frac{\epsilon^2}{\epsilon_1}))$. The output stabilizer group $\hat{G}$ will satisfy $G' \subseteq \hat{G}$, and $\psi$ will be $(\epsilon+\epsilon_1)$-close in trace distance to some state with stabilizer group $\hat{G}$.
\end{corollary}
\begin{proof}
The stabilizer group learning algorithm uses exactly $N = \frac{64n+32\log(1/\delta)}{\epsilon^2}$ copies of $\psi$. If $\psi$ is $\epsilon_1$-close in trace distance to some $t$-doped stabilizer state $\ket{\psi'}$, the maximum probability of distinguishing the two is $N \epsilon_1$ (this is the maximum possible trace distance between $\ket{\psi}^{\otimes N}$ and $\ket{\psi'}^{\otimes N}$). So, we define a distinguisher as follows. If the output of our algorithm is one of the `good' $\hat{G}$ (i.e., $G \subseteq \hat{G}$ and $\ket{\psi'}$ is $\epsilon$-close in trace distance to some state with stabilizer group $\hat{G}$) that could be output assuming the input state was $\ket{\psi'}$, we output $1$, and $0$ otherwise. By definition, this distinguisher outputs $1$ with probability at least $1-\delta$ for an input $\ket{\psi'}$. The properties of the trace distance tell us that the distinguisher will output $1$ with probability at least $1-\delta-2N \epsilon_1$ when the input is $\psi$. That is, the failure probability is at most $\delta + 2N \epsilon_1$. If we set $\delta = \frac{64 \epsilon_1}{\epsilon^2}$, the overall failure probability (i.e., the probability that one of the `good' $\hat{G}$ is \emph{not} output) is at most
\begin{equation}
\frac{64\epsilon_1}{\epsilon^2} \qty(1 + 2n +\log(\frac{\epsilon^2}{64\epsilon_1})) \leq \frac{64\epsilon_1}{\epsilon^2} \qty(2n +\log(\frac{\epsilon^2}{\epsilon_1})).
\end{equation}
\end{proof}

\parhead{Testing for the entanglement-dominated phase.} Having shown the ability to learn a state's stabilizer group, we now apply this to answer even more specific questions about particular states. Given access to an unknown state $\ket{\psi}$, can we tell whether it is ED or MD? In complete generality, this task is not efficiently possible as it would necessitate efficient measurement of entanglement for general states. However, we can reformulate this task as a property testing problem that is efficiently solvable. Recall that a property tester for a class $C$ of quantum states takes copies of a state $\ket{\psi}$ as input and determines whether $\ket{\psi}\in C$ or $\ket{\psi}$ is $\epsilon$-far in trace distance from all states in $C$, \emph{promised that one of these is the case}. Our algorithm efficiently tests whether an input state $\ket{\psi}$ belongs to the ED phase or is $\epsilon$-far from all such states and, as such, it is MD.

\begin{theorem}[Property testing for entanglement-dominated states]\label{th:distinction}
Let $\ket{\psi}$ be an unknown quantum state. For any $\epsilon=\Theta(n^{-\alpha})$ with $\alpha \geq 1$, there is a ED property tester that uses polynomially many copies of $\ket{\psi}$, polynomial classical computing time, and has a failure probability of at most $\delta=O(\exp(-n))$ to distinguish between the following cases:
\begin{itemize}
    \item $\ket{\psi}$ is ED;
    \item $\ket{\psi}$ is $\varepsilon$-far in trace distance from any ED state, hence $\ket{\psi}$ is MD.
\end{itemize}
\begin{proof}
As shown in \cref{thm:stab-learn}, it is possible through $O(\poly n)$ queries to $\ket{\psi}$ to determine an Abelian group $\hat{G}$ such that $G \subseteq \hat{G}$ (where $G$ is the true stabilizer group) and $\ket{\psi}$ is $\epsilon$-close in trace distance from a state $\ket*{\hat{\psi}}$ with stabilizer group $\hat{G}$, with a failure probability that is exponentially small in $n$. 

Assume that $\ket{\psi}$ falls into case A) -- that is, it is an ED state. We know that $n-\log_2 \abs*{G} \leq o(\mathcal{E}(\psi;A|B))$ by definition of the ED phase. Since $\abs*{\hat{G}} \geq \abs{G}$ and $\abs*{S_1(\psi_A) - S_1(\hat{\psi}_A)} \leq O(1)$ by the Fannes inequality, we must then also have
\begin{equation}
    n-\log_2 \abs*{\hat{G}} \leq o(\mathcal{E}(\hat{\psi};A|B)).\label{eq:ent-dom-case}
\end{equation}
This is to say, $\ket*{\hat{\psi}}$ is within the ED phase as well. 

Take the other possibility, which is that $\ket{\psi}$ falls into case B). Then, $\ket*{\hat{\psi}}$ must also be MD by the assumption that $\ket{\psi}$ is $\epsilon$-far from all ED states. In this case, we immediately have
\begin{equation}
    n-\log_2 \abs*{\hat{G}} \geq \Omega(\mathcal{E}(\hat{\psi};A|B)).\label{eq:magic-dom-case}
\end{equation}
Since distinguishing between \cref{eq:ent-dom-case,eq:magic-dom-case} is efficient (calculating both $\log_2 \abs*{\hat{G}}$ and $\mathcal{E}(\hat{\psi};A|B)$ is efficient), the property tester is overall efficient.
\end{proof}
\end{theorem}

\section{Physical applications}\label{sec:physical}

Entanglement is a fundamental physical quantity, and we therefore expect the tools we developed here to have applications for a wide array of physics problems. Of course, for our tools to apply, we need to be able to find entanglement-dominated states in nature. Happily, we are able to identify a number of settings in which these states appear quite naturally, such as quantum error correcting codes and many-body systems.

The ubiquity of entanglement-dominated states in these diverse physical settings allows us to leverage the ED-MD framework developed in this work to gain new insights into various subfields of physics. Our tools not only provide a deeper understanding of the role of entanglement in these systems but also offer practical applications, such as efficient entanglement estimation protocols that can be used to study scrambling dynamics and characterize quantum chaos. By bridging the gap between these seemingly disparate fields, our work highlights the fundamental role of the entanglement-magic interplay in understanding the complex behavior of quantum systems. In this section, we showcase the power and versatility of our approach by applying it to three distinct areas: quantum error correction, many-body physics, and the characterization of scrambling dynamics. These applications demonstrate the unifying nature of the ED-MD perspective and its potential to provide fresh insights into seemingly unrelated physical phenomena.

\subsection{Entanglement-dominated states as energy eigenstates}\label{sec:manybodyphysics}

In this section, we show the emergence of ED states in the contexts of many-body Hamiltonians.

\medskip 
\parhead{Entanglement of perturbed stabilizer Hamiltonians.} To establish this connection between energy eigenstates and ED states, it is more convenient to begin with stabilizer states. Stabilizer states also find relevance in many-body physics as eigenstates of so-called stabilizer Hamiltonians. A stabilizer Hamiltonian $H_{G}$ is defined as
\be\label{eq:hstab}
H_{G}=\sum_{P\in G}\alpha_P P
\ee
where $G\subset\mathbb{P}_n$ is an Abelian subgroup of the Pauli group. It is well known~\cite{coble2023local} that all the eigenstates of $H_{G}$ correspond to stabilizer states associated with the stabilizer group $G$. Recently, we have shown that perturbation to $H_G$ leads to $\nu$-compressible states, see Ref.~\cite{gu_hamiltonian_2024}. In particular, let us introduce a perturbation to $H_G$ formed by $k$ generic Pauli operator $P_{1},\ldots, P_k$ as
\be\label{eq:hplhstab}
H=H_{G}+\sum_{i=1}^{k}\gamma_i P_i
\ee
with $\gamma_i$ real coefficients. Letting $K = \expval{\qty{P_1,P_2,\ldots,P_k}}$, we have shown that there is an Abelian subgroup $J \subset \mathbb{P}_n$ of dimension at least $n-\dim K$ (which is in turn lower bounded by $n-k$), such that there is an eigenbasis for $H$ in which every eigenstate is stabilized by the elements of $J$ (up to $\pm$ phases). This is to say that they are low-magic states with $\nu \leq \dim K \leq k$. Below, we analyze the implications of the existence of the separation between ED and MD states in many-body physics, employing the fact that the Hamiltonian \eqref{eq:hplhstab} features $\nu$-compressible eigenstates.  

\begin{corollary}[Near-identical entanglement for all eigenstates]\label{cor:eig-sim}
Let $H=H_{G}+\sum_{i=1}^{k}\gamma_i P_i$ with $H_G$ being a stabilizer Hamiltonian corresponding to the stabilizer group $G$. In the $\nu$-compressible eigenbasis of $H$ corresponding to the stabilizer group $J$, the von Neumann entanglement entropy (across any cut) of any two eigenstates differs by at most $\dim K$.
\end{corollary}
\begin{proof}
All of the eigenstates have the same stabilizer group, hence have the same stabilizer entanglement. Since \cref{thm:stab-approx} says that the stabilizer entanglement differs from the true entanglement by at most $\frac{\nu}{2}$ (which is at most $\frac{\dim K}{2}$), this shows the desired claim.
\end{proof}
This may seem to contradict physical intuition, as we expect that for local Hamiltonians, generic high energy states will be volume law entangled while low energy states are area law states. However, note that \cref{cor:eig-sim} simply says that there exists \emph{one particular} eigenbasis for which all the eigenstates have a similar entanglement. It does not rule out the possibility that most of the states in a high energy subspace have high entanglement, rather it merely says that a select few of these states have the same entanglement as the ground states.

\begin{lemma}[Robust ground state entanglement entropy]\label{lem:pert-ent}
Consider a Hamiltonian $H=H_{G}+\sum_{i=1}^{k}\gamma_{i}P_i$. The ground state entanglement entropy (across any cut) of $H$ differs from the ground state entanglement entropy of $H_G$ by at most $\frac{3}{2} \dim K$.
\end{lemma}
\begin{proof}
Let $G'$ be any completion of $G$ (i.e., it will be a size-$2^n$ Abelian subgroup of $\mathbb{P}_n$ such that $G \subseteq G'$). The eigenstates $\ket{\psi_G}$ of $H_G$ are stabilizer states with a stabilizer group $G'$~\cite{coble2023local}. Define $K^\perp \coloneqq \qty{P \in \mathbb{P}_n \mid \comm{Q}{P}=0,\, \forall Q \in K}$. Consider the subgroup $J \coloneqq K^\perp \cap G'$ -- this contains all Pauli operators that commute with every element in $K$ and $G$. In Ref.~\cite{gu_hamiltonian_2024}, it was shown that $H$ admits an eigenbasis in which every eigenstate is stabilized by $J$. The group $J$ has dimension $\dim(J) \geq n-\dim K$. Since $J \subseteq G'$,
\begin{equation}
    \abs{\mathcal{E}(\psi;A|B) - \mathcal{E}(\psi_G;A|B)} \leq \dim K.
\end{equation}
Now, using the fact that $\psi_G$ is an exact stabilizer state and $\psi$ is a $\nu$-compressible state with $\nu \leq \dim K$, we apply \cref{thm:stab-approx} to conclude
\begin{equation}
    \abs{S_1(\psi) - S_1(\psi_G)} \leq\frac{3}{2}\dim K.
\end{equation}
\end{proof}
Combining the above two results is particularly useful for models where we have knowledge of the entanglement for \emph{just one} of the eigenstates, such as the ground state. In fact, thanks to \cref{lem:pert-ent}, we only need to know about the ground state entanglement of the unperturbed (i.e., stabilizer) Hamiltonian, and given the ease with which this ground state can be found analytically, this information can always be found in polynomial classic time. Then, when $k$ is dominated by the ground state entanglement, we can first use \cref{lem:pert-ent} to estimate the ground state entanglement of the perturbed model, then use \cref{cor:eig-sim} to estimate the entanglement for \emph{any} of the eigenstates of the perturbed model.

\medskip 
\parhead{A group-theoretic topological-entropy robustness proof from ED states.} We can also consider more sophisticated types of entanglement, such as the topological entanglement entropy, which can often be defined as a sum of bipartite entanglement entropies such as
\begin{equation}
    S_{\text{topo}} = S_{AB} + S_{BC} - S_B - S_{ABC},\label{eq:stopo}
\end{equation}
where $A$, $B$, and $C$ are appropriately chosen regions~\cite{zeng_quantum_2018}. This quantity is also often known as the conditional mutual information $I(A:C \mid B)$. We note that since each of the bipartite entanglement entropies in \cref{eq:stopo} can be estimated up to an additive error $O(\nu)$, it follows that $S_{\text{topo}}$ can be estimated up to an error $O(\nu)$ as well. Although $S_{\text{topo}}$ is traditionally viewed to be a small constant number (e.g., it is 2 for a GHZ state), there are topologically ordered states of matter that can have very large topological entanglement entropies. For instance, in Ref.~\cite{ma_topological_2018}, it was shown that the ground states for two archetypical models of three-dimensional topologically ordered phases of matter (namely, the X-cube model and Haah's code), have a topological entanglement entropy that scales \emph{linearly} in system size. We expect the tools we developed for estimating entanglement entropy to be very useful in settings of this sort, as the small $O(\nu)$ correction will be a very small relative correction to the estimate of $S_{\text{topo}}$. The following corollary argues that the topological entanglement entropy of a stabilizer Hamiltonian $H_G$ is also robust against perturbations. This is especially important for models where $S_{\text{topo}}$ is extensive, such as those in Ref.~\cite{ma_topological_2018}, because the following corollary implies that for perturbations to these models with up to $o(n)$ terms, the topological entanglement of the ground state will still scales extensively in system size.

\begin{lemma}[Robust topological entanglement]
The ground state topological entanglement entropy of $H$ differs from the ground state topological entanglement of $H_G$ by at most $O(\dim K)$.
\end{lemma}
\begin{proof}
The topological entanglement entropy is simply a linear combination of bipartite entanglement entropies \eqref{eq:stopo}. Since each of these bipartite entanglement entropies are the same for the ground states up to an error $4 \dim K$, the topological entanglement entropies are also the same up to an error $O(\dim K)$.
\end{proof}

\subsection{Entanglement-dominated states as code states: topological error correction}\label{sec:err-correct}

The toric code Hamiltonian is one of the most widely studied models of topological order, and as its name suggests, it is intimately connected with error correction. The Hamiltonian is
\begin{equation}
    H = -\sum_v A_v - \sum_p B_p,
\end{equation}
where $A_v$ and $B_p$ are the star and plaquette operators, respectively. Formally, $A_v=\prod_{j \in \text{star}(v)} Z_j$, $B_p=\prod_{j\in \text{bd}(p)} X_j$, where $\text{star}(v)$ is a star around the vertex $v$ and $\text{bd}(p)$ is the boundary of the plaquette $p$. The ground state of this Hamiltonian is 4-fold degenerate, which allows these states to be used as logical states. That is, the four energy eigenstates can be identified with $\ket{00}_L$, $\ket{01}_L$, $\ket{10}_L$, and $\ket{11}_L$, where the $L$ subscript denotes the fact that these are logical states, rather than actual physical ones on the lattice. Each of these ground states is the $+1$ eigenstate of $2n-2$ \emph{independent} stabilizers. The remaining two stabilizers are strings of $Z$ operators that form horizontal or vertical loops around the torus -- these are referred to as the logical operators $Z_1$ and $Z_2$ of the code. Whether the state is a $+1$ or $-1$ eigenstate of these two logical operators determines the logical state (e.g., if the state is the $+1$ stabilizer of the two loop operators, it corresponds to $\ket{00}_L$). However, this implies that any state in the ground state space of the toric code Hamiltonian (i.e., an error-free state, or a \emph{code state}) can be written
\begin{equation}
    \rho = \qty(\prod_{v=1}^{n-1} \frac{\id + A_v}{2}) \qty(\prod_{p=1}^{n-1} \frac{\id + B_p}{2}) \rho_L,
\end{equation}
where the portion of the state $\rho_L$ determines the logical state represented by $\rho$. For instance, if $\rho_L = \frac{(\id + Z_1)(\id+Z_2)}{4}$, then $\rho$ is identified with $\ketbra{00}_L$. Of course, in general, if we have a superposition of logical states, $\rho_L$ will not necessarily be a stabilizer state. However, due to the $2n-2$ shared stabilizers, this form immediately makes it obvious that any state in the ground state space, even if it is a superposition of logical states, is a state with $\nu \leq 2$. We note here that this supplies a new interpretation of the nullity distillation circuit in \cref{fig:clifft}: the distillation circuit is a \emph{decoding circuit}, since the state left over after the distillation represents the logical content of the code state. 

Furthermore, these states fall into the ED phase. Previous work~\cite{hamma_bipartite_2005} has shown that entanglement for a biparition $A|B$ in the ground space of the Toric code grows as $S_{1}(\psi_A)\sim\sqrt{n_A}$, following an area-law state in 2D. Therefore, with $S_{1}(\psi_A)=\omega(\nu)\coloneqq\omega(1)$, the ground space of the Toric codes is described by ED states, cfr. \cref{def:ent-dom}.

Another implication of the fact that all code states are $\nu$-compressible states is that, using the same logic as \cref{lem:pert-ent}, the entanglement of $\rho$ across any cut can differ from any one of the logical states by at most $4\nu=8=O(1)$. Indeed, in general, any stabilizer code Hamiltonian with $k$ logical states will have $\nu \leq k$, and since $k=O(1)$ for most topological codes \cite{bravyi2020bpt}, $\nu=O(1)$ for these codes. An immediate corollary of this means that the entanglement of any code state is equal to the entanglement of the logical states up to an $O(1)$ correction. Why is this important? It tells us that we can reason about the entanglement properties of any code state simply by knowing the entanglement properties of \emph{just a single logical state}. Furthermore, since the stabilizers of this state are analytically derivable simply by looking at the Hamiltonian, these entanglement properties are readily available. This is an extremely special property. On the other hand, this reasoning breaks for the recently proposed good qLDPC codes \cite{Breuckmann_2021,Panteleev_2021,xu2023constantoverhead}. These qLDPC Hamiltonians support $\Theta(n)$ logical qubits \cite{Panteleev_2021}, hence $\nu=\Theta(n)$, in which case our bounds become too loose to say anything useful.

\subsection{Lyapunov exponent estimation with entanglement-dominated states}\label{sec:lyapunov}
Lyapunov exponents were originally introduced in classical mechanics to probe chaos. They measure the rate of divergence of nearby trajectories in phase space, effectively quantifying the sensitivity to initial conditions. This phenomenon is the well-known ``butterfly effect,'' wherein small changes in initial conditions lead to dramatically different outcomes over time. In the context of quantum mechanics, the notion of Lyapunov exponents maintains a similar role as a probe of chaos. However, due to the unitarity of quantum mechanics, they cannot measure the divergence of trajectories in phase space as in classical mechanics. Instead, in quantum mechanics, Lyapunov exponents are related to the degree of noncommutativity between two operators.

Consider two operators $V$ and $W$ supported on non-overlapping domains such that $[W,V]=0$ and take a local unitary evolution $U(\tau)$ parametrized by a real parameter $\tau$. Note that $\tau$ can be interpreted as the time in Hamiltonian evolution or, as considered below, the number of layers of an evolution from a local quantum circuit. A local operator $W(0)=W$ evolves under the Heisenberg picture as $W(\tau)=U^{\dag}(\tau)WU(\tau)$. The commutator between $W(\tau)$ and $V$ will now increase because the domain of $W$ will gradually expand and eventually, it reaches the domain of $V$. The main difference, in the choice of the unitary evolution, lies in the speed of this expansion. A way to quantify the speed of this expansion is given by the out-of-time-order correlator ($\otoc$), which is connected to the notion of a group commutator. A way to define it is via the so-called, \textit{entanglement in time}~\cite{hosur_chaos_2016}. Let us consider a system made of two bipartitions, an input bipartition $A\vert B$ and an output bipartition $C\vert D$, and Hilbert space of $2n$ qubits where $n_A+n_B=n_C+n_D$. The OTOC is then defined as
\be
\otoc\left(U(\tau)\right)\coloneqq \frac{1}{2^{n+2(n_A+n_D)}}\sum_{P_A,P_D}\tr\left(P_D(\tau)P_AP_D(\tau)P_A\right)\, ,\label{eq:s2otoc}
\ee  
where $P_X$ is a local Pauli operator, $P_D(\tau)\coloneqq U(\tau)^{\dagger}P_D U(\tau)$. The OTOC is a well-known probe to chaos~\cite{kitaev_hidden_2014}, and as shown in~\cite{hosur_chaos_2016} can be connected to the entanglement entropy of the Choi representation of the unitary $U(t)$ as
\be
S_{2}(\rho_{AC}(\tau))=n_C-n_A-\log_2\left(\otoc(U(\tau)) \right)\,,
\ee
where $\rho_{AC}(\tau):=\tr_{BD}\ketbra{U(\tau)}{U(\tau)}$. $\ket{U(\tau)}$ is the Choi representation of the unitary operator $U(\tau)$ and is defined as $    \ket{U(\tau)}\coloneqq 2^{-n/2}\sum_{i}\id \otimes U(\tau) \ket{ii}$. Note that also $\ket{U(\tau)}$ is a $\nu$-compressible state since it corresponds to the action of a $t$-doped Clifford unitary on the stabilizer state $\ket{\id}=2^{-n/2}\sum_{i}\ket{ii}$ (see~\cref{cor:nucompressibility}).

In what follows, we will show that the $\otoc$ and the entanglement entropy are both connected with the Lyapunov exponent. Below, we will focus on the class of unitary operator $U(\tau)$ generated by a local quantum circuit. Following Ref.~\cite{vonkeyserlingk_operator_2018}, we consider a discrete-time evolution, made of layers of two-site random unitary gates, acting on pairs of neighboring sites. Each time step consists of two layers, an odd layer acting on all odd neighbouring bonds, and an even layer acting on all the even neighbouring bonds. In a clear analogy to what was done in~\cite{vonkeyserlingk_operator_2018} for the entanglement, one can easily show that the average OTOC for two-qubit gates drawn according to a $2$-design distribution in a brickwork fashion behaves as 
\be 
\langle\otoc(U(\tau))\rangle &=2^{-2 n_A}+2^{-2 n_D}-2^{-2(n_A+n_D)}+\left(1-2^{-2 n_A}+2^{-2 n_D}-2^{-2(n_A+n_D)}\right)e^{-\frac{\tau}{2} \log\left(\frac{1}{1-\lambda_L^2}\right)}
\ee 
where we denote the average as $\langle\cdot\rangle$. Therefore, one has the following scaling for 
\be
\langle\otoc(U(\tau))\rangle= \exp\left(-\frac{\tau}{2} \log\left(1-\lambda_L^2\right)^{-1}\right)+O(\exp(\tau)/2^{2n_A})\,.
\ee
If we use the connection between $\otoc(U)$ and entanglement entropy \eqref{eq:s2otoc} we obtain
\be \label{eq:s2entav}
\langle S_{2}(\rho_{AC}(\tau))\rangle =n_C-n_A + c(\lambda_L)\tau+O(\exp(\tau)/2^{2n_A})
\ee 
for $c(\lambda_L)=\frac{1}{2\log 2}\log \frac{1}{(1-\lambda_L^2)}$. The result shows how in a $2$-design random circuit model, the scrambling velocity is connected to the Lyapunov exponent. Notably, \cref{eq:s2entav} holds for local random Clifford circuits (being unitary $2$-design), as well as for $t$-doped local random Clifford circuit for any value of the doping. 

Motivated by \cref{eq:s2entav}, we hypothesize that for a single instance of a typical random local $t$-doped Clifford circuit, the behaviour of $S_2(\rho_{AC}(\tau))$ is approximately the one given in \cref{eq:s2entav}. Therefore, we assume that for $\tau=o(n_A)$ one has $S_{2}(\rho_{AC}(t))=n_c-n_A+c(\lambda_L)\tau$ with the constant $c(\lambda_L)=\frac{1}{2\log 2}\log \frac{1}{(1-\lambda_L^2)}$ and some Lyapunov exponent $\lambda_L$ to be determined. Then, one can estimate the Lyapunov exponent of the circuit, by performing a linear regression of the entanglement entropy of $\rho_{AC}(\tau)$ versus the number of layers $\tau$. In the following, we will consider the case $n_A=n_C-\omega(\log n)$, with $n_C>n_A>1$. In such a case, it is clear that a simple swap test cannot be used to estimate the Lyapunov exponent because the number of measurements required to estimate $S_2(\rho_{AC})$ will be superpolynomial. However, using our results on ED phase, from \cref{eq:detectability}, we are able to estimate the entanglement entropy $S_1$, or the $2$-R\'enyi entropy $S_2$, up to a $o(1)$ relative error in the ED phase. Assuming the entanglement growth linear in $\tau$, for any $\tau\in[\omega(\max(\log n, \nu)),o(n_A)]$ the state $\rho(\tau)$ is ED in the bipartition $AB|CD$. Therefore, as long as the doping $\nu=o(n_A)$, one can always find a useful time-window to perform a linear regression and estimate the Lyapunov exponent $\lambda_L$. The only aspect that changes with the nullity, $\nu$, of $\ket{U(\tau)}$ is the useful time window during which one can estimate the $\otoc$ using stabilizer entanglement. In particular, as shown in \cref{eq:detectability}, in the ED phase, we can estimate $c(\lambda_L)$ with an asymptotically vanishing error $o(1)$, which, via standard error propagation, transfers to an asymptotic vanishing error on the Lyapunov exponent $\lambda_L$.

In summary, given query access to a unitary $U(\tau)$, an experimenter can construct its Choi representation by applying it to a maximally entangled state $2^{-n/2}\sum_{i}\ket{ii}$. Then, they can apply \cref{thm:stab-learn} to approximate the stabilizer entanglement, hence allowing them to estimate the Lyapunov exponent of the circuit. Alternatively, if the experimenter has access to a gate set, it is possible to use \cref{alg:monitor} to directly monitor the entanglement of the Choi representation $U(\tau)$ and obtain an estimate of the Lyapunov exponent efficiently through classical simulation.

\section{Understanding phenomenology with Entanglement- vs. Magic-dominated phase}\label{sec:phenomenology}

In addition to the physical applications discussed in the previous section, the tools and techniques developed in this work can also provide theoretical insights into various phenomena observed in numerical simulations of quantum circuits. In this section, we showcase the utility of the separation between ED and MD phases for gaining a theoretical understanding of previously observed phenomena in numerical simulations. 

\medskip 
\parhead{Entanglement cooling.} The first phenomenon revolves around a task known as \emph{entanglement cooling} \cite{chamon_emergent_2014,shaffer_irreversibility_2014}. The essential idea is that the complexity of a state's entanglement structure can be measured by the difficulty of disentangling a state across a given bipartition without any prior knowledge of the entangling dynamics that generated the state. To be more precise, fix $A|B$ and a entangling unitary $U$ that produces $\ket{\psi}=U\ket{0}^{\otimes n}$. The entanglement cooling task is the following: given only query access to $\ket{\psi}$ (i.e., no knowledge of $U$), design a unitary operator $U'$ such that the state $\ket{\psi'} \coloneqq U'\ket{\psi}$ is minimally entangled across $A|B$. The success of the entanglement cooling is quantified by the ratio $R$ between the entanglement entropy before and after the cooling:
\be
R=\frac{S_{1}(\psi_{A}')}{S_{1}(\psi_{A})}
\ee
The lower the ratio $R$, the more successful the protocol. In Ref.~\cite{true_transitions_2022}, the authors explore the entanglement cooling problem concerning states formed by random Clifford+T circuits. They restrict themselves to Clifford operations in the cooling process and aim to identify good unitaries via a local greedy search. Importantly, they find that the efficacy of this entanglement cooling method decreases as the number of T-gates used to prepare the state increases. In particular, their numerical results indicate that Clifford unitaries suffice to cool the entanglement of a random $t$-doped state as long as $t\lesssim n$, but fails otherwise. Notably, the tools and methods developed in this paper give us an analytical way of understanding these results.

\begin{lemma}[Entanglement cooling via Clifford operations]
For any bipartition $A|B$, let $\ket{\psi}$ be a state generated by a Clifford+T circuit with $t\leq \frac{n}{2}$ T gates. There exists a Clifford circuit $C'$ such that, when applied on the state $\ket{\psi'}=C'\ket{\psi}$, the cooling ratio is $R=0$. That is, $\ket{\psi'}$ is a product state with respect to the bipartition $A|B$. Moreover, it is possible to find $C'$ using just $O(n)$ queries to $\ket{\psi}$.
\end{lemma}
\begin{proof}
Given query access to some state $\ket{\psi}$, we first learn the stabilizer generators of the state using $O(n)$ queries. We note that this state must have at least $n-t$ stabilizer generators (as opposed to the general bound $n-2lt$, which comes from the fact that the $T$ gate is diagonal~\cite{leone_learning_2024}). Then, using nullity distillation (\cref{alg:nullity}), we can find a Clifford unitary that maps $\ket{\psi}$ to $\ket{\psi'} 
=\ket{0}^{\otimes (n-t)} \otimes \ket{\psi''}$, where $\ket{\psi''}$ is supported on $t \leq \frac{n}{2}$ qubits. Also, at least one of the bipartitions has $\geq \frac{n}{2}$ qubits. Simply by applying swaps appropriately, we can move $\ket{\psi'}$ so that is supported only on this bipartition, with the rest of the system in a product of $\ket{0}$ states. Now, the resulting state is completely unentangled across the bipartition $A|B$.
\end{proof}

\parhead{Entanglement and magic phases in hybrid quantum circuits.} The second phenomenon we study revolves around hybrid quantum circuits, which are characterized by random circuits interspersed with measurements. These constructions have attracted significant attention in recent years~\cite{zabalo_critical_2020,gullans_dynamical_2020,oliviero_transitions_2021}, and a series of works have found a phase transition from volume-law entanglement to area-law entanglement in the states produced these hybrid quantum circuits~\cite{skinner_measurementinduced_2019,li_measurementdriven_2019,bao_theory_2020}. This motivated an inquiry into whether a similar phase transition exists for magic, and in Ref.~\cite{fux2023entanglementmagic}, the authors numerically demonstrated an ``entanglement-magic separation" in hybrid quantum circuits constructed with random local Clifford gates interspersed with both measurements and $T$-gates. While one might expect four phases corresponding to each quadrant of (volume law, area law)$\times$(entanglement, magic), only three were found: the \textit{volume law entanglement} and \textit{area law magic} phase was never realized. Intriguingly, this missing phase corresponds to ED states, where magic scales sub-extensively with entanglement. In our analysis, leveraging the toolkit developed in this paper, we provide analytical proof of why this should be the case. More precisely, we establish analytically that for a hybrid quantum circuit to exhibit entanglement-dominated states, the probability of applying $T$ gates must vanish exponentially with $n$. This provides a quantitative rationale for the absence of this phase in the numerical results of Ref.~\cite{fux2023entanglementmagic}, which used only a polynomial vanishing probability of applying $T$-gates with respect to $n$, insufficient to cool down magic towards the ED phase.

Since our focus lies in the transition ED-MD, we design a simplified version of hybrid quantum circuits characterized solely by volume law entangled states. Typically, random hybrid circuits consist of random local gates arranged in a brickwork structure. However, in our case, we employ global random Clifford circuits, pumping up entanglement to its maximum value while leaving magic invariant. This design choice allows us to specifically address the ``ratio" between magic and entanglement. In particular, we consider a hybrid quantum circuit composed of layers of \textit{random global} Clifford circuits interleaving the application of $R_z(\theta)$ gates (with random $\theta$) and measurements in the computational basis (see \cref{fig:sketch}). While we assume that the Clifford layers are applied with unit probability, at each layer we consider a probability $p_t$ for applying a $R_z(\theta)$-gate and $p_m$ for performing a measurement with $p_t+p_m\le 1$. On the other hand, we consider the probability of applying both at the same time to be zero. Thanks to the Clifford invariance of the problem, the Markov chain governing the evolution of the stabilizer nullity $\nu$ of the resultant state is extremely simple and can be mapped into a simple birth-death process.
\begin{figure}
    \centering
       {
       \graphicspath{{figures/}}
        \newcommand{\fs}{\scriptsize}
        \newcommand{\fss}{\footnotesize}
    \def\svgwidth{0.57\textwidth}
    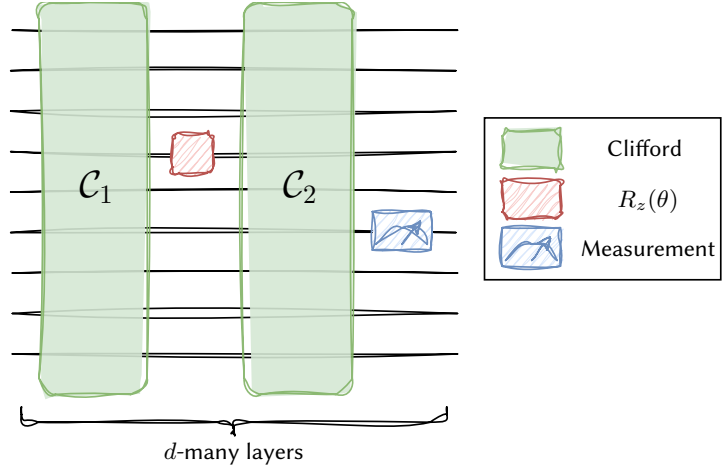
    }
    \caption{A schematic representation of a $t$-doped random circuit consists of alternating layers of random Clifford circuits, random diagonal gates $R_z(\theta)$ and measurements in the local $Z$ basis.}
    \label{fig:sketch}
\end{figure}

Before stating our result, we lay down some preliminary considerations. Consider a state $\ket{\psi}$ with a stabilizer group $G_{\psi}$. We apply a $R_z(\theta)$ gate on $\ket{\psi}$ sandwiched by a global \textit{random} Clifford circuit $\mathcal{C}_1$, so we define $\ket{\phi}=\mathcal{C}_1^\dagger R_z(\theta) \mathcal{C}_1 \ket{\psi}$. This sandwiching can always be done by mapping $\mathcal{C}_2 \to \mathcal{C}_2 \mathcal{C}_1$, which does not change the distribution of $\mathcal{C}_2$ because both $\mathcal{C}_1,\mathcal{C}_2$ are uniformly sampled from the Clifford group. Trivially, we know that $\abs{G_{\phi}} \ge \abs{G_{\psi}}/2$. However, since $\theta$ is chosen at random $\theta\in[0,2\pi]$, up to a set of measure zero, it holds that  $\abs{G_{\phi}}\le \abs{G_{\psi}}$. Therefore, we have $\abs{G_{\phi}}\in\{\abs{G_{\psi}},\abs{G_{\psi}}/2\}$ with probability 1. For a random Clifford $\mathcal{C}_1$, $\mathcal{C}_1^{\dagger}R_z(\theta) \mathcal{C}_1=e^{i\theta P}$ where $P$ is a random non-identity Pauli operator. Now, there are two cases: either $P$ commutes with all the elements in $G_{\psi}$, or $P$ anticommutes with $\abs{G_{\psi}}/2$ elements of $G_{\psi}$. In the first case, one has that $\abs{G_{\phi}}=\abs{G_{\psi}}$, while in the latter case one has $\abs{G_{\phi}}=\abs{G_{\psi}}/2$.
It is quite easy to be convinced that the probability that a random Pauli operator anticommutes with $\abs{G_{\psi}}/2$ elements of $G_{\psi}$ is $f(\nu)\coloneqq1-\frac{2^{n+\nu}-1}{4^n-1}$. Now, we analyze measurement in the computational basis. Consider a measurement sandwiched by a random Clifford $\mathcal{C}_2$ applied to a state $\ket{\psi}$ with stabilizer group $G_{\psi}$. The resultant states are $\ket{\phi_{\pm}}\propto (I\pm P)\ket{\phi}$ with $P$ being a random Pauli operator. Since we are interested only in what happens to the stabilizer group (and not its phases), without loss of generality, we can consider the $+1$ outcome. There are two situations that may arise, either $P$ commutes with all the elements of $G_{\psi}$ or $P$ anticommutes with $\abs{G_{\psi}}/2$ elements of $G_{\psi}$. In the latter case, it is easy to see that $\abs{G_{\phi}}=\abs{G_{\psi}}$~\cite{beverland_lower_2020}. In the former case, we are not yet ensured that the dimension of the stabilizer group of $\ket{\phi}$ doubles. In order for the stabilizer group to double, on top of commuting with all the elements in $G_{\psi}$, we must also have $P \not\in G_\psi$. The probability that $P$ commutes with $G_{\psi}$ is given by $1-f(\nu)$ as before, and if $P$ commutes with $G_{\psi}$ and given that both $P$ and $G_{\psi}$ are sandwiched by random Clifford operations, the probability that $P$ is generated by $G_{\psi}$ is simply $\abs{G_{\psi}}/d=2^{-\nu}$. Therefore, putting all together, the probability that a measurement doubles the size of $G_{\psi}$ is given by $(1-f(\nu))(1-2^{-\nu})$.  From this simple analysis, it becomes clear that the localizing power of the measurements, is much less than the spreading of magic given by the application of $R_z(\theta)$ gates. We make this intuition rigorous in the following theorem.
\begin{theorem}[Absence of entanglement-dominated states in hybrid quantum circuits]\label{th13} Consider infinite layers of the hybrid quantum circuit in \cref{fig:sketch} with $p_t$ the probability of applying a $R_z(\theta)$ and $p_m$ being the probability of applying a measurement in the computational basis. Define $r=p_t/p_m$, then for any $r=\omega(\exp(-n))$ the resulting state is in the MD phase.
\begin{proof}
As anticipated above, the stabilizer nullity at each layer can be modelled as a simple Markov chain -- specifically, a birth-death process. The birth probability $p_{+}(\nu)$ refers to the probability of an increasing stabilizer nullity, while the death probability $p_{-}(\nu)$ refers to the case where nullity decreases:
\begin{equation}
    \begin{aligned}
        p_{-}(\nu)&= p_m \times (1-f(\nu))(1-2^{-\nu}), \\ p_{+}(\nu)&=p_t \times f(\nu).        
    \end{aligned}
\end{equation}
Note that $p_{-}(0)=p_{+}(n)=0$, while the probability that the nullity stays the same is $1-p_{+}(\nu)-p_{-}(\nu)$. In the limit of infinite layers, we converge to the stationary distribution of a birth-death Markov process with birth probability $p_{+}$ and death probability $p_{-}$. The steady-state distribution of such a process is~\cite{karlin1957classification}:
\be
\pi(\nu)=\begin{cases}
\frac{p_{+}(0)}{Z}\quad\quad\quad\quad\quad\,\,\,\,\nu=0\\
\frac{1}{Z}\frac{p_{+}(0)\cdots p_{+}(\nu-1)}{p_{-}(1)\cdots p_{-}(\nu)}\quad \nu>0
\end{cases};\quad\quad Z=\sum_{\nu=0}^{n}\frac{p_{+}(0)\cdots p_{+}(\nu-1)}{p_{-}(1)\cdots p_{-}(\nu)}
\ee
Recalling the definition of $f(\nu)=1-\frac{2^{n+\nu}-1}{4^n-1}$, we approximate $f(\nu) \approx 1-2^{\nu-n}$ for simplicity. Therefore, we have $p_{-}(\nu)=\frac{p_{m}}{2^n}(2^{\nu}-1)$ and $p_{+}(\nu)=p_t(1-2^{\nu-n})$. Plugging the probabilities, the steady probability distribution $\pi(\nu)$ reads
\be\label{eq:pi-bound}
\pi(\nu)=\frac{r^{\nu}}{Z}\prod_{i=1}^{\nu}\frac{(2^n-2^{i-1})}{(2^{i}-1)} \geq \frac{r^\nu 2^{n \nu}}{Z \cdot 2^{(\nu+1)/2}} \prod_{i=1}^{\nu} (1-2^{i-n-1}) \geq \frac{r^\nu 2^{n \nu}}{4Z \cdot 2^{(\nu+1)/2}} \geq \frac{1}{6Z}\qty(\frac{2^n r}{\sqrt{2}})^\nu.
\ee
Since we aim to lower bound the expectation $\expval{\nu}$, it suffices to upper bound $Z$.
\begin{equation}\label{eq:zbound}
\begin{gathered}
        Z =p_{t}(1-2^{-n})+\sum_{\nu=1}^{n}r^{\nu}\prod_{i=1}^{\nu}\frac{2^n-2^{i-1}}{2^i-1} \leq p_t + \sum_{\nu=1}^n r^\nu \prod_{i=1}^\nu \frac{2^n-2^{i-1}}{2^i-1} \leq p_t + \sum_{\nu=1}^n r^\nu \prod_{i=0}^{\nu-1} \qty(2^{n-i}-1) \\
        \leq p_t + \sum_{\nu=1}^n r^\nu 2^{n\nu - \nu^2/2^2/2} \leq \sum_{\nu=0}^n r^\nu 2^{\nu(n-\nu/2)\nu^2/2} \leq O(2^{n^2/2}),
\end{gathered}
\end{equation}
where we assume $r \leq 1$ WLOG in the last line (if $r > 1$, clearly the expected stabilizer nullity can only decrease). Now, observe that we can combine \cref{eq:pi-bound,eq:zbound} to find $\pi(n) \geq \Omega(r^n 2^{n^2/2})$, so $\expval{\nu} \geq n \pi(n) \geq \Omega(n)$ when $r = \omega(\exp(-n))$. 

\end{proof}
\end{theorem}

The implications of the above theorem are twofold. Firstly, it provides an analytical explanation for the empirical observation in Ref.~\cite{fux2023entanglementmagic} that the fourth phase -- namely, the entanglement-dominated one -- is absent from random hybrid quantum circuits. Secondly, it firmly establishes that hybrid quantum circuits cannot generate ED states within a reasonable timeframe, since \cref{th13} says that the probability ratio must decrease exponentially with $n$, consequently the convergence time of the process scale at least exponentially with $n$. A natural question is whether it is possible to improve upon this result. The answer is yes, but only slightly. In the following, we design a slightly more optimistic scenario for which hybrid quantum circuits create ED states.  This is achieved without an exponentially small probability of applying $R_z(\theta)$ gates, but rather by adapting probabilities to the stabilizer nullity of the state at a given layer, with the ratio being $r(\nu)$ a function of $\nu$.

\begin{theorem}[Producing entanglement-dominated states via adaptive measurements]\label{th14} Consider infinite layers of the hybrid quantum circuit in \cref{fig:sketch} with $p_t(\nu)$ the probability of applying a $R_z(\theta)$ and $p_m(\nu)$ being the probability of applying a measurement in the computational basis, both depending on the nullity $\nu$ of the state at each layer. Defining $r(\nu)=p_t(\nu-1)/p_{m}(\nu)$ for $\nu>0$ and  $r(\nu)=p_t(0)$, we obtain:
\begin{itemize}
        \item For $r(\nu)=\omega(\exp(-k))$ for $k=n-\nu$, then the state is within the MD phase.
        \item For $r(\nu)=O(\exp(-k))$ for $k=n-\nu$ then the resulting state is within the ED phase.
    \end{itemize}

\begin{proof}
The proof uses the same notation and techniques of \cref{th13}. The steady-state distribution $\pi(\nu)$ is
\be
\pi(\nu)=\frac{1}{Z}\left(\prod_{i=1}^{\nu}r(i)\right)\frac{\prod_{i=0}^{\nu-1}(2^n-2^i)}{\prod_{i=1}^{\nu}(2^i-1)}=\frac{1}{Z}\prod_{i=1}^{\nu}r(i)\frac{(2^n-2^{i-1})}{(2^{i}-1)},\label{pinu}
\ee
where we defined $r(i)=p_{t}(i-1)/p_{m}(i)$ for $i\ge1$. We only need to establish that for $r(\nu)=O(\exp(-(\nu-n))$, ED states are realized. We will proceed with a constructive proof as follows. Starting from \cref{pinu}, we choose
\be
r(i)=r_{0}\frac{(2^{i}-1)}{(2^n-2^{i-1})},\quad i\in[1,n]\label{choice}
\ee
with some $r_{0}\neq 1$ to be fixed later, so we have $\pi(\nu)=r_{0}^{\nu}/Z$. With such a choice, we compute $Z$:
\be
Z=Z\pi(0)+\sum_{\nu=1}^{n}r_{0}^{\nu}=Z\pi(0)+r_0\frac{r_{0}^{n}-1}{r_0-1}\implies Z=\frac{1}{1-\pi(0)}r_0\frac{r_{0}^{n}-1}{r_0-1}.
\ee
We have the following consistency condition
\be
1=(1-\pi(0))\frac{1}{r_0}\frac{r_0-1}{r_{0}^n-1}\implies \pi(0)=1-r_{0}\frac{1-r_{0}^{n}}{1-r_0}.
\ee
If we fix $r_{0}=1/2$, we get $\pi(0)=1/2^{n}$. We can thus express the steady probability simply as
\be
\pi_n(\nu)=\begin{cases}2^{-n},\quad \nu=0\\
2^{-\nu},\quad \nu\in[1,n]
\end{cases}
\ee
from which we can bound the average stabilizer nullity $\langle\nu\rangle_{\nu\sim\pi_{n}} \leq \sum_{\nu=1}^\infty \nu 2^{-\nu} = 2$, which realizes the ED phase. In this case, recall that through \cref{choice}, we chose $r(\nu)=\frac{1}{2}\frac{2^\nu-1}{2^{n}-2^{n-\nu}}=\Theta(2^{-(n-\nu)})$. For the case where $r(\nu)=\omega(\exp(-(n-\nu)))$ and defining $r_0=\min_{\nu\in[1,n]}r(\nu)$, we have $r_{0}=\omega(\exp(-n))$, so a proof similar to the one presented in \cref{th13} follows.
\end{proof}

\end{theorem}

In \cref{th14}, we demonstrated that by adapting the probabilities $p_t$ and $p_m$ based on the outcome of the stabilizer nullity measurement (which can be conducted efficiently, see \cref{sec:unknownstates}), the outlook for the presence of an ED phase is more optimistic compared to the result in \cref{th13}. Specifically, by tuning the ratio between the probability of applying $T$-gates and measurements to decrease exponentially in the size of the stabilizer group $n-\nu$ rather than in $n$, hybrid quantum circuits can indeed generate ED states. While this scaling of the ratio is significantly better than $\exp(-n)$, particularly in the high stabilizer nullity regime, achieving convergence to the ED phase, i.e., in the regime of small nullity $\nu$, still requires the ratio to be exponentially small for assurance of convergence. Thus, although this approach slightly accelerates the convergence to ED states, the convergence time remains exponential in $n$. 

Despite these findings, the emergence of ED states in hybrid quantum circuits appears unlikely. However, Ref.~\cite{bejan2023dynamical} introduces another scenario involving correlated monitoring, where measurements are conducted immediately after a non-Clifford gate with a certain probability. This technique can evidently result in ED states. In fact, with perfect correlation (i.e., a measurement following every non-Clifford gate application), the entire circuit effectively behaves as a Clifford circuit. While this scenario provides valuable insights, it assumes perfect knowledge of the applied gates from the observer's perspective, which may not always be realistic. In contrast, our adaptive measurement strategy allows for determining nullity in polynomial time (see \cref{sec:unknownstates}), but identifying the occurrence and location of a $T$-gate remains challenging~\cite{vandewetering2023optimising}. Nonetheless, investigating the transition from MD to ED phases based on the observer's monitoring knowledge is an intriguing avenue for future research, as discussed in Ref.~\cite{bejan2023dynamical}.

\section{Pseudocode}\label{sec:Pseudocodesfornumerics}
Here we outline algorithms based on the main results of this work. These algorithms are designed for three tasks: exact calculation of the $2$-Rényi entropy, entanglement monitoring, and efficient evaluation of a multipartite entanglement witness.

\medskip

\parhead{2-R\'enyi entropy estimation for $t$-doped stabilizer states.} \cref{alg:2renyi} exactly calculates the 2-R\'enyi entanglement entropy of a state $\psi$. We assume we know the stabilizer tableau $\mathcal{T}_S$ of the state $\psi$ as well as its bad generators $(h_i, \tr(h_i \psi))$. We first apply the map $\mathcal{P}_B: g_A \otimes g_B \to \id \otimes g_B$ to each row of the tableau. Then since $G_A=\ker(\mathcal{P}_B(\mathcal{T}_S))$, to estimate $|S_A|$, one computes $\dim(\ker(\mathcal{P}_B(\mathcal{T}_S)))$. Then, the map $\mathcal{P}_B(\mathcal{T}_S)$ is applied to the bad generators $h_i$, and if they are within $\rowspan(\mathcal{P}_B(\mathcal{T}_S))$, their corresponding values $\tr(h_i\psi)$ are summed. This yields $r=\sum_{\mathcal{P}_B(h_i)\in \rowspan(\mathcal{P}_B(\mathcal{T}_S))}\tr(h_i\psi)$. Finally, the $2$-Rényi entropy is $S_2(\psi_A)=n_A-\lvert S_A\rvert-\log(r)$. The overall computational complexity is $O(n^3 2^{4lt})$, because the linear algebraic operations each have a runtime $O(n^3)$, and we loop over $i=0,\ldots,k$, where $k$ can be as large as $k=2^{4lt}$.
\begin{algorithm}[H]
\caption{Exact calculation of the $2$-R\'enyi entanglement entropy across a bipartition $A|B$}\label{alg:2renyi}
\begin{algorithmic}[1]
\Require{A representation of $\psi$ in terms of its stabilizer tableau $\mathcal{T}_S$ and bad generators $\qty{(h_i, \tr(h_i \psi)) \mid i=0,\ldots,k}$.}
\Function{2R\'enyiEntropy}{$\psi$}
\State $\abs{S_A} \gets \dim(\ker(\mathcal{P}_B(\mathcal{T}_S)))$ \Comment{$\mathcal{P}_B$ is the map that takes $g_A \otimes g_B \to \id \otimes g_B$}
\State $r \gets 0$
\For{$i \gets 0$ to $k$} 
\State{$\vec{y}_i \in \mathbb{F}_2^{2n} \gets $ symplectic representation of $h_i$}
\If{$\mathcal{P}_B(\vec{y}_i) \in \rowspan(\mathcal{P}_B(\mathcal{T}_S))$} \Comment{$\mathcal{P}_B(\vec{y}_i)$ is in the row space of $\mathcal{P}_B(\mathcal{T}_G) \iff \delta_i=1$}
\State{$r \gets r + \tr^2(h_i \psi)$}
\EndIf
\EndFor
\State \Return {$n_A - \abs{S_A} - \log(r)$}
\EndFunction
\end{algorithmic}
\end{algorithm}

\parhead{Efficient entanglement monitoring in $t$-doped Clifford circuits.}
Here, we present an algorithm for classically monitoring the entanglement of $t$-doped Clifford circuits, \emph{without running anything on quantum hardware}. The task is to compute the entanglement entropy across a bipartition $A\vert B$ of a state generated by a sequence of gates $\{U_1,U_2,\ldots U_{|C|}\}$, where $|C|$ is the number of gates in the circuit. To start, we initialize the stabilizer tableau $\mathcal{T}_S$ corresponding to the stabilizer set $S\coloneqq \{Z_1,Z_2,\ldots,Z_n\}$ of the state $\ketbra{0}{0}^{\otimes n}$, to which we will sequentially apply the gates $U_i$. If $U_i$ is a Clifford gate, we update the tableau $\mathcal{T}_S$ using the stabilizer formalism. Otherwise, letting $X_i$ be the sites on which the non-Clifford gate acts, we perform a row reduction of $\mathcal{T}_S$ to reformat the tableau such that that all generators except the first $2|X_i|$ act trivially on the sites $X_i$, then we remove the first $2|X_i|$ generators from $S$. After applying every gate, we compute $\abs{S_A}=\dim(\ker(\mathcal{P}_B(\mathcal{T}_S)))$, and finally return $n_A-\abs{S_A}+2lt$ as our entanglement estimate. The algorithm has computational complexity $O(\abs{C} n^3)$, as the linear algebraic manipulations each have runtime $O(n^3)$ and they are repeated $\abs{C}$ times. 
\begin{algorithm}[H]
\caption{Entanglement monitoring}\label{alg:monitor}
\begin{algorithmic}[1]
\Require{A $t$-doped circuit $C$ in terms of a sequence of gates $\qty{U_1, U_2, \ldots, U_{\abs{C}}}$.}
\Require{A bipartition $A|B$}
\Ensure{$\tilde{S}_\alpha(\psi_A)$: an estimate of the entanglement entropy of the state $\psi$ generated by $C$ across $A|B$, such that $\abs*{\tilde{S}_\alpha(\psi_A)-S_\alpha(\psi_A)} \leq 2lt$}
\Function{EstimateEntropy}{$C$, $A|B$}
\State $S \gets \qty{Z_1, Z_2, \ldots, Z_n}$ \Comment{Initial state is $\ketbra{0}^{\otimes n}$}
\For{$i \gets 1, \ldots, \abs{C}$}
\If{$U_i$ is a Clifford gate}
\For{$g_j \in S$}
\State $g_j \gets U_i g_j U_i^\dagger$
\EndFor
\Else
\State $X_i \gets$ the sites on which $U_i$ acts (note that $l \coloneqq \max_i \abs{X_i}$)
\State Using row reduction on $\mathcal{T}_S$, reformat the stabilizer generators such that all except the first $2\abs{X_i}$ act as the identity on $X_i$
\State Remove the first $2\abs{X_i}$ generators from $S$
\EndIf
\EndFor
\State $\abs{S_A} \gets \dim(\ker(\mathcal{P}_B(\mathcal{T}_S)))$ \Comment{$\mathcal{P}_B$ is the map that takes $g_A \otimes g_B \to \id \otimes g_B$}
\State \Return $n_A - \abs{S_A} + 2lt$ \Comment{See \cref{lem:salpha-bound}.}
\EndFunction
\end{algorithmic}
\end{algorithm}    

\parhead{Efficient estimation for multipartite entanglement witness.} To efficiently estimate a multipartite entanglement witness, we have developed \cref{alg:witness}. In words, the algorithm takes as input the stabilizer tableau $\mathcal{T}_S$ associated with the target state $\ketbra{\psi}{\psi}$, a list of bipartitions $\mathcal{B} = \qty{(A_i,B_i) \mid i=1,\ldots,\abs{\mathcal{B}}}$, a threshold error $\epsilon$, and a maximum failure probability $\delta$. The algorithm then proceeds as follows:
It computes the minimum bound on the entropy across the bipartitions, $M= \min_i \lfloor\mathcal{E}(\psi;\,A_i|B_i) - lt\rfloor$, then samples a random stabilizer $P\in G$, and measures the expectation value of the observable $P$ over $\psi$. This process is repeated $N=\left\lceil \frac{2 \log(2/\delta)}{(1-\epsilon - 2^{-M})^2}\right\rceil$ times, and one computes $\tilde{\Pi}=\sum_{i}b_i/N$. At the end of the computation, the algorithm returns $2^{-M}-\tilde{\Pi}$. The algorithm has computational complexity $O(n^3\left\lceil \frac{2 \log(2/\delta)}{(1-\epsilon - 2^{-M})^2}\right\rceil +\mathcal{B}n^3)$. Estimating $\mathcal{E}(\psi;\,A_i|B_i)$ for each partition requires at most $O(n^3)$ resources, and such a task has to be repeated $\abs{\mathcal{B}}$ times, then the runtime to find the minimum value has runtime $O(\abs{\mathcal{B}})$. The loop, on the other hand, has computational complexity $O(n^3\left\lceil \frac{2 \log(2/\delta)}{(1-\epsilon - 2^{-M})^2}\right\rceil)$ due to the repeated sampling of Pauli operators. 
\begin{algorithm}[H]
\caption{Efficient estimation protocol for a multipartite entanglement witness}\label{alg:witness}
\begin{algorithmic}[1]
\Require{A representation of the target state $\ket{\psi}$ in terms of its stabilizer tableau $\mathcal{T}_S \in \mathbb{F}_2^{\abs{S} \times 2n}$}
\Require{A list of bipartitions $\mathcal{B} = \qty{(A_i,B_i) \mid i=1,\ldots,\abs{\mathcal{B}}}$}
\Require{An upper bound on the state preparation error $\epsilon \geq \frac{1}{2} \norm{\psi-\ketbra{\psi}}_1$ and a maximum failure probability $\delta$}
\Function{MultipartiteWitness}{$\psi,\mathcal{B},\epsilon,\delta$}
\State $M \gets \min_i \lfloor\mathcal{E}(\psi;\,A_i|B_i) - lt\rfloor$ 
\State $\tilde{\Pi} \gets 0$ \Comment{Running estimate of $\tr(\Pi \rho)$}
\State $N \gets \left\lceil \frac{2 \log(2/\delta)}{(1-\epsilon - 2^{-M})^2}\right\rceil$
\For{$i \gets 1, \ldots, N$}
\State $\vec{x} \gets $ a random binary vector $\mathbb{F}_2^{\abs{S}}$
\State $P \in \mathbb{P}_n \gets \vec{x}^T \mathcal{T}_S$ \Comment{Select a random stabilizer $P \in G$}
\State $b \in \qty{-1,1} \gets$ \Call{Measure}{$\psi, P$} \Comment{Measure the two-outcome Pauli observable $P$}
\State $\tilde{\Pi} \gets \tilde{\Pi} + b/N$
\EndFor
\State \Return{$2^{-M} - \tilde{\Pi}$}
\EndFunction
\end{algorithmic}
\end{algorithm}

\let\oldaddcontentsline\addcontentsline%

\renewcommand{\addcontentsline}[3]{}%
\section*{Acknowledgments} Special thanks to Ludovico Lami for inspiring our exploration of entanglement dilution, suggesting delving into entanglement witnessing, and providing guidance through the literature of entanglement theory.  We thank Sumeet Khatri for insightful comments on birth/death chains. We express our appreciation to Jacopo Rizzo for recommending a perspective on our results through the lens of entanglement irreversibility. Additionally, we acknowledge Antonio and Francesco A. Mele for their valuable comments and engaging discussions. We also thank Susanne Yelin, Pablo Bonilla, Nazl\i \ U\u{g}ur K\"oyl\"uo\u{g}lu, and Varun Menon for helpful discussions and suggestions on potential applications of the ED-MD framework. We thank the Unitary Fund for their support. A.G. acknowledges support from the NSF through the Q-IDEAS HDR, as well as the AWS Generation Q Fund at the Harvard Quantum Initiative. S.F.E.O. acknowledges support from PNRR MUR project PE0000023-NQSTI. L.L. is funded through the Munich Quantum Valley project (MQV-K8) by Bayerisches Staatsministerium für Wissenschaft und Kunst.

\newcommand{\etalchar}[1]{$^{#1}$}

\let\addcontentsline\oldaddcontentsline%

\end{document}

%% file: figures/distill.pdf_tex
\begingroup%
  \makeatletter%
  \providecommand\color[2][]{%
    \errmessage{(Inkscape) Color is used for the text in Inkscape, but the package 'color.sty' is not loaded}%
    \renewcommand\color[2][]{}%
  }%
  \providecommand\transparent[1]{%
    \errmessage{(Inkscape) Transparency is used (non-zero) for the text in Inkscape, but the package 'transparent.sty' is not loaded}%
    \renewcommand\transparent[1]{}%
  }%
  \providecommand\rotatebox[2]{#2}%
  \newcommand*\fsize{\dimexpr\f@size pt\relax}%
  \newcommand*\lineheight[1]{\fontsize{\fsize}{#1\fsize}\selectfont}%
  \ifx\svgwidth\undefined%
    \setlength{\unitlength}{519bp}%
    \ifx\svgscale\undefined%
      \relax%
    \else%
      \setlength{\unitlength}{\unitlength * \real{\svgscale}}%
    \fi%
  \else%
    \setlength{\unitlength}{\svgwidth}%
  \fi%
  \global\let\svgwidth\undefined%
  \global\let\svgscale\undefined%
  \makeatother%
  \begin{picture}(1,0.80202312)%
    \lineheight{1}%
    \setlength\tabcolsep{0pt}%
    \put(0,0){\includegraphics[width=\unitlength,page=1]{distill.pdf}}%
    \put(0.22687861,0.04190751){\color[rgb]{0.52156863,0.70980392,0.41568627}\makebox(0,0)[t]{\lineheight{1.25}\smash{\begin{tabular}[t]{c}\sf State preparation\end{tabular}}}}%
    \put(0,0){\includegraphics[width=\unitlength,page=2]{distill.pdf}}%
    \put(0.47109827,0.04190751){\color[rgb]{0.42352941,0.56078431,0.74901961}\makebox(0,0)[t]{\lineheight{1.25}\smash{\begin{tabular}[t]{c}\sf Entanglement\end{tabular}}}}%
    \put(0,0){\includegraphics[width=\unitlength,page=3]{distill.pdf}}%
    \put(0.88439306,0.77456647){\color[rgb]{0,0,0}\makebox(0,0)[t]{\lineheight{1.25}\smash{\begin{tabular}[t]{c}\Large \textsf{\textbf{A}}\end{tabular}}}}%
    \put(0,0){\includegraphics[width=\unitlength,page=4]{distill.pdf}}%
    \put(0.81213873,0.69508671){\color[rgb]{0,0,0}\makebox(0,0)[t]{\lineheight{1.25}\smash{\begin{tabular}[t]{c}$\ket{\sigma'}$\end{tabular}}}}%
    \put(0,0){\includegraphics[width=\unitlength,page=5]{distill.pdf}}%
    \put(0.81213873,0.6517341){\color[rgb]{0,0,0}\makebox(0,0)[t]{\lineheight{1.25}\smash{\begin{tabular}[t]{c}$\ket{0}$\end{tabular}}}}%
    \put(0,0){\includegraphics[width=\unitlength,page=6]{distill.pdf}}%
    \put(0.75578035,0.6083815){\color[rgb]{0,0,0}\makebox(0,0)[t]{\lineheight{1.25}\smash{\begin{tabular}[t]{c}$\ket{0}$\end{tabular}}}}%
    \put(0,0){\includegraphics[width=\unitlength,page=7]{distill.pdf}}%
    \put(0.68930636,0.5650289){\color[rgb]{0,0,0}\makebox(0,0)[t]{\lineheight{1.25}\smash{\begin{tabular}[t]{c}$\ket{0}$\end{tabular}}}}%
    \put(0,0){\includegraphics[width=\unitlength,page=8]{distill.pdf}}%
    \put(0.55346821,0.5216763){\color[rgb]{0,0,0}\makebox(0,0)[t]{\lineheight{1.25}\smash{\begin{tabular}[t]{c}$\ket{\phi_+}$\end{tabular}}}}%
    \put(0,0){\includegraphics[width=\unitlength,page=9]{distill.pdf}}%
    \put(0.55346821,0.4783237){\color[rgb]{0,0,0}\makebox(0,0)[t]{\lineheight{1.25}\smash{\begin{tabular}[t]{c}$\ket{\phi_+}$\end{tabular}}}}%
    \put(0,0){\includegraphics[width=\unitlength,page=10]{distill.pdf}}%
    \put(0.55346821,0.40462428){\color[rgb]{0,0,0}\makebox(0,0)[t]{\lineheight{1.25}\smash{\begin{tabular}[t]{c}$\ket{\phi_+}$\end{tabular}}}}%
    \put(0,0){\includegraphics[width=\unitlength,page=11]{distill.pdf}}%
    \put(0.55346821,0.36271676){\color[rgb]{0,0,0}\makebox(0,0)[t]{\lineheight{1.25}\smash{\begin{tabular}[t]{c}$\ket{\phi_+}$\end{tabular}}}}%
    \put(0,0){\includegraphics[width=\unitlength,page=12]{distill.pdf}}%
    \put(0.68786127,0.31936416){\color[rgb]{0,0,0}\makebox(0,0)[t]{\lineheight{1.25}\smash{\begin{tabular}[t]{c}$\ket{0}$\end{tabular}}}}%
    \put(0,0){\includegraphics[width=\unitlength,page=13]{distill.pdf}}%
    \put(0.76734104,0.27601156){\color[rgb]{0,0,0}\makebox(0,0)[t]{\lineheight{1.25}\smash{\begin{tabular}[t]{c}$\ket{0}$\end{tabular}}}}%
    \put(0,0){\includegraphics[width=\unitlength,page=14]{distill.pdf}}%
    \put(0.85404624,0.19219653){\color[rgb]{0,0,0}\makebox(0,0)[t]{\lineheight{1.25}\smash{\begin{tabular}[t]{c}$\ket{\sigma'}$\end{tabular}}}}%
    \put(0,0){\includegraphics[width=\unitlength,page=15]{distill.pdf}}%
    \put(0.57369942,0.77456647){\color[rgb]{0.46666667,0.52941176,0.72156863}\makebox(0,0)[t]{\lineheight{1.25}\smash{\begin{tabular}[t]{c}\sf Entanglement dilution\end{tabular}}}}%
    \put(0,0){\includegraphics[width=\unitlength,page=16]{distill.pdf}}%
    \put(0.88439306,0.10693642){\color[rgb]{0,0,0}\makebox(0,0)[t]{\lineheight{1.25}\smash{\begin{tabular}[t]{c}\Large \textsf{\textbf{B}}\end{tabular}}}}%
    \put(0,0){\includegraphics[width=\unitlength,page=17]{distill.pdf}}%
    \put(0.47109827,0.01589595){\color[rgb]{0.42352941,0.56078431,0.74901961}\makebox(0,0)[t]{\lineheight{1.25}\smash{\begin{tabular}[t]{c}\sf distillation\end{tabular}}}}%
    \put(0,0){\includegraphics[width=\unitlength,page=18]{distill.pdf}}%
    \put(0.64595376,0.04190751){\color[rgb]{0.61568627,0.48627451,0.6745098}\makebox(0,0)[t]{\lineheight{1.25}\smash{\begin{tabular}[t]{c}\sf Nullity\end{tabular}}}}%
    \put(0,0){\includegraphics[width=\unitlength,page=19]{distill.pdf}}%
    \put(0.64595376,0.01589595){\color[rgb]{0.61568627,0.48627451,0.6745098}\makebox(0,0)[t]{\lineheight{1.25}\smash{\begin{tabular}[t]{c}\sf distillation\end{tabular}}}}%
  \end{picture}%
\endgroup%

%% file: figures/nu-compress.pdf_tex
\begingroup%
  \makeatletter%
  \providecommand\color[2][]{%
    \errmessage{(Inkscape) Color is used for the text in Inkscape, but the package 'color.sty' is not loaded}%
    \renewcommand\color[2][]{}%
  }%
  \providecommand\transparent[1]{%
    \errmessage{(Inkscape) Transparency is used (non-zero) for the text in Inkscape, but the package 'transparent.sty' is not loaded}%
    \renewcommand\transparent[1]{}%
  }%
  \providecommand\rotatebox[2]{#2}%
  \newcommand*\fsize{\dimexpr\f@size pt\relax}%
  \newcommand*\lineheight[1]{\fontsize{\fsize}{#1\fsize}\selectfont}%
  \ifx\svgwidth\undefined%
    \setlength{\unitlength}{374.25bp}%
    \ifx\svgscale\undefined%
      \relax%
    \else%
      \setlength{\unitlength}{\unitlength * \real{\svgscale}}%
    \fi%
  \else%
    \setlength{\unitlength}{\svgwidth}%
  \fi%
  \global\let\svgwidth\undefined%
  \global\let\svgscale\undefined%
  \makeatother%
  \begin{picture}(1,0.57915832)%
    \lineheight{1}%
    \setlength\tabcolsep{0pt}%
    \put(0,0){\includegraphics[width=\unitlength,page=1]{nu-compress.pdf}}%
    \put(0.03406814,0.30460922){\color[rgb]{0,0,0}\makebox(0,0)[t]{\lineheight{1.25}\smash{\begin{tabular}[t]{c}\large $\ket{\psi}$\end{tabular}}}}%
    \put(0,0){\includegraphics[width=\unitlength,page=2]{nu-compress.pdf}}%
    \put(0.28657315,0.44288577){\color[rgb]{0,0,0}\makebox(0,0)[t]{\lineheight{1.25}\smash{\begin{tabular}[t]{c}$\ket{0}$\end{tabular}}}}%
    \put(0,0){\includegraphics[width=\unitlength,page=3]{nu-compress.pdf}}%
    \put(0.38677355,0.38476954){\color[rgb]{0,0,0}\makebox(0,0)[t]{\lineheight{1.25}\smash{\begin{tabular}[t]{c}$\ket{0}$\end{tabular}}}}%
    \put(0,0){\includegraphics[width=\unitlength,page=4]{nu-compress.pdf}}%
    \put(0.4749499,0.32865731){\color[rgb]{0,0,0}\makebox(0,0)[t]{\lineheight{1.25}\smash{\begin{tabular}[t]{c}$\ket{0}$\end{tabular}}}}%
    \put(0,0){\includegraphics[width=\unitlength,page=5]{nu-compress.pdf}}%
    \put(0.29859719,0.54509018){\color[rgb]{0.50980392,0.70196078,0.4}\makebox(0,0)[t]{\lineheight{1.25}\smash{\begin{tabular}[t]{c}\sf Encodes $\Pi$\end{tabular}}}}%
    \put(0,0){\includegraphics[width=\unitlength,page=6]{nu-compress.pdf}}%
    \put(0.56513026,0.54308617){\color[rgb]{0.72156863,0.32941176,0.31372549}\makebox(0,0)[t]{\lineheight{1.25}\smash{\begin{tabular}[t]{c}\sf Encodes $H$\end{tabular}}}}%
    \put(0,0){\includegraphics[width=\unitlength,page=7]{nu-compress.pdf}}%
    \put(0.31863727,0.09018036){\color[rgb]{0,0,0}\makebox(0,0)[t]{\lineheight{1.25}\smash{\begin{tabular}[t]{c}\sf Nullity distillation\end{tabular}}}}%
    \put(0,0){\includegraphics[width=\unitlength,page=8]{nu-compress.pdf}}%
    \put(0.40280561,0.02204409){\color[rgb]{0,0,0}\makebox(0,0)[t]{\lineheight{1.25}\smash{\begin{tabular}[t]{c}\sf State preparation\end{tabular}}}}%
    \put(0,0){\includegraphics[width=\unitlength,page=9]{nu-compress.pdf}}%
    \put(0.56382766,0.21823647){\color[rgb]{0,0,0}\rotatebox{90}{\makebox(0,0)[t]{\lineheight{1.25}\smash{\begin{tabular}[t]{c}\small $\nu$ qubits\end{tabular}}}}}%
    \put(0,0){\includegraphics[width=\unitlength,page=10]{nu-compress.pdf}}%
    \put(0.65931864,0.27254509){\color[rgb]{0,0,0}\makebox(0,0)[t]{\lineheight{1.25}\smash{\begin{tabular}[t]{c}$\ket{0}$\end{tabular}}}}%
    \put(0,0){\includegraphics[width=\unitlength,page=11]{nu-compress.pdf}}%
    \put(0.65931864,0.22044088){\color[rgb]{0,0,0}\makebox(0,0)[t]{\lineheight{1.25}\smash{\begin{tabular}[t]{c}$\ket{0}$\end{tabular}}}}%
    \put(0,0){\includegraphics[width=\unitlength,page=12]{nu-compress.pdf}}%
    \put(0.65931864,0.16833667){\color[rgb]{0,0,0}\makebox(0,0)[t]{\lineheight{1.25}\smash{\begin{tabular}[t]{c}$\ket{0}$\end{tabular}}}}%
    \put(0,0){\includegraphics[width=\unitlength,page=13]{nu-compress.pdf}}%
    \put(0.90781563,0.31262525){\color[rgb]{0,0,0}\makebox(0,0)[t]{\lineheight{1.25}\smash{\begin{tabular}[t]{c}\sf Clifford\end{tabular}}}}%
    \put(0,0){\includegraphics[width=\unitlength,page=14]{nu-compress.pdf}}%
    \put(0.91182365,0.26052104){\color[rgb]{0,0,0}\makebox(0,0)[t]{\lineheight{1.25}\smash{\begin{tabular}[t]{c}\sf Non-Clifford\end{tabular}}}}%
  \end{picture}%
\endgroup%

%% file: figures/cliff-t.pdf_tex
\begingroup%
  \makeatletter%
  \providecommand\color[2][]{%
    \errmessage{(Inkscape) Color is used for the text in Inkscape, but the package 'color.sty' is not loaded}%
    \renewcommand\color[2][]{}%
  }%
  \providecommand\transparent[1]{%
    \errmessage{(Inkscape) Transparency is used (non-zero) for the text in Inkscape, but the package 'transparent.sty' is not loaded}%
    \renewcommand\transparent[1]{}%
  }%
  \providecommand\rotatebox[2]{#2}%
  \newcommand*\fsize{\dimexpr\f@size pt\relax}%
  \newcommand*\lineheight[1]{\fontsize{\fsize}{#1\fsize}\selectfont}%
  \ifx\svgwidth\undefined%
    \setlength{\unitlength}{422.25bp}%
    \ifx\svgscale\undefined%
      \relax%
    \else%
      \setlength{\unitlength}{\unitlength * \real{\svgscale}}%
    \fi%
  \else%
    \setlength{\unitlength}{\svgwidth}%
  \fi%
  \global\let\svgwidth\undefined%
  \global\let\svgscale\undefined%
  \makeatother%
  \begin{picture}(1,0.51509769)%
    \lineheight{1}%
    \setlength\tabcolsep{0pt}%
    \put(0,0){\includegraphics[width=\unitlength,page=1]{cliff-t.pdf}}%
    \put(0.18827709,0.48490231){\color[rgb]{0,0,0}\makebox(0,0)[t]{\lineheight{1.25}\smash{\begin{tabular}[t]{c}\sf (a) \ $t$-doped Clifford Circuit $U$\end{tabular}}}}%
    \put(0,0){\includegraphics[width=\unitlength,page=2]{cliff-t.pdf}}%
    \put(0.73889876,0.48490231){\color[rgb]{0,0,0}\makebox(0,0)[t]{\lineheight{1.25}\smash{\begin{tabular}[t]{c}\sf (b) \ Algebraic structure of $t$-doped circuits\end{tabular}}}}%
    \put(0,0){\includegraphics[width=\unitlength,page=3]{cliff-t.pdf}}%
    \put(0.74166075,0.225){\color[rgb]{0,0,0}\rotatebox{90}{\makebox(0,0)[t]{\lineheight{1.25}\smash{\begin{tabular}[t]{c}\fss $\leq 2lt$ qubits\end{tabular}}}}}%
    \put(0,0){\includegraphics[width=\unitlength,page=4]{cliff-t.pdf}}%
    \put(0.48312611,0.26998224){\color[rgb]{0,0,0}\makebox(0,0)[t]{\lineheight{1.25}\smash{\begin{tabular}[t]{c}\large $U=$\end{tabular}}}}%
    \put(0,0){\includegraphics[width=\unitlength,page=5]{cliff-t.pdf}}%
    \put(0.39431616,0.03197158){\color[rgb]{0,0,0}\makebox(0,0)[t]{\lineheight{1.25}\smash{\begin{tabular}[t]{c}\sf Clifford\end{tabular}}}}%
    \put(0,0){\includegraphics[width=\unitlength,page=6]{cliff-t.pdf}}%
    \put(0.65541741,0.03197158){\color[rgb]{0,0,0}\makebox(0,0)[t]{\lineheight{1.25}\smash{\begin{tabular}[t]{c}\sf Non-Clifford\end{tabular}}}}%
  \end{picture}%
\endgroup%

%% file: figures/venn.pdf_tex
\begingroup%
  \makeatletter%
  \providecommand\color[2][]{%
    \errmessage{(Inkscape) Color is used for the text in Inkscape, but the package 'color.sty' is not loaded}%
    \renewcommand\color[2][]{}%
  }%
  \providecommand\transparent[1]{%
    \errmessage{(Inkscape) Transparency is used (non-zero) for the text in Inkscape, but the package 'transparent.sty' is not loaded}%
    \renewcommand\transparent[1]{}%
  }%
  \providecommand\rotatebox[2]{#2}%
  \newcommand*\fsize{\dimexpr\f@size pt\relax}%
  \newcommand*\lineheight[1]{\fontsize{\fsize}{#1\fsize}\selectfont}%
  \ifx\svgwidth\undefined%
    \setlength{\unitlength}{512.25bp}%
    \ifx\svgscale\undefined%
      \relax%
    \else%
      \setlength{\unitlength}{\unitlength * \real{\svgscale}}%
    \fi%
  \else%
    \setlength{\unitlength}{\svgwidth}%
  \fi%
  \global\let\svgwidth\undefined%
  \global\let\svgscale\undefined%
  \makeatother%
  \begin{picture}(1,0.45827233)%
    \lineheight{1}%
    \setlength\tabcolsep{0pt}%
    \put(0,0){\includegraphics[width=\unitlength,page=1]{venn.pdf}}%
    \put(0.66032211,0.1522694){\color[rgb]{0,0,0}\makebox(0,0)[t]{\lineheight{1.25}\smash{\begin{tabular}[t]{c}\small $t$-doped\end{tabular}}}}%
    \put(0,0){\includegraphics[width=\unitlength,page=2]{venn.pdf}}%
    \put(0.74963397,0.07906296){\color[rgb]{0,0,0}\makebox(0,0)[t]{\lineheight{1.25}\smash{\begin{tabular}[t]{c}$\nu$-compressible\end{tabular}}}}%
    \put(0,0){\includegraphics[width=\unitlength,page=3]{venn.pdf}}%
    \put(0.24011713,0.42752562){\color[rgb]{0,0,0}\makebox(0,0)[t]{\lineheight{1.25}\smash{\begin{tabular}[t]{c}\large \sf (a) \ Pictoral definition of ED and MD phases\end{tabular}}}}%
    \put(0,0){\includegraphics[width=\unitlength,page=4]{venn.pdf}}%
    \put(0.24743777,0.0966325){\color[rgb]{0,0,0}\makebox(0,0)[t]{\lineheight{1.25}\smash{\begin{tabular}[t]{c}\sf Matrix product states\end{tabular}}}}%
    \put(0,0){\includegraphics[width=\unitlength,page=5]{venn.pdf}}%
    \put(0.12298682,0.05124451){\color[rgb]{0,0,0}\makebox(0,0)[t]{\lineheight{1.25}\smash{\begin{tabular}[t]{c}$O(\log n)$\end{tabular}}}}%
    \put(0,0){\includegraphics[width=\unitlength,page=6]{venn.pdf}}%
    \put(0.06175832,0.12960871){\color[rgb]{0,0,0}\rotatebox{90}{\makebox(0,0)[t]{\lineheight{1.25}\smash{\begin{tabular}[t]{c}$O(\log n)$\end{tabular}}}}}%
    \put(0,0){\includegraphics[width=\unitlength,page=7]{venn.pdf}}%
    \put(0.02634152,0.23426112){\color[rgb]{0,0,0}\rotatebox{90}{\makebox(0,0)[t]{\lineheight{1.25}\smash{\begin{tabular}[t]{c}\large \sf \textbf{Entanglement} $S_1$\end{tabular}}}}}%
    \put(0,0){\includegraphics[width=\unitlength,page=8]{venn.pdf}}%
    \put(0.25475842,0.01756955){\color[rgb]{0,0,0}\makebox(0,0)[t]{\lineheight{1.25}\smash{\begin{tabular}[t]{c}\large \sf \textbf{Magic} $\nu$\end{tabular}}}}%
    \put(0,0){\includegraphics[width=\unitlength,page=9]{venn.pdf}}%
    \put(0.20058565,0.27672035){\color[rgb]{0,0,0}\makebox(0,0)[t]{\lineheight{1.25}\smash{\begin{tabular}[t]{c}\Large \sf ED\end{tabular}}}}%
    \put(0,0){\includegraphics[width=\unitlength,page=10]{venn.pdf}}%
    \put(0.31039531,0.17715959){\color[rgb]{0,0,0}\makebox(0,0)[t]{\lineheight{1.25}\smash{\begin{tabular}[t]{c}\Large \sf MD\end{tabular}}}}%
    \put(0,0){\includegraphics[width=\unitlength,page=11]{venn.pdf}}%
    \put(0.73499268,0.43191801){\color[rgb]{0,0,0}\makebox(0,0)[t]{\lineheight{1.25}\smash{\begin{tabular}[t]{c}\large \sf (b) \ ED and MD states in\end{tabular}}}}%
    \put(0,0){\includegraphics[width=\unitlength,page=12]{venn.pdf}}%
    \put(0.74963397,0.30307467){\color[rgb]{0,0,0}\makebox(0,0)[t]{\lineheight{1.25}\smash{\begin{tabular}[t]{c}\Large $\mathscr{H}$\end{tabular}}}}%
    \put(0,0){\includegraphics[width=\unitlength,page=13]{venn.pdf}}%
    \put(0.74963397,0.40556369){\color[rgb]{0,0,0}\makebox(0,0)[t]{\lineheight{1.25}\smash{\begin{tabular}[t]{c}\large \sf Hilbert space $\mathscr{H}$\end{tabular}}}}%
    \put(0,0){\includegraphics[width=\unitlength,page=14]{venn.pdf}}%
    \put(0.81620665,0.231373){\color[rgb]{0,0,0}\rotatebox{-15}{\makebox(0,0)[t]{\lineheight{1.25}\smash{\begin{tabular}[t]{c}$\nu=o(n)$\end{tabular}}}}}%
  \end{picture}%
\endgroup%

%% file: figures/phase-transition.pdf_tex
\begingroup%
  \makeatletter%
  \providecommand\color[2][]{%
    \errmessage{(Inkscape) Color is used for the text in Inkscape, but the package 'color.sty' is not loaded}%
    \renewcommand\color[2][]{}%
  }%
  \providecommand\transparent[1]{%
    \errmessage{(Inkscape) Transparency is used (non-zero) for the text in Inkscape, but the package 'transparent.sty' is not loaded}%
    \renewcommand\transparent[1]{}%
  }%
  \providecommand\rotatebox[2]{#2}%
  \newcommand*\fsize{\dimexpr\f@size pt\relax}%
  \newcommand*\lineheight[1]{\fontsize{\fsize}{#1\fsize}\selectfont}%
  \ifx\svgwidth\undefined%
    \setlength{\unitlength}{345bp}%
    \ifx\svgscale\undefined%
      \relax%
    \else%
      \setlength{\unitlength}{\unitlength * \real{\svgscale}}%
    \fi%
  \else%
    \setlength{\unitlength}{\svgwidth}%
  \fi%
  \global\let\svgwidth\undefined%
  \global\let\svgscale\undefined%
  \makeatother%
  \begin{picture}(1,0.6173913)%
    \lineheight{1}%
    \setlength\tabcolsep{0pt}%
    \put(0,0){\includegraphics[width=\unitlength,page=1]{phase-transition.pdf}}%
    \put(0.16086957,0.36304348){\color[rgb]{0,0,0}\makebox(0,0)[t]{\lineheight{1.25}\smash{\begin{tabular}[t]{c}\Large $\mathcal{C}_1$\end{tabular}}}}%
    \put(0,0){\includegraphics[width=\unitlength,page=2]{phase-transition.pdf}}%
    \put(0.33913043,0.02173913){\color[rgb]{0,0,0}\makebox(0,0)[t]{\lineheight{1.25}\smash{\begin{tabular}[t]{c}\sf $d$-many layers\end{tabular}}}}%
    \put(0,0){\includegraphics[width=\unitlength,page=3]{phase-transition.pdf}}%
    \put(0.42608696,0.36304348){\color[rgb]{0,0,0}\makebox(0,0)[t]{\lineheight{1.25}\smash{\begin{tabular}[t]{c}\Large $\mathcal{C}_2$\end{tabular}}}}%
    \put(0,0){\includegraphics[width=\unitlength,page=4]{phase-transition.pdf}}%
    \put(0.87391304,0.4173913){\color[rgb]{0,0,0}\makebox(0,0)[t]{\lineheight{1.25}\smash{\begin{tabular}[t]{c}\sf Clifford\end{tabular}}}}%
    \put(0,0){\includegraphics[width=\unitlength,page=5]{phase-transition.pdf}}%
    \put(0.87826087,0.35){\color[rgb]{0,0,0}\makebox(0,0)[t]{\lineheight{1.25}\smash{\begin{tabular}[t]{c}$R_z(\theta)$\end{tabular}}}}%
    \put(0,0){\includegraphics[width=\unitlength,page=6]{phase-transition.pdf}}%
    \put(0.87826087,0.28695652){\color[rgb]{0,0,0}\makebox(0,0)[t]{\lineheight{1.25}\smash{\begin{tabular}[t]{c}\sf Measurement\end{tabular}}}}%
    \put(0,0){\includegraphics[width=\unitlength,page=7]{phase-transition.pdf}}%
  \end{picture}%
\endgroup%